\documentclass[10pt]{article} % For LaTeX2e
\usepackage{tmlr}
\usepackage{amsmath,amsfonts,bm}

\def\eqref#1{equation~\ref{#1}}
\def\1{\bm{1}}

\DeclareMathAlphabet{\mathsfit}{\encodingdefault}{\sfdefault}{m}{sl}
\SetMathAlphabet{\mathsfit}{bold}{\encodingdefault}{\sfdefault}{bx}{n}

\usepackage{amsthm}  % loaded safely by tmlr

\newtheorem{theorem}{Theorem}
\newtheorem{lemma}[theorem]{Lemma}

\newtheorem{corollary}[theorem]{Corollary}
\theoremstyle{definition}
\newtheorem{definition}[theorem]{Definition}
\theoremstyle{remark}
\newtheorem{remark}[theorem]{Remark}
\theoremstyle{assumption}
\newtheorem{assumption}[theorem]{Assumption}

\usepackage{algorithm}
\usepackage{algorithmicx}
\usepackage{algpseudocode}
\usepackage{float}

\usepackage{amssymb}
\usepackage{amsmath}

\title{Regret-Guaranteed Safe Switching: LQR Setting
with Unknown Dynamics}

\author{
      \name Jafar Abbaszadeh Chekan\thanks{Corresponding author.} 
      \email jafar2@illinois.edu \\
      \addr Coordinated Science Lab, UIUC\\
      \AND
      \name S. Rasoul Etesami 
      \email etesami1@illinois.edu \\
      \addr Coordinated Science Lab, UIUC\\
      \AND
      \name Cedric Langbort 
      \email langbort@illinois.edu\\
      \addr Coordinated Science Lab, UIUC \\
}

\begin{document}

\maketitle

\begin{abstract}
We consider learning-based control in a switched Linear Quadratic Regulator (LQR) setting, where the parameters associated with each mode are \emph{a priori} unknown. The next mode to be activated is revealed online only at the time of switching. The objective is to determine both the switching times and the control gains for each mode such that (1) the norm of the system state remains bounded according to a prescribed criterion, and (2) the accumulated cost is minimized. To formalize the state-norm requirement, we introduce the notion of $(\alpha,\beta)$-controllability for given parameters $\alpha$ and $\beta$. A naive strategy would be to switch rapidly according to the revealed modes, applying an arbitrary stabilizing policy and remaining in each mode for at most one step. However, such a strategy may violate the $(\alpha,\beta)$-controllability condition. Consequently, a more careful design of both the switching times and the feedback controllers is necessary to satisfy the two objectives simultaneously. We first study the problem in a known model setting and show that, under the switching mechanism described above and under the assumption that each mode is visited infinitely often, the strategy that minimizes the average expected cost consists of applying, in each mode, the feedback gain obtained from the solution of the discrete algebraic Riccati equation (DARE), while selecting dwell times that sufficiently satisfy the controllability condition.  We refer to this strategy as the \emph{benchmark} policy. Next, we propose an algorithm for the unknown-model setting that minimizes the \emph{regret}, defined as the difference between the cumulative cost incurred by the online algorithm and that of the offline benchmark. By accurately estimating dwell-time errors, our method achieves an expected regret of $\mathcal{O}(|\mathcal{M}|^{1/4} n_s^{3/4} + n_m)$ up to logarithmic factors, where $n_s$ denotes the number of switches, $|\mathcal{M}|$ is the number of modes, and $n_m$ is the number of \emph{malignant} switches, defined as those that tend to induce more costly switches in the future. Our approach leverages the primal and dual semidefinite programming (SDP) formulations of the linear quadratic (LQ) control problem, relaxed through confidence ellipsoids constructed around the unknown system parameters. 
\end{abstract}

\section{Introduction}\label{Introduction}

Significant progress has been made in recent years in the development of learning-based and data-driven control techniques for both linear and nonlinear systems with unknown parameters~\cite{boffi2021regret, kakade2020information, chen2021black, dean2020sample, cohen2018online, mania2019certainty, cohen2019learning}. However, one class of systems for which learning-based control remains relatively underdeveloped is \emph{switched systems}—systems composed of multiple continuous-time modes, with transitions governed by a discrete switching protocol~\cite{zhu2015optimal, zhao2008stability, lin2009stability, liberzon1999stability}. The challenge in such systems arises from the need to handle both continuous dynamics and discrete mode transitions, making the learning problem considerably more complex.

An important subclass of switched systems operates under an \emph{externally forced switching} (EFS) protocol, in which the switching times, the sequence of modes, or both are determined by an external agent~\cite{zhu2015optimal}. In this work, we focus on a setting in which only the \emph{next mode} is revealed at the time of switching, and the system must determine \emph{when} to execute the switch as well as \emph{how} to control the system in the current mode. The former decision is particularly critical, since rapid switching is known to induce instability and lead to \emph{state explosion}, even when the individual closed-loop subsystems are stable. Specifically, we study this problem in the context of the switched linear quadratic regulator (LQR) setting, where prior knowledge of the system dynamics mode parameters is unavailable, while the mode cost parameters are known. Our objective is to develop an online control strategy that tracks the externally prescribed mode sequence while simultaneously: (1) stabilizing the system within each mode, (2) ensuring that the norm of the system state remains bounded—an objective we refer to as the \emph{safety of switching}, and (3) minimizing the cumulative control cost in a \emph{regret} sense. The safety requirement prevents frequent or impulsive switching, which may appear optimal from a cost perspective but can destabilize the system.

As a motivating application, one can consider a compelling problem in healthcare involving the design of an efficient treatment strategy for a disease with multiple available therapies (see, e.g., \cite{liang2025randomization, joshi2025ai} for some related intervention applications in healthcare). The patient’s physiological response can exhibit distinct dynamic behaviors depending on the treatment administered, which is generally unknown. Consequently, the system can be modeled as having different dynamics under different treatments. Each treatment aims to stabilize the patient’s health state while incurring a cost that reflects both its effectiveness and its associated resources or financial burden. Since the optimal sequence of treatments is not known a priori, the physician selects each subsequent treatment based on the patient’s observed response to the ongoing therapy. However, determining when to switch to a new treatment is critical, as medications may have side effects, and switching too early without allowing sufficient response time may lead to undesirable or unstable physiological conditions (e.g., elevated blood pressure). Therefore, the treatment strategy must ensure that such transitions are performed safely to avoid adverse events. In practice, the decision parameters involve determining both the appropriate dosage of the current medication and the timing for switching to the next treatment, as prescribed by the physician. It is also noteworthy that, in this setting, the cost function is known, as it is determined by the physician; however, the underlying dynamics corresponding to different treatments are unknown, which constitutes the main challenge in designing such a treatment strategy.

{We refer to the setting considered in this paper as ``unknown modes–unknown sequence,'' denoted by $\bar{M}\bar{S}$, in which the switching sequence is not known a priori, the system dynamics are unknown, and only noisy state measurements are available. To motivate our approach, we first introduce a related case that serves as a benchmark for our proposed strategy. We refer to this case as ``known modes–unknown sequence,'' denoted by $M\bar{S}$. This setting follows the same mode information disclosure protocol as our target setting and provides a meaningful foundation for defining the notion of regret. We show that even in the $M\bar{S}$ case, achieving performance guarantees requires more than just access to the next mode. This insight leads us to a practical baseline approach using a Discrete Algebric Riccati Equation (DARE) based control law with a mode-dependent dwell-time policy. We then extend this approach to the $\bar{M}\bar{S}$ setting by introducing a semidefinite programming (SDP)-based learning algorithm that jointly learns the feedback gain and minimum dwell time using only state observations. Our method achieves low regret relative to the $M\bar{S}$ benchmark and enforces switching safety entirely online.}

%\textcolor{red}{The paragraph above is a modification of the following paragraph, reflecting our agreement to exclude the $\bar{M}\bar{S}$ case.
%}

%To address these challenges, we introduce the ``unknown modes–unknown sequence'' setting, denoted $\bar{M}\bar{S}$, where the system dynamics are unknown and only noisy state measurements are available. We first analyze two related cases: (i) the fully known setting ($MS$), where both the modes and switching sequence are known, and (ii) a partially known setting ($M\bar{S}$), where modes are known but the sequence is not. The $M\bar{S}$ setting, which follows the same information disclosure protocol as our target setting, provides a meaningful benchmark for our algorithmic development. We show that even in the $M\bar{S}$ case, achieving performance guarantees requires more than just access to the next mode. This insight leads us to a practical baseline approach using a DARE-based control law with a mode-dependent dwell-time policy. We then extend this approach to the $\bar{M}\bar{S}$ setting by introducing a semidefinite programming (SDP)-based learning algorithm that jointly learns the feedback gain and minimum dwell time using only state observations. Our method achieves low regret relative to the $M\bar{S}$ benchmark and enforces switching safety entirely online.

\subsection{Related Work}

The problem of state explosion under rapid switching is well known in the switched systems literature~\cite{liberzon2003switching, hespanha1999stability, zhao2011stability}. A standard approach to mitigate this risk is to enforce a \emph{minimum dwell time}, i.e., a minimum duration that the system must remain in a mode before switching~\cite{lin2009stability, liberzon1999stability}. However, computing an appropriate dwell time typically requires full knowledge of the system dynamics. In the absence of such knowledge, dwell-time design must be carried out dynamically and based only on data. This requirement introduces a new layer of complexity to control design for switched systems with unknown modes and switching sequences.

Our setting departs from prior work in several ways: (i) the nature and extent of prior knowledge about the system modes and switching signals, (ii) whether the timing of mode switches is treated as a controllable design variable, and (iii) the performance criteria—specifically, whether the goal is merely stabilization or additional guarantees such as regret minimization. For example, the work of~\cite{kenanian2019data} addresses probabilistic stability of switched systems using random observations, without requiring knowledge of the mode dynamics or switching rule. In~\cite{rotulo2022online}, the authors propose a data-driven approach for unknown discrete-time linear systems with arbitrary switching. Their method adapts to changes in actuation modes using data alone, without explicit system identification. Similarly,~\cite{dai2018moments, dai2022convex} propose stabilizing controllers for linear switched systems using data obtained from persistently exciting experiments. These methods guarantee stability for arbitrary switching but do not provide regret performance guarantee.

In contrast, ~\cite{du2022data} address a Markov Jump System (MJS) in an LQ setting, jointly learning mode dynamics and the transition matrix while minimizing regret. Related efforts in MJS learning and control also appear in~\cite{sayedana2023relative, sayedana2024strong}. Closer to our setting,~\cite{li2023online} consider regret minimization while switching between different policies. Their algorithm selects policies from a finite set using bandit-based techniques and applies a conservative dwell-time policy to ensure stability. This conservative design limits the achievable performance.

In our prior work~\cite{chekan2024learn}, we addressed switching in over-actuated systems using a projection-based learning method but did not handle regret guarantees. The current work builds on our recently introduced SDP-based learning algorithm, ARSLO~\cite{chekan2024any}, which enables online control with provable regret in unknown LQR settings.

\subsection{Contributions}

Our main contributions are as follows:

\begin{itemize}
    \item We propose a computationally efficient algorithm to jointly design feedback controllers and estimate mode-dependent dwell times in switched systems with unknown dynamics.
    
    \item We provide upper bounds on the error between the estimated dwell time and the (unknown) true dwell time, showing that our estimates are statistically reliable.
    
    \item We prove that our proposed algorithm achieves an expected regret of $\mathcal{O}(|\mathcal{M}|^{1/4} n_s^{3/4} + n_m)$ relative to the benchmark case where all mode parameters are known, where $n_s$ denotes the number of switches, $|\mathcal{M}|$ is the number of modes, and $n_m$ is the number of malignant switches.
\end{itemize}

At a technical level, our approach is rooted in model-based reinforcement learning techniques that incorporate system identification via confidence sets. In particular, we leverage our previously proposed ARSLO algorithm~\cite{chekan2024any}, extending it to the switched system setting with unknown modes and switching sequences.

\subsection{Organization and Notations}

The remainder of the paper is organized as follows. Section~\ref{sec:probStat} presents the problem statement. In Section~\ref{eq:knownsetting}, we discuss the solution for the case where all system parameters are known. In Section~\ref{sec:prelim}, we present preliminary results on confidence set construction (system identification) and a relaxation of the standard LQR-SDP formulation for joint control and dwell-time design, which are then used to develop our main online algorithm. Section~\ref{sec:solutionP} presents the algorithm and its theoretical foundation, including data-driven dwell-time design. Section~\ref{sec:theoreticalg} provides the theoretical guarantees, including a regret bound analysis of the proposed algorithm. Finally, Section~\ref{sec:conclusion} concludes the paper and outlines directions for future research. Omitted proofs and additional results are provided in Appendix~\ref{partD}.

\noindent
{\bf Notations:} The trace norm of a matrix $M$ is defined as $\|M\|_* = \operatorname{Tr}(\sqrt{M^\top M})$. $\|M\|_F=\sqrt{\operatorname{tr} (M^\top M)}$ is the Frobenius norm and $\|M\|$ is the spectral (operator) norm. Additionally, $\underline{\lambda}(M)$ and $\overline{\lambda}(M)$ denote the minimum and maximum eigenvalues of $M$, respectively. For a set $\mathcal{S}$, $|\mathcal{S}|$ denotes its cardinality. We use $A \bullet B$ to represent the element-wise dot product of matrices, defined as $A \bullet B = \operatorname{Tr}(A^\top B)$. Finally, we write \(f \lesssim g\) to denote that \(f(x) \leq C\, g(x)\) for some universal constant \(C > 0\).

\section{Problem Formulation} \label{sec:probStat}

Consider a time-invariant switched LQR system defined by:
\begin{align}
x_{t+1} &= A^{\sigma(t)}_{*} x_t + B^{\sigma(t)}_{*} u_t + \omega_{t+1}, \label{eq:dyn_atttt} \\
c^{\sigma(t)}_{t} &= x_t^\top Q^{\sigma(t)} x_t + u_t^\top R^{\sigma(t)} u_t, \label{eq:obs}
\end{align}
where the switching signal $\sigma: \{0\} \cup \mathbb{N} \rightarrow \mathcal{M}$ is a right-continuous, piecewise-constant function that specifies which mode is active at each time step $t$, where $\mathcal{M}$ is a finite index set representing all possible modes. The state and control input satisfy $x_t\in\mathbb{R}^{d_x^{\sigma(t)}}$ and $u_t\in\mathbb{R}^{d_u^{\sigma(t)}}$, respectively. For notational simplicity, we assume throughout that all modes share a common state dimension $d_x$, while the input dimension $d_u^{\sigma(t)}$ may depend on the active mode; the extension to mode-dependent state dimensions is straightforward.
 For each mode $i \in \mathcal{M}$, the matrices $A^i_*$ and $B^i_*$ represent the true (but initially unknown) system dynamics. The matrices $Q^i$ and $R^i$ are known, with $Q^i$ positive semi-definite and $R^i$ positive definite, representing the state and control cost weights for mode $i$, respectively. Moreover, $\omega_{t+1}$ denotes the process noise, which is assumed to satisfy the following standard assumption \cite{abbasi2011regret, cohen2019learning, lale2022reinforcement}: 
\begin{assumption}
\label{Assumption 1}
There exists a filtration $\mathcal{F}_{t}$ such that

$(1.1)$ $\omega_{t+1}$ is a martingale difference, i.e., $\mathbb{E}[\omega_{t+1}|\;\mathcal{F}_{t}]=0.$

$(1.2)$ $\mathbb{E}[\omega_{t+1}\omega_{t+1}^\top|\;\mathcal{F}_{t}]={\sigma}_{\omega}^2I_{n}=:W$ for some ${\sigma}_{\omega}^2>0.$

% $(1.3)$ $\omega_{t+1}$ are component-wise sub-Gaussian, i.e., there exists $\sigma_{\omega}\!>\!0$ such that for any $\gamma \in \mathbb{R}$ and $j=1,\ldots,n,$
% \begin{align*}
% \mathbb{E}[e^{\gamma(\omega_{t+1})_{j}}|\;\mathcal{F}_{t}]\leq e^{\gamma^2\sigma_{\omega}^2/2}.
% \end{align*}
\end{assumption}

Let $\mathcal{I} = \{i_0, i_1, \ldots, i_n\}$ denote the sequence of modes through which the system evolves. This sequence is \emph{not} known in advance and is revealed sequentially; that is, only the next mode to switch to, denoted by $i_{k+1} \in \mathcal{M}$, becomes known after the current switch to mode $i_k \in \mathcal{M}$ has occurred. In particular, the termination of the sequence is uncertain and is announced only after the final mode $i_n$ is revealed. In conclusion, after each switch, either the next mode or the termination of the sequence is announced. We refer to the time interval between two consecutive switches as an \emph{epoch}. The main objective in this work is to design an algorithm that ensures actuation according to a sequence of switches $\mathcal{I}$ while minimizing the cumulative cost and maintaining control over the expected growth of the state norm.

\begin{definition}
For a sequence of modes $\mathcal{I} = \{i_0, \ldots, i_n\}$ through which the system evolves, we let $t_{k+1}$ be the time instance at which the system switches from mode $i_k$ to mode $i_{k+1}$. Moreover, we let $\tau_{k,k+1}\!=t_{k+1}-t_k$ and $t^{-}_{k+1}=t_{k+1}-1$. Since the switching signals are right-continuous, we have $\sigma(t^-_{k+1})=i_k$ and $\sigma(t_{k+1})=i_{k+1}$.     
\end{definition}

\begin{definition} \label{def:underControl}
Let $0 < \bar{\alpha} < 1$ and $\bar{\beta} > 0$. A sequence of switching modes $\mathcal{I} = \{i_0, i_1, \ldots, i_n\}$ is called \emph{$(\bar{\alpha}, \bar{\beta})$-controllable} if, for any $k$, the expected state norm growth satisfies\footnote{Here, $\bar{\alpha}$ is a user-defined parameter, while $\bar{\beta}$ depends on the specific policy class used for control design and on the cost function associated with each mode (see, Lemma \ref{lem:betaDef}).}
\begin{align}
&\mathbb{E}[x^{\top}_{t_{k+1}} x_{t_{k+1}}| \mathcal{F}_{t^-_{k}}] \leq \bar{\alpha}\; \mathbb{E}[x_{t_k}^\top x_{t_k}| \mathcal{F}_{t_k^-}] + \bar{\beta}\;\sigma_{\omega}^2 .\label{eq:stategrowthpp}
\end{align}
\end{definition}

For any $i\in \mathcal{M}$, let us introduce the notation ${\Theta_*^{i}} = (A_*^i, B_*^i)^\top$. Given an active mode $\sigma(t) = i$ at time $t$, we can equivalently express~(\ref{eq:dyn_atttt}) as
\begin{align}
	x_{t+1} = {\Theta_*^{i}}^\top z_t + \omega_{t+1}, 
	\quad z_t = \begin{pmatrix} x_t \\ u_t \end{pmatrix}, 
	\label{eq:dynam_by_theta}
\end{align}
This compact representation will be used frequently throughout the paper to describe the system dynamics.

\begin{definition}[Strong stability] \label{def:sequentially}
Consider the linear time-invariant system in~(\ref{eq:dyn_atttt}).  
The closed-loop matrix $A_* + B_* K$ is said to be \emph{$(\kappa, \gamma)$–strongly stable} for some $\kappa > 0$ and $0 < \gamma < 1$ if there exist matrices $H \succ I$ and $L$ such that
\[
A_* + B_* K = H L H^{-1},
\]
and the following conditions hold:
\begin{enumerate}
    \item $\|L\| \leq 1 - \gamma$ and $\|K\| \leq \kappa$;
    \item $\|H\|\|H^{-1}\|\leq \kappa$.
\end{enumerate}
Furthermore, we say control gains $K$ is $(\kappa, \gamma)-$strongly stabilizing for a plant $(A_*, B_*)$ if $A_*+B_*K$ is $(\kappa, \gamma)-$strongly stable.
\end{definition}
 For control design purposes, at each mode $i\in \mathcal{M}$, we restrict the control policy to belong to the compact set $\mathcal{S}(\Theta_*^i)$ specified by constants $\kappa_c^i > 0$ and $0 < \gamma_c^i < 1$, which define stabilizing policies $K_i$ for the system parameterized by $\Theta_*^i$ as follows:
\begin{align}\label{eq:policyClass}
\mathcal{S}(\Theta_*^i) = \big\{ K_i \in \mathbb{R}^{d_u^i \times d_x^i} \mid A_*^i + B_*^i K_i \quad \textit{is}\quad (\kappa_c^i,\gamma_c^i)-\textit{strongly stable}  \big\}. 
\end{align}
In other words, when the system is in mode $\sigma(t) = i$ at time $t$, we assume the control input is given by the state feedback law $u_t = K_i x_t$, where $K_i \in \mathcal{S}(\Theta_*^i)$. Given the policy class defined in~(\ref{eq:policyClass}), Lemma \ref{lem:betaDef} provides an explicit characterization of the parameter $\bar{\beta}$ introduced in Definition~\ref{def:underControl}.

To lay the foundation for developing an algorithm that minimizes the cumulative cost of the switching dynamics while maintaining control over the expected growth of the state norm under the general Unknown-Modes and Unknown-Sequence ($\bar{M}\bar{S}$) setting, it is crucial to first thoroughly examine the problem in the simpler Known-Modes and Unknown-Sequence ($M\bar{S}$) setting. This analysis not only deepens our understanding of the problem but also establishes a benchmark for defining regret, whose upper bound serves as a metric to evaluate the performance of our proposed online algorithm. 

We define the \emph{regret} of our proposed algorithm for the $\bar{M}\bar{S}$ setting with respect to a benchmark algorithm for the $M\bar{S}$ setting. More precisely, let the cumulative cost of an online algorithm under the $\bar{M}\bar{S}$ setting be denoted by $\mathcal{C}_{\mathcal{I}}^{\bar{M}\bar{S}}$, and let $\mathcal{C}_{\mathcal{I}}^{M\bar{S}}$ denote the cumulative cost of a benchmark algorithm (to be defined shortly) under the $M\bar{S}$ setting. Then, the \emph{regret} of the online algorithm in the $\bar{M}\bar{S}$ setting is defined as
\begin{align}
R(\mathcal{I}) = \mathcal{C}_{\mathcal{I}}^{\bar{M}\bar{S}}- \mathcal{C}_{\mathcal{I}}^{M\bar{S}}.
\label{eq:regdefn}
\end{align}
In the following sections, we formally present the strategies of the online algorithm and the benchmark, together with their underlying rationale and technical development. The overarching goal of this paper is to design an online algorithm that minimizes the regret defined in~(\ref{eq:regdefn}), while adapting to the switching sequence revealed sequentially and ensuring that the resulting sequence remains $(\bar{\alpha}, \bar{\beta})$-controllable in the sense of Definition~\ref{def:underControl}.

\begin{remark}
To minimize accumulated cost, fast switching may seem rational. However, even when the mode characteristics are known, it is well established that rapid switching between modes can lead to an explosion in the system state norm, jeopardizing global stability—even if each mode is stable individually \cite{zhu2015optimal, zhao2008stability, lin2009stability, liberzon1999stability}. Accordingly, fast switching can violate state norm growth control as defined in Definition~\ref{def:underControl}. This issue can be mitigated by ensuring that the system remains in each mode for a computable minimum duration. Hence, the key challenge in addressing this problem lies in the simultaneous design of control policies to be applied during different epochs and determining the appropriate switching times. 
\end{remark}

% \begin{align*}
%     C^{i}_{\mathcal{I}}(t) = \sum_{k=0}^t c\big(\sigma_{i}(k), u_k(i)\big), \quad i \in \{m\bar{s}, \; \bar{m}\bar{s}\},
% \end{align*}
% where the instantaneous cost depends on the control input $u_k(i)$ and the switching signal $\sigma_i(k)$, which follows the revealed mode sequence:
% \begin{align*}
%     c\big(\sigma_{i}(k), u_k(i)\big) = x^{\top}_k(i) Q^{\sigma_i(k)} x_k(i) + u^{\top}_k(i) R^{\sigma_i(k)} u_k(i).
% \end{align*}
% The state $x_k(i)$ for the corresponding setting evolves according to
% \begin{align*}
%     x_{k+1}(i) = {\Theta_*^{\sigma_i(k)}}^\top z_k(i) + \omega_{k+1}, 
%     \quad z_k(i) = \begin{pmatrix} x_k(i) \\ u_k(i) \end{pmatrix}.
% \end{align*}
% We note that $\sigma_i(k)$ for both settings, $m\bar{s}$ and $\bar{m}\bar{s}$, follows the revealed switching sequences; the difference lies in the timing of the switches, which may vary between the two strategies.

% Note that the notation specifying the setting is introduced here primarily to simplify the representation of the regret definition. In the remainder of the paper, we will adhere to the previously established notation, with the specific setting under analysis being clear from the context.

\section{Known-Modes and Unknown-Sequence ($M\bar{S}$) Setup} \label{eq:knownsetting}

In this section, we develop a control strategy for the case in which the system mode dynamics parameters are known, while the switching sequence follows the revealed-sequence mechanism described earlier. We refer to this benchmark configuration as the $M\bar{S}$ setting.

Before presenting the corresponding strategy, we first analyze the problem under the assumption that the switching sequence $\mathcal{I}=\{i_0,i_1,\ldots,i_n\}$ is known \emph{a priori}. The decision variables for this case are defined as follows: the set $\mathcal{T}_0 = \{t_0, t_1, \ldots, t_n\}$, representing the switching times, and the set of control inputs $\mathcal{K}_0 = \{K_t^{\sigma(t)}\}_{t=0}^{T}$, representing the linear feedback control gain applied at each time step. The problem can then be formulated as follows:

\textbf{Problem $0$}
\vspace{-0.2cm}
\begin{align}
&(\mathcal{T}^*_0, \mathcal{K}^*_0\operatorname*)\in{\arg\min}_{\mathcal{K}_0, \mathcal{T}_0}  \sum _{t=0}^{T} c_{t}^{\sigma(t)}\\
  &  x_{t+1} =A^{\sigma(t)}_{*}x_{t} + B^{\sigma(t) }_{*}u_{t}+\omega_{t+1} \label{eq:DynModeg}\\
  & u_t=K_t^{\sigma(t)}x_t \label{eq:feedbackGain}\\
 &  \sigma(t)=i_k\; \text{for}\; t_k\leq t\leq t^-_{k+1}, \quad \sigma (T)=i_n \label{eq:indexSwitchId}\\
 & \mathbb{E}[x_{t_{k+1}}^\top  x_{t_{k+1}}| \mathcal{F}_{t^-_{k+1}}] \leq \bar{\alpha}\; \mathbb{E}[x_{t_k}^\top x_{t_k}| \mathcal{F}_{t_{k}^-}] + \bar{\beta}\; \sigma_{\omega}^2 \label{eq:StateBoundedness}\\
 &\tau_{k,k+1}:=t_{k+1}-t_k\geq 1. \label{eq:condSta}
\end{align}

The constraint (\ref{eq:DynModeg}) corresponds to the continuous component of the dynamics, describing the system's evolution within each mode. The constraint (\ref{eq:feedbackGain}) expresses that the control input is a static state feedback of the system.
The constraint (\ref{eq:indexSwitchId}) is a well-posedness condition, ensuring that the discrete part of the dynamics follows the revealed sequence and thus determines which mode is active over a given interval. The constraint (\ref{eq:StateBoundedness}) is imposed to guarantee that the sequence $\mathcal{I}$ is $(\bar{\alpha},\bar{\beta})$-controllable (Definition \ref{def:underControl}). This serves as a safety constraint on mode switching, regulating the growth of the state norm that may result from rapid switching. Finally, the constraint (\ref{eq:condSta}) ensures that the system remains in each mode for at least one time step.

\begin{remark} \label{rem:sumswi}
Let $\sigma(t)$ be defined as in~(\ref{eq:indexSwitchId}) for the switching mode sequence $\mathcal{I}$. Then, the objective cost in Problem 0 can be expressed equivalently in terms of switching times $t_k$ as
\begin{align}\nonumber
\sum_{t=0}^{T} c_t^{\sigma(t)}
=
\sum_{k=0}^{n} \sum_{j=t_k}^{t_{k+1}-1} c_j^{i_k}.
\end{align}
\end{remark}

We note that solving Problem~0, which assumes full knowledge of the switching sequence $\mathcal{I}$, remains computationally challenging for several reasons: (i) The constraint~(\ref{eq:StateBoundedness}), i.e., the boundedness of the expected state norm as defined in Definition~\ref{def:underControl}, does not explicitly depend on the design variables $\mathcal{K}_0$ and $\mathcal{T}_0$, thereby introducing an additional layer of complexity to Problem~0. (ii) The set of stabilizing policies $\mathcal{S}(\Theta_*^{i})$ is generally nonconvex, as demonstrated through a counterexample in~\cite{fazel2018global}. (iii) The problem is computationally hard due to the combinatorial nature of the switching decisions, since switches may occur in many possible combinations, rendering combinatorial optimization in control settings generally intractable. In particular,~\cite{gurvits2009np} studies the problem of determining whether polynomial-time algorithms exist for checking the stability of all convex combinations of matrices, highlighting the NP-hardness of this class of problems. 

Aside from its computational complexity, Problem~0 is a very strong benchmark for defining regret, since it assumes full access to the switching sequence $\mathcal{I}$, making it an unrealistically ideal baseline for comparison with algorithms that possess only limited information about $\mathcal{I}$. Under the switching mechanism considered here, after each switch either the next mode index or the termination of the sequence is revealed. Consequently, the problem should be solved in an epoch-by-epoch manner, since the past is fixed while the future remains entirely adversarial. As we will show later, explicitly incorporating this switching mechanism into the problem formulation helps circumvent challenge~(iii).

In this work, we assume that the control policy remains fixed within each epoch. The justification for this assumption will be discussed later. One immediate advantage of this assumption concerns the $(\bar{\alpha}, \bar{\beta})$-controllability requirement. In particular, we provide below a sufficient condition for this requirement that translates it into a lower bound on the dwell time, i.e., the minimum amount of time that must elapse before switching to the next mode. This lower bound depends on the fixed policy applied during the epoch operating in mode $i_k$ and the subsequent mode $i_{k+1}$ to which the system switches. Importantly, this condition is useful for mitigating complexity~(i), since it depends explicitly on the decision variables $t_{k+1}$ and the feedback policies. The following theorem summarizes this result.

\begin{theorem} \label{thm:alfabet_K_S}
    Let $t_k$ and $t_{k+1}$ denote the switching times corresponding to two consecutive modes $i_k$ and $i_{k+1}$ in the switching sequence $\mathcal{I}$. Furthermore, let the policies applied during the corresponding epochs be $K^{i_k}$ and $K^{i_{k+1}}$, respectively. Then, a sufficient condition for~(\ref{eq:StateBoundedness}) is given as follows:
      \begin{align}
        &\tau_{k,k+1}=\bigg\{1, -\frac{\ln \bar{\rho}(K^{i_k})+\ln \bar{\mathcal{X}}(K^{i_{k+1}})-\ln \mathcal{\bar{\alpha}}}{\ln \big(1-\bar{\eta} \big(K^{i_k})\big)} \bigg\} \label{eq:dwellFix0}
    \end{align}
  where 
 \begin{align}
	&\nonumber \bar{\eta} \big(K^{i_k}) := \frac{\underline{\lambda}\big(H_{K^{i_k}}\big)}{\overline{\lambda}\big(P_{K^{i_k}}\big)}, \quad \bar{\rho}(K^{i_k}):=\frac{\overline{\lambda}\big(P_{K^{i_{k}}}\big)}{\underline{\lambda}\big(P_{K^{i_k}})},\\
& \bar{\mathcal{X}}(K^{i_{k+1}}):= \frac{\overline{\lambda}\big(P_{K^{i_{k+1}}}\big)}{\underline{\lambda}\big(P_{K^{i_{k+1}}})}, \quad H_{K^{i_k}}=Q^{i_k}+K^{} R^{i_k} K^{i_k},	\label{eq:operatorsDef}
\end{align}  
where $P_{K^j}$ is a unique solution of 
 \begin{align}
     P_{K^j}=Q^j+{K^j}^\top R^jK^j+(A^j_*+B^j_*K^j)^\top P_{K^j}(A^j_*+B^j_*K^j). \label{eq:lyapunovEq}
 \end{align}
\end{theorem}

Note that~(\ref{eq:lyapunovEq}) is commonly referred to as the Lyapunov equation for the closed-loop system, or equivalently, the policy evaluation Riccati equation. Moreover, for any stabilizing policy $K^j$, the corresponding solution $P_{K^j}$ is known to be a positive semidefinite (PSD) matrix.

\begin{lemma} \label{lem:betaDef}
Let for any modes $i\in \mathcal{M}$, the policy $k^i\in \mathcal{S}(\Theta_*^i)$. Then $\bar{\beta}$ in (\ref{eq:stategrowthpp}) is specified as follows
    \begin{align}
        \bar{\beta}:= \frac{\alpha^{*}_1}{\alpha^{*}_0} (1+\kappa_c^{*^2})\;\frac{\kappa_c^{*^6}}{\gamma_c^{*}}, \label{eq:betaDef}
    \end{align}
    where $\alpha^*_0$ and $\alpha^*_1$ are such that $\alpha^*_0 I\preceq Q^i, R^i \preceq \alpha^*_1 I$ for all $i\in \mathcal{M}$, is an $\bar{\alpha}$-dependent parameter,  $\kappa^*_c=\max_{i\in |\mathcal{M}|}\kappa^i_c$, and $\gamma^*_c=\max_{i\in |\mathcal{M}|}\gamma^i_c$. 
\end{lemma}

% \textbf{Problem 1}
% \vspace{-0.2cm}
% \begin{align}
% &\operatorname*{min}_{t_{k+1},\; \{K_j^{\sigma(t)}\}_{j=t_k}^{t_{k+1}}}\;  \sum _{j=t_k}^{t_{k+1}-1} c_{j}^{\sigma(t)}\\
%   \nonumber  &\textbf{s.t.}\\
%   &  x_{t+1} =A^{\sigma(t)}_{*}x_{t} + B^{\sigma(t) }_{*}u_{t}+\omega_{t+1} \label{eq:DynModeg1}\\
%   & u_t=K_t^{\sigma(t)}x_t \label{eq:feedbackGain1}\\
%  &  \sigma(t)=i_k\; \text{for}\; t_k\leq t\leq t^-_{k+1}, \quad \sigma (t_{k+1})=i_{k+1} \label{eq:indexSwitchId1}\\
%  & \mathbb{E}[x_{t_{k+1}}^\top  x_{t_{k+1}}| \mathcal{F}_{t^-_{k+1}}] \leq \bar{\alpha}\; \mathbb{E}[x_{t_k}^\top x_{t_k}| \mathcal{F}_{t_{k}^-}] + \bar{\beta}\; \sigma_{\omega}^2 \label{eq:StateBoundedness1}\\
%  &\tau_{k,k+1}:=t_{k+1}-t_k\geq 1. \label{eq:condSta1}
% \end{align}

% The setting of Problem 1, owing to the specific mechanism of the switch index sequence revelation, does not suffer from combinatorial complexity. However, it still inherits the challenges (i) and (ii), namely the nonconvexity of the set of stabilizing policies $\mathcal{S}(\Theta_*^{i_k})$ and the lack of explicit dependence of the $(\bar{\alpha}, \bar{\beta})$-controllability condition on the decision variables.

The following program formulates the optimization problem for an epoch starting at time $t_k$, in which the decision variables are the switching time to the newly revealed mode $i_{k+1}$, denoted by $t_{k+1}$, and the feedback policies applied over the interval $t_k,\ldots,t_{k+1}-1$, denoted by $K^{j} \in \mathcal{S}(\Theta_*^{i_k}), j=t_k,\ldots,t_{k+1}-1$. In addition, the program incorporates the minimum dwell-time condition as a surrogate for the $(\bar{\alpha}, \bar{\beta})$-controllability.

\textbf{Problem 1}
\vspace{-0.2cm}
\begin{align}
&\operatorname*{min}_{t_{k+1},\; \{K^{i_k}, K^{i_{k+1}}\}}\;  \sum _{j=t_k}^{t_{k+1}-1} c_{j}^{\sigma(t)}\\
  &  x_{t+1} =A^{\sigma(t)}_{*}x_{t} + B^{\sigma(t) }_{*}u_{t}+\omega_{t+1} \label{eq:DynModeg2}\\
  & u_t=K^{\sigma(t)}x_t \label{eq:feedbackGain2}\\
 &  \sigma(t)=i_k\; \text{for}\; t_k\leq t\leq t^-_{k+1}, \quad \sigma (t_{k+1})=i_{k+1} \label{eq:indexSwitchId2}\\
 &t_{k+1}\geq t_k+\max \big\{1, \frac{\ln \bar{\rho}(K^{i_k}, K^{i_{k+1}})+\ln \bar{\mathcal{X}}(K^{i_k}, K^{i_{k+1}})-\ln \mathcal{\bar{\alpha}}}{\ln \big(1-\bar{\eta} \big(K^{i_k})\big)}  \big\}. \label{eq:Cont+condSta2}
\end{align}

Problem~1 no longer suffers from complication~(iii) discussed above, owing to the epoch-by-epoch optimization and planning induced by the specific switching mechanism. It also avoids complication~(i), thanks to the minimum dwell-time condition established in Theorem~\ref{thm:alfabet_K_S}. Despite this progress, an important issue remains. The formulation indicates that the optimal choice of the policy $K^{i_k}$ applied during the epoch starting in mode $i_k$ depends on the choice of the subsequent policy $K^{i_{k+1}}$. In turn, the latter policy depends on the policies associated with potentially forthcoming modes whose indices have not yet been revealed. This interdependence constitutes a fundamental source of complexity inherent to the problem and appears unavoidable under the given switch-index revealment mechanism. One possible way to mitigate this complication is to employ a fixed policy within each epoch that does not depend on the next mode and, in particular, does not rely on the policy associated with that mode. A natural and appealing candidate in this direction is the solution of the discrete-time algebraic Riccati equation (DARE). Such a choice may be interpreted as a greedy strategy, which is reasonable in the present setting. Moreover, this approach eliminates the need to solve optimization problems over the nonconvex stabilizing sets $\mathcal{S}(\Theta_*^{i})$.

Let $t_k^*$ denote the time at which system~(\ref{eq:dyn_atttt}) begins operating in mode $i_k$, and let $i_{k+1}$ denote the index of the next mode revealed by the switching mechanism. Based on the preceding reasoning, the control input applied during the corresponding epoch is given by $u_t = K(\Theta_*^{i_k})x_t$, where
\begin{align}
P(\Theta_*^{j}) &= Q^j + {A_*^j}^\top P(\Theta_*^{j}) A_*^j - {A_*^j}^\top P(\Theta_*^{j})\; B_*^{j}\; (R^j + {B_*^j}^\top P(\Theta_*^{j})\; B_*^j)^{-1} {B_*^j}^\top P(\Theta_*^{j}) A_*^j, \\
K(\Theta_*^{j}) &= -(R^j + {B^{j}_*}^\top P(\Theta_*^{j})\; B_*^j)^{-1} {B_*^j}^\top P(\Theta_*^{j}) A_*^j,
\end{align}
and the switching time decision, $t_{k+1}^*$, is computed as follows:
     \begin{align}
    t^*_{k+1}&=t^*_{k}+\tau^*_{k,\,{k+1}}\\
   \tau^*_{k,\,{k+1}}&= \max\bigg\{1, -\frac{\ln {\rho}_*^{{i_k}}+\ln \mathcal{X}_*^{i_{k+1}}-\ln \bar{\alpha}}{\ln \big(1-{\eta }_*^{i_k}\big)}\bigg\},\label{eq:tau_a_known} 
     \end{align}
     where
     \begin{align}
	\nonumber &\eta_*^{i_k} := \frac{\underline{\lambda}\big(H(\Theta_*^{i_k})\big)}{\overline{\lambda}\big(P({\Theta}^{i_k}_*)\big)}, \quad \rho_*^{i_k}:=\frac{\overline{\lambda}\big(P({\Theta}^{i_{k}}_*)\big)}{\underline{\lambda}\big(P({\Theta}^{i_k}_*)\big)}\\
\nonumber & 
\mathcal{X}_*^{i_{k+1}}:=\frac{\overline{\lambda}\big(P({\Theta}^{i_{k+1}}_*)\big)}{\underline{\lambda}\big(P({\Theta}^{i_{k+1}}_*)\big)},\; \textit{and} \quad H(\Theta_*^{i_k})=Q^{i_k}+ K^\top ({\Theta}^{i_k}_{*})R^{i_k} K({\Theta}^{i_k}_{*}).
\label{eq:HDefknown}
	\end{align}
The accumulated cost incurred under the proposed $M\bar{S}$ strategy during the $k$th epoch can be written as
\begin{align}
C_{i_k,i_{k+1}}
=
\sum_{j=t_k^*}^{t_{k+1}^*-1}
x_j^\top
\Big(
Q^{i_k}
+
K^\top(\Theta_*^{i_k})
R^{i_k}
K(\Theta_*^{i_k})
\Big)
x_j,
\end{align}
which clearly shows that the accumulated cost of an epoch depends on its duration $\tau_{k,k+1}^*$, which in turn depends not only on the current mode but also on the mode to which the system switches next. As a result, different switching realizations may lead to significantly different epoch lengths, thereby inducing additional cost and contributing to regret.

\section{Unknown-Modes and Unknown-Sequence ($\bar{M}\bar{S}$) Setup}
\label{sec:prelim}
In this section, we present an algorithm for solving the problem in the absence of knowledge of the system dynamics model. Before introducing the algorithm, we briefly review preliminary material that will be useful for its design. Although this section presents preliminaries for the standard LQR setting, we adapt the material to our problem by introducing additional notation to incorporate the mode index in the context of switched systems.

%The algorithm interacts with an oracle adapted from \cite{chekan2024any}. 

\subsection{Preliminaries for LQR Systems}\label{sec:pre}

\begin{assumption} \label{Ass_2}
We assume the existence of known constants $\alpha_0, \alpha_1, \vartheta, \nu > 0$ such that
\begin{align*}
\alpha_0 I \preceq Q, R \preceq \alpha_1 I,
\qquad
\|\Theta_*\| \leq \vartheta,
\qquad
\|J_*\| \leq \nu.
\end{align*}
Moreover, we assume the availability of an initial $(\kappa_0,\gamma_0)-$strongly stabilizing policy $K_0$.
\end{assumption}

Referring to Appendix $B_2$ of \cite{chekan2024any}, any stabilizing policy \(K\) is indeed \((\kappa, \gamma)\)-strongly stabilizing for some values of \(\kappa\) and \(\gamma\).

In the context of switched systems, the above definition is extended to all modes, where subscripts are used to denote the mode index (e.g., $\kappa_i$, $\gamma_i$).

\subsubsection*{Confidence Bound Construction}
\label{sec:ident}

Here, we outline the construction of a confidence ellipsoid around the unknown true parameters $\Theta_*$, while leaving detailed derivations to the supplementary material. Let $X_t$ and $Z_t$ denote matrices whose rows are $x_1^\top, \ldots, x_t^\top$ and $z_0^\top, \ldots, z_{t-1}^\top$, respectively. Using regularized least squares with a trace-based regularization term $\operatorname{Tr}\!\big((\Theta - \Theta_0)^\top (\Theta - \Theta_0)\big)$ for some prior $\Theta_0$, and a possibly time-varying regularization parameter $\lambda_t$, the estimator of $\Theta_*$ is given by
\begin{align}
    \hat{\Theta}_{t} = V_t^{-1}\!\left(Z_t^\top X_t + \lambda \Theta_0\right),
    \label{eq:LSE_Sol123}
\end{align}
where the covariance matrix is defined as $V_t = \lambda I + \sum_{s=0}^{t-1} z_s z_s^\top$.

For a prescribed $\epsilon^2$ satisfying $\operatorname{Tr}\!\big((\Theta_* - \Theta_0)^\top (\Theta_* - \Theta_0)\big) \le \epsilon^2$, the confidence ellipsoid ${\mathcal{C}}_t(\delta)$ is constructed as
\begin{align}
    {\mathcal{C}}_t(\delta)
    &:= \Big\{\Theta \in \mathbb{R}^{(n+m)\times n} \,\Big|\, 
    \operatorname{Tr}\!\big((\Theta - \hat{\Theta}_t)^\top V_t (\Theta - \hat{\Theta}_t)\big)
    \le r_t \Big\}, \label{eq:confSet1_tighterghfff2}\\
    r_t &= \Bigg( 
    \sigma_{\omega} \sqrt{2d_x \log \frac{n \det(V_t)}{\delta \det(\lambda I)}}
    + \sqrt{\lambda}\,\epsilon
    \Bigg)^{\!2}. 
    \label{radius_centralEl_realTime200}
\end{align}

Then, one can show that the ellipsoid (\ref{eq:confSet1_tighterghfff2}) contains $\Theta_*$ with probability at least $1 - \delta$ (see \cite{chekan2024any}). 

To adapt this framework to the context of our problem, we define $n_i(t)$ as the number of time steps during which the system has operated in mode $i$ up to time $t$. Furthermore, we extend the notation introduced above by appending a subscript $i$ to indicate the corresponding mode (e.g., $V^i_{n_i(t)}$, $\mathcal{C}^i_{n_i(t)}(\delta)$, etc.).

\subsubsection*{Learning-Based State Feedback Control Design}

\cite{cohen2019learning} employed primal and dual SDP formulations of the LQR problem and relaxed them using the confidence ellipsoid described above to propose an algorithm that designs the control policy solely from state measurements. However, their approach lacks an anytime regret guarantee and relies on an inefficient warm-up phase, rendering it unsuitable for our setting. Building on this SDP formulation, \cite{chekan2024any} later relaxed these restrictive assumptions and introduced an anytime variant that employs time-varying regularization together with carefully tuned input perturbations. In this section, we adapt their control design framework to our setting.

Given the parameter estimate in the form of confidence ellipsoid $\mathcal{C}_{n_i(t)}(\delta)$ (or simply $\mathcal{C}_{n_i(t)}$) defined by (\ref{eq:confSet1_tighterghfff2}), the relaxed primal SDP is formulated as follows:
\begin{align}\label{eq:RelaxedSDP}
\min & \qquad{\begin{pmatrix}
	Q^i & 0 \\
	0 & R^i
	\end{pmatrix}\bullet \Sigma}\cr
\textrm{s.t.} &\qquad \Sigma_{xx}\succeq {\hat{\Theta}^{i^\top}_{n_i(t)}} \Sigma\hat{\Theta}^i_{n_i(t)}+W-\mu_{n_i(t)}\big(\Sigma\bullet {{V}^{-1}_{n_i(t)}}\big)I,\cr
&\qquad\Sigma:= \begin{pmatrix}
\Sigma_{xx} & \Sigma_{xu} \\
\Sigma_{ux} & \Sigma_{uu}
\end{pmatrix}\succ 0,
\end{align}
where $\Sigma$ denotes the steady-state covariance matrix of the state-action pair $(x,u)$, partitioned to the blocks $\Sigma_{xx}, \Sigma_{uu}, \Sigma_{ux}, $, and $\Sigma_{xu}$. Moreover, $\mu_{n_i(t)}= r_{n_i(t)}+\sqrt{r_{n_i(t)}}\vartheta_i \|V_{n_i(t)}\|^{1/2}$, where the rationale for this choice is discussed in~\cite{chekan2024any}. Denoting the optimal solution of program~(\ref{eq:RelaxedSDP}) by $\Sigma(\mathcal{C}_{n_i(t)})$, the control signal obtained from the relaxed primal SDP~(\ref{eq:RelaxedSDP}), namely $u_t = K(\mathcal{C}_{n_i(t)})x_t$, is deterministic and linear in the state, where 
\begin{align}
K(\mathcal{C}_{n_i(t)})=\Sigma^i_{ux}(\mathcal{C}_{n_i(t)}){{\Sigma^i_{xx}}^{-1}(\mathcal{C}_{n_i(t)})}. \label{eq:obtPol}
\end{align}

The relaxed primal problem~(\ref{eq:RelaxedSDP}) is primarily used for control design. However, for stability analysis, minimum mode-dependent dwell-time design, and regret analysis, we require its dual program, given as follows:
\begin{align}
\begin{array}{rrclcl}
\max & \multicolumn{1}{l}{P \bullet W} \\
\textrm{s.t.}
&
\begin{pmatrix}
Q^i - P & 0 \\
0 & R^i
\end{pmatrix}
+
\hat{\Theta}^{i}_{n_i(t)}
P
\hat{\Theta}^{i^\top}_{n_i(t)}
\succeq
\mu^i_{n_i(t)}
\|P\|_*
V^{-1}_{n_i(t)},
\\
&
P \succeq 0.
\end{array}
\label{eq:RedSDP_DUAL}
\end{align}

We denote the optimal solution of the relaxed dual problem~(\ref{eq:RedSDP_DUAL}) by $P(\mathcal{C}_{n_i(t)})$. We refer the reader to~\cite{cohen2019learning, chekan2024any} for a detailed discussion of this relaxation.

\begin{algorithm}[H]
\caption{Safety Aware Switching (SAS) Algorithm}
\label{alg:OSL3}
\begin{algorithmic}[1]

\State \textbf{Inputs:} Current mode index $i_k$, switching time $t_k^a$, 
observed next mode index $i_{k+1}$, $\Theta_0^{i_k}$, $\Theta_0^{i_{k+1}}$, 
$\underline{\epsilon}(\bar{\kappa}^{i_k})$, 
$\underline{\epsilon}(\bar{\kappa}^{i_{k+1}})$

\State Use the confidence ellipsoids $\mathcal{C}_{n_{i_k}(t_k^a)}, 
\mathcal{C}_{n_{i_{k+1}}(t_k^a)} \in \Pi_{t_k^a}$

\State Solve the relaxed primal SDP (\ref{eq:RelaxedSDP}) for 
$\Sigma(\mathcal{C}_{n_{i_k}(t_k^a)})$, compute 
$K(\mathcal{C}_{n_{i_k}(t_k^a)})$ via (\ref{eq:obtPol}), 
and solve the dual problem to obtain 
$P(\mathcal{C}_{n_{i_k}(t_k^a)})$ via (\ref{eq:RedSDP_DUAL})

\State Solve the Lyapunov equation for the fallback policy $K_0^{i_k}$ 
to obtain $P_{K_0^{i_k}}$

\State Define the implemented controller as
\[
(\tilde{K}_{i_k}, \tilde{P}_{i_k})=
\begin{cases}
\big(K(\mathcal{C}_{n_{i_k}(t_k^a)}), P(\mathcal{C}_{n_{i_k}(t_k^a)})\big), & \text{if } (n_{i_k}(t_k^a)+\bar{c}_{i_k})^{\frac{1}{2}-\zeta}\geq \log \frac{2d_xt_k^a}{\delta}, \quad $\Comment{Learning-based policy (L)}$\\
(K_0^{i_k}, P_{K_0^{i_k}}), & \text{otherwise} \quad $\Comment{Fallback policy (F)}$.
\end{cases}
\]

\For{$t=t_k^a,t_k^a+1,\ldots,t_k^a+\tau_{k,k+1}^a-1$}

    \State Apply the control input
    \[
    u_t=\tilde{K}_{i_k}x_t+\eta_t^{i_k},
    \]
    where
    \[
    \eta_t^{i_k}\sim
    \mathcal{N}(0,\Gamma_t^{i_k}),\quad \Gamma_t^{i_k}=\frac{2\bar{\kappa}_*^2\bar{p}^i\sigma_{\omega}^2}
    {(n_{i_k}(t)+\bar{c}_i)^{\zeta}}.
    \]

\State Determine the policy for the next mode:

\[\tilde{P}_{i_{k+1}}=
\begin{cases}
P(\mathcal{C}_{i_{k+1}(t_k^a)}) , & \text{if } (n_{i_{k+1}}(t_k^a)+\bar{c}_{i_{k+1}})^{\frac{1}{2}-\zeta}\geq \log \frac{2nt}{\delta}, \quad $\Comment{Learning-based policy (L)}$\\
P_{K_0^{i_{k+1}}}, & \text{otherwise} \quad $\Comment{Fallback policy (F)}$.
\end{cases}\]

\State   Compute $\tau^a_{k,k+1}$ according to 
(\ref{eq:dwellAlg})-(\ref{eq:dwell_cases}) 
for either of the four possible scenarios: 
(L,L), (L,F), (F,L), and (F,F)

\State Update $\tau_{k,k+1}^a$

    \State Save $x_t$ and $u_t$
\EndFor

\State Update the confidence set for mode $i_k$ and store it as 
$\mathcal{C}_{n_{i_k}(t_{k+1}^a)}$ 
\end{algorithmic}
\end{algorithm}

% \begin{algorithm}[H]
% \caption{Safety Aware Switching (SAS) Algorithm \label{alg:OSL3}}
% \begin{algorithmic}[1]
% \STATE \textbf{Inputs:} Current mode index $i_k$, its switched time $t_k^a$ and received mode index $i_{k+1}$, $\Theta_0^{i_k}$, $\Theta_0^{i_{k+1}}$, $\underline{\epsilon}(\bar{\kappa}^{i_k})$ and $\underline{\epsilon}(\bar{\kappa}^{i_{k+1}})$
% \STATE Call $\mathcal{C}_{n_{i_k}(t_k^a)}, \mathcal{C}_{n_{i_{k+1}}(t_k^a)}\in \Pi_{t_k^a}$
% \STATE Solve relaxed primal SDP (\ref{eq:RelaxedSDP}) for $\Sigma(\mathcal{C}_{n_{i_k}(t_k^a)})$ and compute $K(\mathcal{C}_{n_{i_k}(t_k^a)})$ via (\ref{eq:obtPol}) 
% \STATE Solve relaxed dual SDP for $\mathcal{C}_{n_{i_k}(t_k^a)}$ and  $\mathcal{C}_{n_{i_{k+1}}(t_k^a)}$ to obtain $P(\mathcal{C}_{n_{i_k}(t_k^a)})$ and $P(\mathcal{C}_{n_{i_{k+1}}(t_k^a)})$, and then compute $\tau_{k,k+1}^a$
% \FOR{$t=t_k^a, t_k^a+1,...,t_k^a+\tau_{k,k+1}^a$}
%     \STATE Play $u_t=K(\mathcal{C}_{n_{i_k}(t_k^a)})x_t+\eta^{i_k}_t$ where $\eta^{i_k}_t\sim\mathcal{N}\big(0, \Gamma^{i_k}_{t}\big)$, with $\Gamma^{i_k}_t$, is given by (\ref{eq:Gamarslop}), with the superscript $i_k$ denoting the mode index.
%     \STATE Save $x_t$ and $u_t$
% \ENDFOR

% \STATE Update the confidence bound for mode $i_k$ and save as $\mathcal{C}_{n_{i_k}(t_{k+1}^a)}$.
% \end{algorithmic}
% \end{algorithm}
\section {Algorithm Design} \label{sec:solutionP}

In this section, we introduce the proposed online algorithm, referred to as the 
\emph{Safety-Aware Switching (SAS) algorithm}, which jointly designs feedback controllers and minimum mode-dependent dwell times when the dynamics of all subsystems are unknown and only noisy state measurements are available. The key challenge is to ensure safety in the sense of Definition \ref{def:underControl} while exploiting collected data to improve the controller performance. To this end, SAS combines an optimism-based learning procedure with a certification mechanism that determines whether a learned controller can be safely deployed or whether the algorithm should revert to a known stabilizing fallback controller.

We first introduce notation used throughout the algorithm description. Let $t_{k+1}^a$ denote the switching time determined by the algorithm, at which the system transitions from mode $i_k$ to mode $i_{k+1}$. The corresponding epoch length, or minimum mode-dependent dwell time, is given by $\tau_{k,k+1}^a=t_{k+1}^a-t_k^a$, where by definition, $t_0^a=t_0=0$.

The proposed algorithm builds upon the Optimism in the Face of Uncertainty (OFU) framework introduced in~\cite{chekan2024any} for LQ control with anytime regret guarantees. The algorithm employs an SDP-based design procedure to construct learning-based feedback gains and determine their associated dwell times whenever the learned controller can be certified as safe. Specifically, when the available confidence set is sufficiently informative, the primal and dual solutions of the relaxed SDP~(\ref{eq:RelaxedSDP}) are used to synthesize 
the feedback gain and compute a sufficient dwell time guaranteeing bounded state evolution in the sense of Definition~\ref{def:underControl}. When certification is not available, the algorithm activates a known stabilizing fallback controller $K_0^{i_k}$ and computes the dwell time using the corresponding Lyapunov solution of the fallback controller. A detailed description of SAS is provided in Algorithm~\ref{alg:OSL3}.

Algorithm~\ref{alg:OSL3} assumes an initial estimate $\Theta_0^i$ for each subsystem satisfying $\|\Theta_0^i-\Theta_*^i\|_F\leq\underline{\epsilon}(\bar{\kappa}_i)$ where $\underline{\epsilon}(\bar{\kappa}_i)$ is defined in~(\ref{eq:epsstst0}). The algorithm follows a commit-then-update strategy: the controller remains fixed throughout each switching epoch, while all collected data are incorporated only after the epoch terminates. Define the collection of confidence ellipsoids available at time $t$ as
\begin{align}
\Pi_t=
\left\{
\mathcal{C}_{n_i(t)}(\delta)\mid i=1,\ldots,|\mathcal{M}|
\right\}.
\label{eq:confsetssets}
\end{align}

Suppose that at switching time $t_k^a$ the system enters mode $i_k$. The algorithm first checks whether the current confidence set is sufficient to certify the learned controller. Specifically, if
\begin{align*}
    (n_{i_k}(t_k^a)+\bar{c}_{i_k})^{\frac{1}{2}-\zeta}\geq \log \frac{2nt_k^a}{\delta}
\end{align*}

then the learning-based controller is certified and synthesized by solving the relaxed primal SDP~(\ref{eq:RelaxedSDP}) using the confidence ellipsoid  $\mathcal C_{n_{i_k}(t_k^a)}(\delta)\in\Pi_{t_k^a}$. Otherwise, the algorithm selects the known stabilizing controller $K_0^{i_k}$. Regardless of the selected policy, the algorithm injects an additive input perturbation $\eta_t^{i_k}\sim\mathcal{N}(0,\Gamma_t^{i_k})$, where $\Gamma_t^{i_k}$ is defined in (\ref{eq:Gamarslop}). The constants $\bar{c}_{i_k}$ are chosen according to the proof of Theorem~\ref{Stability_thm17} to ensure that the perturbation covariance $\Gamma_t^{i_k}$ remains of the same order as the process noise covariance $\sigma_{\omega}^2 I$.

After the next mode $i_{k+1}$ becomes available, SAS computes the minimum dwell time required before the next switching event. Since consecutive epochs may use either learning-based (L) or fallback (F) controllers, four possible switching patterns arise: (L,L), (L,F), (F,L), and (F,F). For all cases, the dwell time is computed as
\begin{align}
\tau_{k,k+1}^{a}=\max\left\{1,-\frac{\ln \tilde{\rho}_{i_k}(t_k^a)
+\ln \tilde{\mathcal{X}}_{i_{k+1}}(t_k^a)-\ln\bar{\alpha}}{\ln(1-\tilde{\eta}_{i_k}(t_k^a))}\right\}
\label{eq:dwellAlg}
\end{align}
where
\begin{align}
\tilde{\rho}_{i_k}(t_k^a)
&=
\frac{\overline{\lambda}(\tilde P_{i_k}(t_k^a))}
{\underline{\lambda}(\tilde P_{i_k}(t_k^a))},\\
\tilde{\mathcal X}_{i_{k+1}}(t_k^a)
&=
\frac{\overline{\lambda}(\tilde P_{i_{k+1}}(t_k^a))}
{\underline{\lambda}(\tilde P_{i_{k+1}}(t_k^a))},\\
\tilde{\eta}_{i_k}(t_k^a)
&=
\frac{\underline{\lambda}(\tilde H_{i_k}(t_k^a))}
{\overline{\lambda}(\tilde P_{i_k}(t_k^a))}.
\end{align}

The matrices used in the dwell-time computation are selected according to the 
controller types employed in the two consecutive epochs:
\begin{equation}
(\tilde P_{i_k},\tilde P_{i_{k+1}},\tilde H_{i_k})
=
\begin{cases}
(P(\mathcal C_{n_{i_k}}),P(\mathcal C_{n_{i_{k+1}}}),
\mathcal H(\mathcal C_{n_{i_k}})),
& (L,L),\\[1ex]
(P(\mathcal C_{n_{i_k}}),P_{K_0^{i_{k+1}}},
\mathcal H(\mathcal C_{n_{i_k}})),
& (L,F),\\[1ex]
(P_{K_0^{i_k}},P(\mathcal C_{n_{i_{k+1}}}),
H_{K_0^{i_k}}),
& (F,L),\\[1ex]
(P_{K_0^{i_k}},P_{K_0^{i_{k+1}}},
H_{K_0^{i_k}}),
& (F,F).
\end{cases}
\label{eq:dwell_cases}
\end{equation}
where $\mathcal H(\mathcal C_{n_{i_k}(t_k^a)})$ and $H_{K_0^{i_k}}$ are computed 
according to~(\ref{eq:HDef3p}) and~(\ref{eq:lyapunovEq}), respectively.

For the (L,L) case, the dwell time is computed using the primal and dual SDP solutions associated with the confidence sets, following 
Theorem~\ref{thm:minimum_average_dwell_revised}. For the (F,F) case, the dwell time is obtained using the known stabilizing controllers (fallback controllers) and their corresponding Lyapunov matrices according to Theorem~\ref{thm:alfabet_K_S}. 
The two mixed cases, (L,F) and (F,L), combine these two constructions, and their justification follows from analogous arguments to those in 
Theorems~\ref{thm:alfabet_K_S} and~\ref{thm:minimum_average_dwell_revised}.

After the switch, the newly collected samples are 
incorporated into the corresponding confidence ellipsoid, resulting in the updated collection $\Pi_{t_{k+1}^a}$ used at the next switching instant.

The certification threshold is selected such that whenever the certification condition is satisfied, the learning-based controller is guaranteed to be $(\bar\kappa_{i_k},\bar\gamma_{i_k})$-strongly stabilizing. Conversely, when
\begin{align}
n_{i_k}(t_k^a)+\bar c_{i_k}
\leq
\left(
\log\frac{2d_xt_k^a}{\delta}
\right)^{\frac{2}{1-2\zeta}},
\label{eq:switchtresh}
\end{align}
the algorithm employs the fallback controller. This certification mechanism 
ensures that fallback operation occurs for only a polylogarithmically bounded 
number of time steps, which is essential for establishing the regret guarantee.

In order for any learning-based policy \(K(\mathcal{C}_{n_i(t)})\) designed by the algorithm to be stabilizing, we need to appropriately tune the input perturbation noise power \(\Gamma_t^i\), the regularization parameter \(\lambda_i\) used in the construction of the confidence ellipsoid, and the initial estimate \(\Theta_0^i\), characterized by \(\epsilon (\bar{\kappa}_i)\). The following theorem provides the machinery for such a design and shows that the resulting designed policies are \((\bar{\kappa}_i,\bar{\gamma}_i)\)-strongly stabilizing, as defined in Definition~\ref{def:sequentially}. For notational simplicity, we omit the subscripts and superscripts \(i\), which denote the mode index associated with each epoch, in the statement of the theorem.

\begin{theorem}(Stability of data-driven control in an epoch) \label{Stability_thm17}
Fix $\zeta\in (0,\, 1/2)$. Set the regularization parameter $\lambda_i$ for mode $i\in \mathcal{M}$ in Algorithm \ref{alg:OSL3} to
\begin{align}
    \lambda_i &= \frac{\sigma_{\omega}^2 \bar{c}_i}{80}, \label{eq:lambda_valuefn01}\\ 
    \bar{c}_i&=(2\bar{\kappa}_i^2\bar{p}^i \bar{\vartheta}_{B_*^i} )^{\frac{1}{\zeta}}, \label{eq:cfARSLOplus}
\end{align}
where
\begin{align}
\bar{p}^i:=\frac{5120 \bar{\kappa}_i^2\vartheta_i\sqrt{\sigma_{\omega}^2d_x (d_x+d_u^i)}}{\sigma_{\omega}^2}(c_z+\sqrt{\lambda_i}). \label{eq:pylop},\\
  \nonumber  c_z=\frac{4 \tilde{g}\sqrt{1+\kappa_*^2}}{(1-\bar{\alpha})}\bigg(\|x_1\|+\sqrt{20(d_x+d^*_u) \sigma_{\omega}^2}\bigg),
\end{align}
\begin{align}
    \bar{\vartheta}_{B_*} := \max_{i}\{1, \vartheta_{B^i_*}\}, \quad \vartheta_{B^i_*} \ge \|B^i_*\|, \label{eq:kappabar_no_seq_Def}
\end{align}
$\max_id_u^i=d_u^*$ and with $\kappa_*:=\max\{\bar{\kappa}_*,\kappa_0^*\}$,
\begin{align*}
    \tilde{g}
    :=
    \max\left\{
        2\kappa_*^4,
        \frac{\kappa_*^2}{\gamma_0^*}
    \right\}.
\end{align*}

Define the additive exploration noise by
\begin{align}
    \eta^i_t \sim \mathcal{N}\big(0, \Gamma^i_t\big), \qquad
    \Gamma^i_t := 2\frac{\bar{p}^i \bar{\kappa}_i^2 \sigma_\omega^2}{(n_i(t)+\bar{c})^{\zeta}}.\label{eq:Gamarslop}
\end{align}

Assume the initial estimate satisfies
 $\Theta^i_0$ satisfied
\begin{align}
    \|\Theta^i_0 - \Theta^i_*\|_F \le \underline{\epsilon}(\bar{\kappa}_i), \qquad 
   \underline{\epsilon}(\bar{\kappa}_i) := \frac{\sigma_\omega \sqrt{2 d_x (d_x+d_u^i)}}{\sqrt{\lambda_i}}.\label{eq:epsstst0} 
\end{align}
Suppose that, at time $t$, Algorithm~\ref{alg:OSL3} starts a new epoch in mode $i$, and let $n_i(t)$ denote the total number of visits to mode $i$ up to time $t$. If
\begin{align}
    (n_i(t)+\bar{c}_i)^{\frac{1}{2}-\zeta}\geq \log \frac{2d_xt}{\delta} \label{eq:theresholdpol}
\end{align}
then every policy generated during that epoch by Algorithm \ref{alg:OSL3} is \((\bar{\kappa}_i, \bar{\gamma}_i)\)-strongly stabilizing, with probability at least $1-3\delta$, where $\delta\in(0,\frac13)$ and
\begin{align}
    \bar{\kappa}_i := \sqrt{\frac{2 \nu_i}{\alpha^i_0 \sigma_{\omega}^2}}, \qquad \bar{\gamma}_i := \frac{1}{2 \bar{\kappa}_i^2}.
\end{align}
\end{theorem}

% Recalling Definition 1, a sequentially stabilizing policy keeps the states of system bounded in an epoch where the corresponding upper bound is in terms of the stabilizing policy parameters $\kappa_i$ and $\gamma_i$.

With the state-feedback gain \(K(\mathcal{C}_{n_{i_k}(t_k^a)})\) computed via (\ref{eq:obtPol}) and the index of the next mode \(i_{k+1}\) available, the algorithm also determines the duration of the corresponding epoch. This duration is the data-driven minimum mode-dependent dwell-time that guarantees  \((\bar{\alpha}, \bar{\beta})\)-controllability of the state norm. The following theorem characterizes this duration.

\begin{theorem} (Minimum Mode-Dependent Dwell-Time)\label{thm:minimum_average_dwell_revised}
Consider the switched system described by (\ref{eq:dyn_atttt}), which undergoes a mode transition to mode $i$ at time $t$ and subsequently identifies the next mode index $j$ for the upcoming switch. Let $K(\mathcal{C}_{n_i(t)})$ denote the control feedback designed for mode $i$, and let $P(\mathcal{C}_{n_i(t)})$ and $P(\mathcal{C}_{n_j(t)})$ represent the corresponding solutions of relaxed dual SDPs. Then, the growth in state norms during the transition from mode $i$ to $j$ is under control in the sense of Definition \ref{def:underControl}, provided that the duration of actuation in mode $i$ lasts for at least the minimum dwell-time of $\tau_{i,j}^{a}$ given by: \footnote{Note that in the proof of Theorem \ref{thm:minimum_average_dwell_revised} it is shown that $0< \eta (\mathcal{C}_{n_i(t)})< 1$.}
     \begin{align}
 &\tau_{k,k+1}^{a}:=\max \bigg\{1, -\frac{\ln \rho (\mathcal{C}_{n_{i_k}(t_k^a)})+\ln \mathcal{X} (\mathcal{C}_{n_{i_{k+1}}(t_k^a)})-\ln \bar{\alpha}}{\ln \big(1-\eta \big(\mathcal{C}_{n_{i_k}(t_k^a)}\big)\big)}\bigg\}\label{eq:tau_a_fast}, 
     \end{align}
     where
     \begin{align}
	\eta \big(\mathcal{C}_{n_{i_k}(t_k^a)}\big):= \frac{\underline{\lambda}\big(\mathcal{H}(\mathcal{C}_{n_{i_k}(t_k^a)})\big)}{\overline{\lambda}\big(P(\mathcal{C}_{n_{i_k}(t_k^a)})\big)},\label{wq:etaDef}
	\end{align}
\begin{align}
\nonumber  \rho (\mathcal{C}_{n_{i_k}(t_k^a)}):=\frac{\overline{\lambda}\big(P(\mathcal{C}_{n_{i_{k}}(t_k^a)})\big)}{\underline{\lambda}\big(P(\mathcal{C}_{n_{i_k}(t_k^a)})\big)},\\
\mathcal{X} (\mathcal{C}_{n_{i_{k+1}}(t_k^a)}):=\frac{\overline{\lambda}\big(P(\mathcal{C}_{n_{i_{k+1}}(t_k^a)})\big)}{\underline{\lambda}\big(P(\mathcal{C}_{n_{i_{k+1}}(t_k^a)})\big)},
\label{eq:rhoDef}
\end{align} 
and 
 	\begin{align}
	\mathcal{H}(&\mathcal{C}_{n_{i_k}(t_k^a)})\!=\!Q^{i_k}\!+\!K^\top (\mathcal{C}_{n_{i_k}(t_k^a)})R^{i_k} K(\mathcal{C}_{n_{i_k}(t_k^a)}) -2\mu^{i_k}_{n_{i_k}(t_k^a)}\|P(\mathcal{C}_{n_{i_k}(t_k^a)})\|_* \begin{pmatrix} I \\ K(\mathcal{C}_{n_{i_k}(t_k^a)}) \end{pmatrix}{V}^{-1}_{n_{i_k}(t_k^a)}\begin{pmatrix} I \\ K(\mathcal{C}_{n_{i_k}(t_k^a)}) \end{pmatrix}^\top\!\!\!.
 \label{eq:HDef3p}
	\end{align}
   with probability at least $1-6\delta$.
%  \begin{align}
%    \mathcal{G}(\mathcal{C}_{\tau_q}^j):=P(\hat{\Theta}^j_{\tau_{q}}, Q^j, R^j)+\chi_{\tau_q}^j\label{eq:GDef}
% \end{align}
% and
    %  \begin{align}
    %   \nonumber  \chi_{\tau_q}^i&:=\frac{2\kappa_i^2\mu_i^{\tau_q}}{\gamma}\|P(\hat{\Theta}^i_{\tau_q}, Q^i, R^i)\|_* \times\\
    %    &\bigg\|\begin{pmatrix} I \\ K(\hat{\Theta}^i_{\tau_q}, Q^i, R^i) \end{pmatrix}{V^i}^{-1}_{n_i(\tau_q)}\begin{pmatrix} I \\ K(\hat{\Theta}^i_{\tau_q}, Q^i, R^i) \end{pmatrix}^\top\bigg\|I \label{eq:chiddef_th}
    % \end{align}
\end{theorem}

\section{Theoretical Guarantees} \label{sec:theoreticalg}

\subsection{$(\bar{\alpha}, \bar{\beta})$-Controllability and Dwell-Time Estimation Error Bounds}

The following lemma summarizes how Algorithm~\ref{alg:OSL3} guarantees \((\bar{\alpha}, \bar{\beta})\)-controllability of the state and provides an overall upper bound on the expected closed-loop state norm.

% Furthermore, when the mode $i$ is suppressing comparing to the mode $j$ at epoch starting at time $\tau_q$ (See Definition), then $\rho (\mathcal{C}^i_{\tau_q}(\delta), \mathcal{C}^j_{\tau_{q-1}}(\delta))<0$, meaning that there is no need to stay on the mode $i$ and the subsequent switch can happen quickly.

\begin{lemma}\label{thm:res_sequential_stability_resp}
Fix $\delta\in (0,\, \frac{1}{3|\mathcal{M}|})$. Algorithm \ref{alg:OSL3} guarantees
\begin{align}
    &\mathbb{E}[x^\top _{t^a_{k+1}}x_{t^a_{k+1}}| \mathcal{F}_{t^a_k}]\leq \bar{\alpha} \; \mathbb{E}[x^\top _{t^a_k}x_{t^a_k}| \mathcal{F}_{t^a_k}]+ \tilde{\beta}\sigma_{\omega}^2, \label{eq:moghayese}
\end{align}
with probability at least $1-6\delta$ where $\tilde{\beta}:=\bar{\kappa}_*^6$.
% Furthermore, we have
% \begin{align}
% \mathbb{E}[x^\top_{t_{k}^a}x_{t_{k}^a}]\leq \bar{\alpha}^{k} \|x_0\|^2 +\frac{\tilde{\beta}}{1-\bar{\alpha}}\sigma_{\omega}^2
% \end{align}
% \textcolor{blue}{with probability at least $1-3|\mathcal{M}|\delta$.}
\end{lemma}

% \begin{corollary}
% The policy employed in Algorithm 1 falls within the category of candidate policies as defined by (\ref{eq:policyClass}), i.e., $\kappa_i\leq \kappa^i_{c}$ for all $i\in |\mathcal{M}|$. Now, noting that $\gamma^i_{c}<1$, consequently, upon comparing (\ref{eq:moghayese}) from Lemma \ref{thm:res_sequential_stability_resp} with (\ref{eq:stategrowthpp}), it becomes evident that Algorithm 1 ensures the containment of state norm within bounds, aligning with the conditions specified in Definition \ref{def:underControl}. 
% \end{corollary}

Finally, the following theorem establishes an upper bound on the dwell-time estimation error. This bound is particularly useful for the regret analysis of the proposed algorithm.

\begin{theorem} (Dwell-Time Estimation Error Upper-Bound)
\label{Thm:dwellTimeError}
Suppose that at time \(t\), the system initiates actuation in mode \(i\) and immediately receives the next mode \(j\). If the system applies the learning-based controller during the subsequent epochs, then the estimation error of the minimum mode-dependent dwell time,  \(\tau^{a}_{k,k+1}-\tau^*_{k,k+1}\), admits the following upper bound:
\begin{align}
    \tau_{k,k+1}^a-\tau_{k,k+1}^*\leq \bar{C}_0\chi^{i_{k+1}}_{t_k^a}+\bar{C}_1\chi^{i_k}_{t_k^a}+\bar{C}_2\big(\chi_{t_k^a}^{i_k}\big)^2,
\end{align}
where
\begin{align}
\chi_{t}^{i}= \frac{16\| P(\mathcal{C}_{n_{i}(t)})\|^3\mu_{n_{i}(t)}}{\big(\alpha^{i}_0\big)^2}  \Big\|\begin{pmatrix}
I \\
K(\mathcal{C}_{n_{i}(t)})
\end{pmatrix}^\top {V}^{-1}_{n_{i}(t)} \begin{pmatrix}
I \\
K(\mathcal{C}_{n_{i}(t)}) 
\end{pmatrix} \Big\|\label{eq:chiddef3}
\end{align}
and $\bar{C}_0$, $\bar{C}_1$, and $\bar{C}_2$ are problem-dependent constants.
\end{theorem}
For the case where the next mode \(j\) operates with the fallback controller, the corresponding dwell-time estimation error is characterized by Corollary~\ref{cor:LSBound}.

\subsection{Regret Bound Analysis} \label{sec:regban}
% In this subsection, we carry out the regret bound analysis for the proposed Algorithm \ref{alg:OSL3}. We leave the regret bound analysis of warm-up phase (Algorithm \ref{alg:IExp}) a side, as by adapting the analysis provided by \cite{cohen2019learning}, it is straight forward to see that for each mode the algorithm has regret bound of $\sum_{i=1}^{|\mathcal{M}|}\mathcal{O}(\sqrt{T_0^i})$ where $\mathcal{O}$ absorbs system ambient dimension dependent and logarithmic $T_0^i$ terms.

To establish an upper bound on the regret of Algorithm~\ref{alg:OSL3}, we first define a sequence of nested good events under which the confidence sets, controller synthesis procedure, and state-dependent quantities satisfy the required bounds. Specifically, consider the nested sequence of good events $\mathcal{E}^{\zeta}_{n_s-1}\subseteq \mathcal{E}^{\zeta}_{n_s-2}\subseteq \ldots\subseteq \mathcal{E}^{\zeta}_1\subseteq \mathcal{E}^{\zeta}_0$ 
where, for each $k$ 
\begin{align}
\nonumber \mathcal{E}^{\zeta}_{k}=\Big\{&\forall s = 0, \ldots, k, \quad \Theta_*^{i}\in \mathcal{C}_{n_{i}(t_s^a)},\quad  \mu_{n_i(t_k^a)} \|V^{-1}_{n_i(t_k^a)}\|
    \leq
    \bar{D}_i(n_i(t_k^a)+\bar{c}_i)^{-\frac{1}{2}+\zeta} \log \frac{t_k^a}{\delta},\,\forall i\in \mathcal{M},\\
    &\textit{and}\quad \|z_s\|^2\leq c^2_z\log \frac{t_k^a}{\delta})\Big\}.
\end{align}
The event $\mathcal{E}^{\zeta}_{k}$ collects the confidence ellipsoids evaluated at the switching times \(t_k^a\), which are used to construct the feedback gains and the minimum mode-dependent dwell times. It also includes the upper bound on $\mu_{n_i(t_k^a)} \|V^{-1}_{n_i(t_k^a)}\|$, established in Lemma~\ref{lem:useflboun}. Finally, it contains a uniform upper bound on the co-state norm, given by Lemma \ref{lem:defCbarz}.

We first establish that the above sequence of good events holds with high probability.

\begin{lemma} \label{lem:epseventprob}
Fix $\delta\in (0,\, \frac{1}{3|\mathcal{M}|})$, then $\mathbb{P}(\mathcal{E}_{n_s-1}) \ge 1 - 3|\mathcal{M}|\delta.$
\end{lemma}

Next, we divide the switching epochs into two classes depending on whether the
learned controller satisfies the certification condition. Define the set of
certified epochs as
\begin{align}
\mathcal{I}_{\mathrm{cert}}
=
\left\{
k:
(n_{i_k}(t_k^a)+\bar{c}_{i_k})^{\frac12-\zeta}
\geq
\log\frac{2d_xt_k^a}{\delta}
\right\},
\end{align}
and the set of fallback epochs as $\mathcal{I}_{\mathrm{fallback}}
=
\mathcal{I}\setminus\mathcal{I}_{\mathrm{cert}}$, where $\mathcal{I}=\{0,\ldots,n_s-1\}$ denotes the set of all switching epochs.
Accordingly, the per-epoch regret is decomposed as
\begin{align*}
R_{k,k+1}(\zeta)
=
\begin{cases}
R_{k,k+1}^{\mathrm{cert}},
&
k\in\mathcal{I}_{\mathrm{cert}},\\[1mm]
R_{k,k+1}^{\mathrm{fallback}},
&
k\in\mathcal{I}_{\mathrm{fallback}}.
\end{cases}
\end{align*}

The objective is to bound the accumulated regret $\mathbb{E}[R_{\zeta}(\mathcal{I})]
=
\mathbb{E}
\left[
\sum_{k=0}^{n_s-1}R_{k,k+1}(\zeta)
\right]$. To this end, we first analyze the truncated regret $\mathbb{E}[\tilde{R}_{\zeta}(\mathcal{I})]
=
\mathbb{E}
\left[
\sum_{k=0}^{n_s-1}
R_{k,k+1}(\zeta)
1_{\mathcal{E}^{\zeta}_{k}}
\right]$,
which captures the regret accumulated when the sequence of good events holds. The resulting bound is then converted into a bound on the original expected
regret by accounting for the probability of the event. Here, by a
slight abuse of notation, the subscript $\zeta$ emphasizes the dependence of
the regret on the perturbation-noise exponent $\zeta$ introduced in
(\ref{eq:Gamarslop}). 

The fallback epochs occur when the certification condition (\ref{eq:theresholdpol})
is not satisfied, in which case the algorithm operates the system using the
fallback controller. According to Lemma~\ref{lem:fallback_steps}, the number of
fallback epochs is bounded by $\mathcal{O}(
|\mathcal{M}|
\log^{^{\frac{2}{1-2\zeta}}}
\frac{n_s}{\delta}
)$. In contrast, the regret incurred during certified epochs requires a more
involved decomposition, which is presented in detail in Appendix~\ref{sec:app-reg}.
In particular, this decomposition includes the term
\begin{align*}
R^{(2)}_{k,k+1}
=
\sum_{t=t_k^a+\tau_{k,k+1}^*}^{t_{k+1}^a-1}
J_*^{i_k}1_{\mathcal{E}^{\zeta}_k},
\end{align*}
which captures the regret caused by the uncertainty in estimating the minimum
dwell time. Specifically, this term accounts for the additional cost incurred
when the system remains in a mode longer than the optimal dwell time due to
learning errors. The regret is defined relative to the feedback controllers
obtained from the solutions of the DAREs in each epoch. We refer the reader to
Appendix~\ref{sec:app-reg} for a detailed derivation of the regret decomposition
and its subsequent analysis.

The regret analysis reveals that the type of switching plays an important role. In particular, a switch from mode \(i_k\) to mode \(i_{k+1}\) such that \(P(\Theta_*^{i_k}) \prec P(\Theta_*^{i_{k+1}})\) is less desirable, as it introduces unfavorable terms into the regret bound. Motivated by this observation, we introduce the following definition.

\begin{definition}{(\textbf{Benign vs.\ Malignant Switches})}
\label{def:benign_malignant}
A switch from mode \(i_k\) to mode \(i_{k+1}\) is called \emph{benign} if $P(\Theta_*^{i_k}) \succeq P(\Theta_*^{i_{k+1}})$,
and \emph{malignant} otherwise.
\end{definition}

The following theorem summarizes the resulting regret bound for Algorithm~\ref{alg:OSL3}.

\begin{theorem}
\label{thm:RegretBound}
Fix $\delta\in (0,\, \frac{1}{3|\mathcal{M}|})$. With probability at least \(1-3|\mathcal{M}|\delta\), Algorithm~\ref{alg:OSL3} achieves the following expected regret bound:
\begin{align}
    \mathbb{E}\big[\mathcal{R}_{\zeta}(\mathcal{I})\big]
    \lesssim
    \max \bigg\{
    &
    \mathcal{O}\big(
        |\mathcal{M}|^{\zeta} n_s^{1-\zeta}
    \big),
    \nonumber\\
    &
    \mathcal{O}\bigg(
        n_m\log\frac{n_s}{\delta}
        +
        |\mathcal{M}|
        \left(
        \log\frac{n_s}{\delta}
        \right)^{\frac{3-2\zeta}{1-2\zeta}}
        +
        |\mathcal{M}|^{1-\gamma}
        n_s^{\gamma}
        \log^2\frac{n_s}{\delta}
    \bigg)
    \bigg\},
\end{align}
where \(n_m\) denotes the number of malignant switches in the revealed sequence \(\mathcal{I}\) (see Definition~\ref{def:benign_malignant}), and \(\gamma = \tfrac{1}{2} + \zeta\) for any \(\zeta \in (0,\tfrac{1}{2})\). The optimal expected regret $\mathbb{E}\big[\mathcal{R}_{*} (\mathcal{I})\big]$ is achieved when \(\zeta = \tfrac{1}{4}\), yielding a bound of order
\begin{align*}
    \mathbb{E}\big[\mathcal{R}_{*}(\mathcal{I})\big]
    \lesssim
    \mathcal{O}\left(
        n_m\log\frac{n_s}{\delta}
        +
        |\mathcal{M}|
        \log^5\frac{n_s}{\delta}
        +
        |\mathcal{M}|^{1/4}
        n_s^{3/4}
        \log^2\frac{n_s}{\delta}
    \right).
\end{align*}
\end{theorem}

Let \(n_s = n_b + n_m\) denote the total number of switches, where \(n_b\) and \(n_m\) represent the numbers of benign and malignant switches, respectively. We observe that the regret upper bound depends primarily on the number of modes \(|\mathcal{M}|\) and the number of malignant switches, while its dependence on the total number of switches is only sublinear.

\section{Conclusion}\label{sec:conclusion}

In this paper, we propose an algorithm that guarantees safe switching with minimum cost in a setting where switching sequences are revealed online. Our algorithm is based on the Optimism in the Face of Uncertainty  principle and integrates system identification, control, and dwell-time design into a unified framework. By leveraging confidence sets constructed for parameter estimates, we develop a strategy for accurately estimating the minimum dwell time. Furthermore, we prove that the proposed algorithm achieves an expected regret of \(\mathcal{O}(|\mathcal{M}|^{1/4} n_s^{3/4} + n_m)\) relative to the benchmark setting in which all mode parameters are known, where \(n_s\) denotes the number of switches, \(|\mathcal{M}|\) is the number of modes, and \(n_m\) is the number of malignant switches. A possible extension of this work involves a setting in which the selection of the next mode becomes an integral part of the strategy design. For example, the LQ formulation studied in~\cite{li2023online} has the potential to jointly optimize both the switching time and the selection of the next mode, with the objective of identifying the most effective actuation mode while minimizing regret. 
% Additionally, we demonstrate that the expected time to complete the task of safe switching according to a finite switch sequence with minimum cost is of the order $\mathcal{O}(|\mathcal{M}|\sqrt{ns})$, where $ns$ represents the total number of switches and $|\mathcal{M}|$ is the number of mode candidates. 

\bibliography{main}

@inproceedings{abbasi2011regret,
  title={Regret bounds for the adaptive control of linear quadratic systems},
  author={Abbasi-Yadkori, Yasin and Szepesv{\'a}ri, Csaba},
  booktitle={Proceedings of the 24th Annual Conference on Learning Theory},
  pages={1--26},
  year={2011}
}

@article{cohen2019learning,
  title={Learning Linear-Quadratic Regulators Efficiently with only sqrt(T) Regret},
  author={Cohen, Alon and Koren, Tomer and Mansour, Yishay},
  journal={arXiv preprint arXiv:1902.06223},
  year={2019}
}

@inproceedings{fazel2018global,
  title={Global convergence of policy gradient methods for the linear quadratic regulator},
  author={Fazel, Maryam and Ge, Rong and Kakade, Sham and Mesbahi, Mehran},
  booktitle={International conference on machine learning},
  pages={1467--1476},
  year={2018},
  organization={PMLR}
}

@article{mania2019certainty,
  title={Certainty equivalence is efficient for linear quadratic control},
  author={Mania, Horia and Tu, Stephen and Recht, Benjamin},
  journal={Advances in Neural Information Processing Systems},
  volume={32},
  year={2019}
}

@inproceedings{dai2018moments,
  title={A moments based approach to designing MIMO data driven controllers for switched systems},
  author={Dai, Tianyu and Sznaier, Mario},
  booktitle={2018 IEEE Conference on Decision and Control (CDC)},
  pages={5652--5657},
  year={2018},
  organization={IEEE}
}

@article{kenanian2019data,
  title={Data driven stability analysis of black-box switched linear systems},
  author={Kenanian, Joris and Balkan, Ayca and Jungers, Raphael M and Tabuada, Paulo},
  journal={Automatica},
  volume={109},
  pages={108533},
  year={2019},
  publisher={Elsevier}
}

@article{rotulo2022online,
  title={Online learning of data-driven controllers for unknown switched linear systems},
  author={Rotulo, Monica and De Persis, Claudio and Tesi, Pietro},
  journal={Automatica},
  volume={145},
  pages={110519},
  year={2022},
  publisher={Elsevier}
}

@article{dai2022convex,
  title={A convex optimization approach to synthesizing state feedback data-driven controllers for switched linear systems},
  author={Dai, Tianyu and Sznaier, Mario},
  journal={Automatica},
  volume={139},
  pages={110190},
  year={2022},
  publisher={Elsevier}
}

@article{chekan2024learn,
  title={Learn and Control While Switching: Guaranteed Stability and Sublinear Regret},
  author={Chekan, Jafar Abbaszadeh and Langbort, C{\'e}dric},
  journal={IEEE Transactions on Automatic Control},
  volume={69},
  number={12},
  pages={8433--8448},
  year={2024},
  publisher={IEEE}
}

@article{zhu2015optimal,
  title={Optimal control of hybrid switched systems: A brief survey},
  author={Zhu, Feng and Antsaklis, Panos J},
  journal={Discrete Event Dynamic Systems},
  volume={25},
  number={3},
  pages={345--364},
  year={2015},
  publisher={Springer}
}

@article{zhao2008stability,
  title={On stability, L2-gain and H${\infty}$ control for switched systems},
  author={Zhao, Jun and Hill, David J},
  journal={Automatica},
  volume={44},
  number={5},
  pages={1220--1232},
  year={2008},
  publisher={Elsevier}
}

@article{lin2009stability,
  title={Stability and stabilizability of switched linear systems: a survey of recent results},
  author={Lin, Hai and Antsaklis, Panos J},
  journal={IEEE Transactions on Automatic control},
  volume={54},
  number={2},
  pages={308--322},
  year={2009},
  publisher={IEEE}
}

@article{liberzon1999stability,
  title={Stability of switched systems: a Lie-algebraic condition},
  author={Liberzon, Daniel and Hespanha, Joao P and Morse, A Stephen},
  journal={Systems \& Control Letters},
  volume={37},
  number={3},
  pages={117--122},
  year={1999},
  publisher={Elsevier}
}

@inproceedings{hespanha1999stability,
  title={Stability of switched systems with average dwell-time},
  author={Hespanha, Joao P and Morse, A Stephen},
  booktitle={Proceedings of the 38th IEEE conference on decision and control (Cat. No. 99CH36304)},
  volume={3},
  pages={2655--2660},
  year={1999},
  organization={IEEE}
}

@inproceedings{lale2022reinforcement,
  title={Reinforcement learning with fast stabilization in linear dynamical systems},
  author={Lale, Sahin and Azizzadenesheli, Kamyar and Hassibi, Babak and Anandkumar, Animashree},
  booktitle={International Conference on Artificial Intelligence and Statistics},
  pages={5354--5390},
  year={2022},
  organization={PMLR}
}

@article{abbasi2011improved,
  title={Improved algorithms for linear stochastic bandits},
  author={Abbasi-Yadkori, Yasin and P{\'a}l, D{\'a}vid and Szepesv{\'a}ri, Csaba},
  journal={Advances in neural information processing systems},
  volume={24},
  year={2011}
}

@inproceedings{cohen2018online,
  title={Online linear quadratic control},
  author={Cohen, Alon and Hasidim, Avinatan and Koren, Tomer and Lazic, Nevena and Mansour, Yishay and Talwar, Kunal},
  booktitle={International Conference on Machine Learning},
  pages={1029--1038},
  year={2018},
  organization={PMLR}
}

@article{li2023online,
  title={Online switching control with stability and regret guarantees},
  author={Li, Yingying and Preiss, James A and Li, Na and Lin, Yiheng and Wierman, Adam and Shamma, Jeff},
  journal={Proceedings of Machine Learning Research vol XX},
  volume={1},
  pages={29},
  year={2023}
}

@inproceedings{du2022data,
  title={Data-driven control of markov jump systems: Sample complexity and regret bounds},
  author={Du, Zhe and Sattar, Yahya and Tarzanagh, Davoud Ataee and Balzano, Laura and Ozay, Necmiye and Oymak, Samet},
  booktitle={2022 American Control Conference (ACC)},
  pages={4901--4908},
  year={2022},
  organization={IEEE}
}

@inproceedings{boffi2021regret,
  title={Regret bounds for adaptive nonlinear control},
  author={Boffi, Nicholas M and Tu, Stephen and Slotine, Jean-Jacques E},
  booktitle={Learning for Dynamics and Control},
  pages={471--483},
  year={2021},
  organization={PMLR}
}

@article{kakade2020information,
  title={Information theoretic regret bounds for online nonlinear control},
  author={Kakade, Sham and Krishnamurthy, Akshay and Lowrey, Kendall and Ohnishi, Motoya and Sun, Wen},
  journal={Advances in Neural Information Processing Systems},
  volume={33},
  pages={15312--15325},
  year={2020}
}

@inproceedings{chen2021black,
  title={Black-box control for linear dynamical systems},
  author={Chen, Xinyi and Hazan, Elad},
  booktitle={Conference on Learning Theory},
  pages={1114--1143},
  year={2021},
  organization={PMLR}
}

@article{dean2020sample,
  title={On the sample complexity of the linear quadratic regulator},
  author={Dean, Sarah and Mania, Horia and Matni, Nikolai and Recht, Benjamin and Tu, Stephen},
  journal={Foundations of Computational Mathematics},
  volume={20},
  number={4},
  pages={633--679},
  year={2020},
  publisher={Springer}
}

@book{liberzon2003switching,
  title={Switching in systems and control},
  author={Liberzon, Daniel},
  volume={190},
  year={2003},
  publisher={Springer}
}

@article{zhao2011stability,
  title={Stability and stabilization of switched linear systems with mode-dependent average dwell time},
  author={Zhao, Xudong and Zhang, Lixian and Shi, Peng and Liu, Ming},
  journal={IEEE Transactions on Automatic Control},
  volume={57},
  number={7},
  pages={1809--1815},
  year={2011},
  publisher={IEEE}
}

@inproceedings{sayedana2023relative,
  title={Relative almost sure regret bounds for certainty equivalence control of markov jump systems},
  author={Sayedana, Borna and Afshari, Mohammad and Caines, Peter E and Mahajan, Aditya},
  booktitle={2023 62nd IEEE Conference on Decision and Control (CDC)},
  pages={6629--6634},
  year={2023},
  organization={IEEE}
}

@article{sayedana2024strong,
  title={Strong consistency and rate of convergence of switched least squares system identification for autonomous markov jump linear systems},
  author={Sayedana, Borna and Afshari, Mohammad and Caines, Peter E and Mahajan, Aditya},
  journal={IEEE Transactions on Automatic Control},
  year={2024},
  publisher={IEEE}
}

@article{chekan2024any,
  title={Any-Time Regret-Guaranteed Algorithm for Control of Linear Quadratic Systems},
  author={Chekan, Jafar Abbaszadeh and Langbort, Cedric},
  journal={arXiv preprint arXiv:2406.07746},
  year={2024}
}

@article{gurvits2009np,
  title={On the NP-hardness of checking matrix polytope stability and continuous-time switching stability},
  author={Gurvits, Leonid and Olshevsky, Alexander},
  journal={IEEE transactions on automatic control},
  volume={54},
  number={2},
  pages={337--341},
  year={2009},
  publisher={IEEE}
}

@article{joshi2025ai,
  title={AI as an intervention: improving clinical outcomes relies on a causal approach to AI development and validation},
  author={Joshi, Shalmali and Urteaga, I{\~n}igo and Van Amsterdam, Wouter AC and Hripcsak, George and Elias, Pierre and Recht, Benjamin and Elhadad, No{\'e}mie and Fackler, James and Sendak, Mark P and Wiens, Jenna and others},
  journal={Journal of the American Medical Informatics Association},
  volume={32},
  number={3},
  pages={589--594},
  year={2025},
  publisher={Oxford University Press}
}

@article{liang2025randomization,
  title={Randomization inference when n equals one},
  author={Liang, Tengyuan and Recht, Benjamin},
  journal={Biometrika},
  volume={112},
  number={2},
  pages={asaf013},
  year={2025},
  publisher={Oxford University Press}
}
\bibliographystyle{tmlr}

 \newpage
\appendix

\section{Preliminaries and Supplementary Results}

In this section, we provide preliminary results and supplementary materials to support the theoretical analysis presented in this paper.

\begin{lemma}[Lemma 25 of \cite{cohen2019learning}]
\label{lem:komaki}
Let $X$ and $Z$ be matrices of equal size, and let $Y$ be a $(\kappa,\gamma)$-stable matrix. If $X \preceq Y^\top X Y + Z$, then $X \preceq \frac{\kappa^2}{\gamma}\|Z\| I$.
\end{lemma}

\subsection*{Primal-Dual SDP formulation of LQR}
 \label{sec:primaldualSDP}

For a system $\Theta_*$ with quadratic cost matrices $Q$ and $R$, the infinite-horizon LQR problem can be reformulated as the following semi-definite program (SDP) (see \cite{cohen2018online}):

\begin{align}
\nonumber\textrm{minimize}\; \; \; \begin{pmatrix}
Q & 0 \\
0 & R
\end{pmatrix}\bullet \Sigma\\
\nonumber\textrm{S.t.}\; \; \; \Sigma_{xx}={\Theta_*}^\top\Sigma\Theta_*+W\\
\Sigma\succ 0\label{eq:SDPKhali}
\end{align}
Here, $\Sigma$ denotes the steady-state covariance matrix of the state-action pair $(x,u)$, partitioned as
\begin{align*}
\Sigma= \begin{pmatrix}
\Sigma_{xx} & \Sigma_{xu} \\
\Sigma_{ux} & \Sigma_{uu}
\end{pmatrix},
\end{align*}
where $\Sigma_{xx} \in  \mathbb{R}^{n \times n}$, $\Sigma_{uu} \in  \mathbb{R}^{m \times m}$, and $\Sigma_{ux}=\Sigma_{xu} \in  \mathbb{R}^{m \times n}$. 

Let $\Sigma(\Theta_*)$ denote the optimal solution of the SDP. The corresponding optimal feedback gain is given by
\begin{align*}
K_*(\Theta_*)
=
\Sigma_{ux}(\Theta_*)
\Sigma_{xx}^{-1}(\Theta_*).
\end{align*}

Since $W\succ0$, it follows that $\Sigma_{xx}(\Theta_*)\succ0$, and hence $\Sigma_{xx}^{-1}(\Theta_*)$ exists.

Moreover,, by applying the stabilizing policy $K_*(\Theta_*)$, the closed-loop state converges to a stationary distribution with covariance $X=\mathbb{E}[xx^T]$. In this case, the matrix
\begin{align*}
\mathcal{E}(K(\Theta_*))= \begin{pmatrix}
X & XK^\top (\Theta_*)\\
K(\Theta_*)X & K(\Theta_*)XK^\top(\Theta_*)
\end{pmatrix}
\end{align*}
constitutes a feasible solution to the SDP.

The dual of this program is written as follows:

\begin{align}
\begin{array}{rrclcl}
\displaystyle \max & \multicolumn{1}{l}{P\bullet W}\\
\textrm{s.t.} & \begin{pmatrix}
Q-P & 0 \\
0 & R
\end{pmatrix}+\Theta_* P {\Theta_*}^\top=0\\
&P\succeq 0 
\end{array}.\label{SDPDualkhali}
\end{align}
Let $P(\Theta_*)$ denote an optimal solution of the dual program. The corresponding optimal infinite-horizon average expected cost is defined as
$J_*=P(\Theta_*) \bullet W$. Under standard stabilizability and detectability assumptions, strong duality holds between the primal and dual SDPs. Consequently, the optimal values of the two programs coincide:
\begin{align*}
    \begin{pmatrix}
Q & 0 \\
0 & R
\end{pmatrix}\bullet \Sigma(\Theta_*)=P(\Theta_*) \bullet W.
\end{align*}

\newpage
\section{Omitted Proofs}
\label{partD}
In this appendix, we delve deeper into the proposed algorithm and provide a rigorous analysis, along with proofs of the main results and regret bounds.

\subsection{Proofs of Section \ref{eq:knownsetting}}

\begin{proof}[{\bf Proof of Theorem \ref{thm:alfabet_K_S}}]
  Let the system operating in mode $j$, under a stabilizing policy $K^{j}$, admit the candidate Lyapunov function
	\begin{align*}
	\mathcal{V}_{j}(t)=x_t^\top P_{K^{j}}x_t
	\end{align*}
where $P_{K^{j}}$ satisfies the Lyapunov equation
    
   \begin{align}
     P_{K^j}=Q^j+K^{j^\top} R^jK^j+(A^j_*+B^j_*K^j)^\top P_{K^j}(A^j_*+B^j_*K^j). \label{eq:LyapunSta}
 \end{align}  
    
By the Rayleigh--Ritz inequality, we obtain
 
	\begin{align}
\underline{\lambda}\big(P_{K^j}\big)\mathbb{E}[x_t^\top x_t| \mathcal{F}_{t-1}]\leq \mathbb{E}[\mathcal{V}_j(t)| \mathcal{F}_{t-1}]\leq \overline{\lambda}\big(P_{K^j}\big)\mathbb{E}[x_t^\top x_t| \mathcal{F}_{t-1}].\label{Ray-Rit}
	\end{align}

We next compute the one-step evolution:

	\begin{align}
\mathbb{E}[\mathcal{V}_j(t+1)| \mathcal{F}_t]-\mathcal{V}_j(t)=& x^\top_t(A^j_*+B^j_*K^j)^\top P_{K^j} (A^i_*+B^i_*K^j)x_t+\mathbb{E}[\omega^\top_{t+1}P_{K^j}\omega_{t+1}|\mathcal{F}_t]-x^T_tP_{K^j}x_t. \label{eq:LyapunocIneq}
	\end{align}

Using (\ref{eq:LyapunSta}), this simplifies to  
\begin{align}
\mathbb{E}[\mathcal{V}_j(t+1)| \mathcal{F}_t]-\mathcal{V}_j(t)=-x_t^\top H_{K^j}x_t+\mathbb{E}[\omega^\top_{t+1}P_{K^j}\omega_{t+1}|\mathcal{F}_t]
 \label{eq:useful3}
\end{align}
where 
\begin{align*}
   H_{K^j}=Q^j+{K^j}^\top R^jK^j.
\end{align*}

Hence,
	\begin{align*}
&\mathbb{E}[\mathcal{V}_i(t+1)| \mathcal{F}_t]-\mathbb{E}[\mathcal{V}_i(t)| \mathcal{F}_{t}] \leq  -\underline{\lambda}\big(H_{K^j}\big)\mathbb{E}[x_t^\top x_t| \mathcal{F}_{t}]+\overline{\lambda}\big(P_{K^j}\big)\sigma_{\omega}^2,
	\end{align*}
	combining which with (\ref{Ray-Rit}) gives
	\begin{align}
\mathbb{E}[\mathcal{V}_i(t+1)| \mathcal{F}_t]\leq\big(1-\frac{\underline{\lambda}\big(H_{K^j}\big)}{\overline{\lambda}\big(P_{K^j}\big)}\big)\mathbb{E}[\mathcal{V}_i(t)| \mathcal{F}_{t}]+\overline{\lambda}\big(P_{K^j}\big)\sigma_{\omega}^2. \label{eq:1it}
	\end{align}
Similarly, for time step $t+2$,
	\begin{align*}
\mathbb{E}[\mathcal{V}_i(t+2)| \mathcal{F}_{t+1}]\leq\big(1-\frac{\underline{\lambda}\big(H_{K^j}\big)}{\overline{\lambda}\big(P_{K^j}\big)}\big)\mathbb{E}[\mathcal{V}_i(t+1)| \mathcal{F}_{t+1}]+\overline{\lambda}\big(P_{K^j}\big)\sigma_{\omega}^.
	\end{align*}
Using the tower property,
\[
\mathbb{E}\big[\mathbb{E}[\mathcal{V}_j(t+1)\mid \mathcal{F}_{t+1}] \mid \mathcal{F}_{t}\big]
= \mathbb{E}[\mathcal{V}_j(t+1) \mid \mathcal{F}_{t}]
\]

we obtain

	\begin{align*}
\mathbb{E}[\mathcal{V}_i(t+2)| \mathcal{F}_{t}]\leq\big(1-\frac{\underline{\lambda}\big(H_{K^j}\big)}{\overline{\lambda}\big(P_{K^j}\big)}\big)\mathbb{E}[\mathcal{V}_i(t+1)| \mathcal{F}_{t}]+\overline{\lambda}\big(P_{K^j}\big)\sigma_{\omega}^2.
	\end{align*}
Iterating (\ref{eq:1it}) yields

\begin{align}
\nonumber \mathbb{E}[\mathcal{V}_i(t+2)| \mathcal{F}_{t}]\leq &\big(1-\frac{\underline{\lambda}\big(H_{K^j}\big)}{\overline{\lambda}\big(P_{K^j}\big)}\big)^2\mathbb{E}[\mathcal{V}_i(t)| \mathcal{F}_{t}]+\big(1-\frac{\underline{\lambda}\big(H_{K^j}\big)}{\overline{\lambda}\big(P_{K^j}\big)}\big) \overline{\lambda}\big(P_{K^j}\big)\sigma_{\omega}^2+\overline{\lambda}\big(P_{K^j}\big)\sigma_{\omega}^2.
\end{align}

Suppose the system starts operating in mode $i_k$ at time $t_k$, and let $t_{k+1}$ denote an instant the next switch to mode $i_{k+1}$. Then, by induction, we have

\begin{align}
\mathbb{E}[\mathcal{V}_{i_k}(t_{k+1})| \mathcal{F}_{t_k}]\leq &\bigg(1-\frac{\underline{\lambda}\big(H_{K^{i_k}}\big)}{\overline{\lambda}\big(P_{K^{i_k}}\big)}\bigg)^{t_{k+1}-t_k}\mathbb{E}[\mathcal{V}_{i_k}(t_k)| \mathcal{F}_{t_k}]+\overline{\lambda}\big(P_{K^{i_k}}\big)\sigma_{\omega}^2\sum_{j=0}^{t_{k+1}-t_k-1} \bigg(1-\frac{\underline{\lambda}\big(H_{K^{i_k}}\big)}{\overline{\lambda}\big(P_{K^{i_k}}\big)}\bigg)^{j}.\label{eq:jbsly}
\end{align}
where $\mathcal{V}_{i_k}(t_{k+1})$ denotes the candidate Lyapunov function evaluated immediately before the switching instant $t_{k+1}$, while the system is still operating in mode $i_k$.

On the other hand, we have 

  	\begin{align}
\nonumber \mathbb{E}[\mathcal{V}_{i_k}(t_{k+1})| \mathcal{F}_{t_k}] \geq  \underline{\lambda}\big(P_{K^{i_k}}\big)\mathbb{E}[x^{-^\top}_{t_{k+1}} x^-_{t_{k+1}}| \mathcal{F}_{t_k}]&=\underline{\lambda}\big(P_{K^{i_k}}\big)\mathbb{E}[x_{t_{k+1}}^\top x_{t_{k+1}}| \mathcal{F}_{t_k}]\\
&\geq \frac{\underline{\lambda}\big(P_{K^{i_k}}\big)}{\overline{\lambda}\big(P_{K^{i_{k+1}}}\big)}\mathbb{E}[\mathcal{V}_{i_{k+1}}(t_{k+1})| \mathcal{F}_{t_{k}}]. \label{eq:limineq}
	\end{align}
where $x^-_{t_{k+1}}$ denotes the state immediately before the switching instant $t_{k+1}$, which, by continuity of the state trajectory, satisfies $x^-_{t_{k+1}} = x_{t_{k+1}}$. Moreover, $\mathcal{V}_{i_{k+1}}(t_{k+1})$ denotes the candidate Lyapunov function evaluated immediately after the switching instant $t_{k+1}$, when the system starts operating in mode $i_{k+1}$.

Define
\begin{align*}
  \bar{\eta}(K^{i_k}):= \frac{\underline{\lambda}\big(H_{K^{i_{k}}}\big)}{\overline{\lambda}\big(P_{K^{i_{k}}}\big)}.
\end{align*}

Applying (\ref{eq:limineq}) to (\ref{eq:jbsly}) yields
    \begin{align*}
\nonumber \mathbb{E}[\mathcal{V}_{i_{k+1}}(t_{k+1})| \mathcal{F}_{t_k}]\leq & \frac{\overline{\lambda}\big(P_{K^{i_{k+1}}}\big)}{\underline{\lambda}\big(P_{K^{i_{k}}}\big)} \bigg(1-\bar{\eta}(K^{i_k})\bigg)^{t_{k+1}-t_k}\mathbb{E}[\mathcal{V}_{i_k}(t_k)| \mathcal{F}_{t_k}]\\
&+\overline{\lambda}\big(P_{K^{i_k}}\big)\sigma_{\omega}^2  \frac{\overline{\lambda}\big(P_{K^{i_{k+1}}}\big)}{\underline{\lambda}\big(P_{K^{i_{k}}}\big)}\sum_{j=0}^{t_{k+1}-t_k-1} \bigg(1-\bar{\eta}(K^{i_k})\bigg)^{j} .
\end{align*}

Applying (\ref{Ray-Rit}) again and defining
\begin{align*}
     \bar{\rho} (K^{i_k}) :=\frac{\overline{\lambda}\big(P_{K^{i_{k+1}}}\big)}{\underline{\lambda}\big(P_{K^{i_{k}}}\big)}, \quad  \bar{\mathcal{X}}(K^{i_{k+1}}):= \frac{\overline{\lambda}\big(P_{K^{i_k}}\big)}{\underline{\lambda}\big(P_{K^{i_{k+1}}})},
\end{align*}
we obtain
\begin{align}\label{eq:ineqcrudbet}
    \mathbb{E}[x^\top _{t_{k+1}}x_{t_{k+1}}| \mathcal{F}_{t_k}]&\leq \bar{\mathcal{X}}(K^{i_{k+1}}) \bar{\rho}(K^{i_k})\bigg(1-\bar{\eta} \big(K^{i_k})\bigg)^{t_{k+1}-t_k}  \mathbb{E}[x^\top _{t_k}x_{t_k}| \mathcal{F}_{t_k}]\cr
    &\qquad+\frac{\bar{\mathcal{X}}(K^{i_{k+1}}) \bar{\rho}(K^{i_k})}{\bar{\eta} \big(K^{i_k})} \sigma_{\omega}^2. 
\end{align}

Let $t_{k+1} - t_k$ denote the duration spent in mode $i_k$, and let it be treated as a decision variable. We aim to specify this duration such that

\begin{align*}
    \bar{\mathcal{X}}(K^{i_{k+1}}) \bar{\rho}(K^{i_k})\bigg(1-\bar{\eta} \big(K^{i_k})\bigg)^{t_{k+1}-t_k} \leq \bar{\alpha},
\end{align*}
where $\bar{\alpha}$ is a user-defined parameter. Equivalently,
\begin{align*}
     & \ln \bar{\rho}(K_{i_k})+ (t_{k+1}-t_k) \ln \bigg(1-\bar{\eta} \big(K_{i_k}\big)\bigg) +\ln \bar{\mathcal{X}}(K_{i_{k+1}})\leq  \ln \bar{\alpha}.
\end{align*}
Thus, a sufficient lower bound on the dwell time is
\begin{align}
    t_{k+1}-t_k\geq -\frac{\ln \bar{\rho}(K_{i_k})+\ln \mathcal{X}(K_{i_{k+1}})-\ln \bar{\alpha}}{\ln \big(1-\bar{\eta} \big(K_{i_k})\big)}. \label{eq:cond}
\end{align}

\end{proof}

\begin{proof}[{\bf Proof of Lemma  \ref{lem:betaDef}}]
Consider (\ref{eq:LyapunSta}). By Definition \ref{def:sequentially}, the policy $K^i$ is $(\kappa_c^i,\gamma_c^i)$-strongly stabilizing. Therefore, applying Lemma~\ref{lem:komaki} yields
\begin{align*}
    \|P_{K^{j}}\|\leq \alpha^j_1 (1+\kappa_c^{j^2})\frac{\kappa_c^{j^2}}{\gamma_c^j}.
\end{align*}
For simplicity of notation, suppose the switch occurs from mode $i$ to mode $j$. Then,
 \begin{align*}
     \frac{1}{\bar{\eta}\big(K^{i})}&= \frac{\overline{\lambda}\big(P_{K^{i}}\big)}{\underline{\lambda}\big(H_{K^{i}}\big)}\leq \frac{\alpha^i_1}{\alpha^i_0} (1+\kappa_c^{i^2})\frac{\kappa_c^{i^2}}{\gamma_c^i}\\
     \bar{\mathcal{X}}(K^{j})&\leq  \kappa_c^{j^2}\\
     \bar{\rho}(K^{i})&\leq \kappa_c^{i^2}.
 \end{align*}
 Define $\kappa_c^{*} := \max_i \kappa_c^{i}$, $\gamma_c^{*} := \min_i \gamma_c^{i}$, $\alpha_1^{*} := \max_i \alpha_1^{i}$, and $\alpha_0^{*} := \min_i \alpha_0^{i}$. Then, using (\ref{eq:ineqcrudbet}), the term $\bar{\beta}$, we obtain
 \begin{align*}
     \frac{\bar{\mathcal{X}}(K^{i_{k+1}}) \bar{\rho}(K^{i_k})}{\bar{\eta} \big(K^{i_k})}\leq \frac{\alpha^{*}_1}{\alpha^{*}_0} (1+\kappa_c^{*^2})\;\frac{\kappa_c^{*^6}}{\gamma_c^{*}}=:\bar{\beta}
 \end{align*}  
 which completes the proof.
\end{proof}

% \textcolor{red}{Before, describing the second goal, we need to find an upper-bound on the duration of the posted problem, given number of switches $n_s$. The following lemma gives an upper bound on duration of an epoch.} 

% The second goal directly corresponds to the regret of Algorithm \ref{alg:OSL3}. We need the initial estimate $\|\Theta_0^i-\Theta_*^i\|\leq \bar{\epsilon}_i$ such that $\bar{\epsilon}_i\bar{\lambda}^i_{max}<1$ where $\bar{\lambda}^i_{max}=4\bar{\mu}^i_{max}\nu_i/\alpha_0^i\sigma_{\omega}^2$ and

% \begin{align*}
%      \big(1+2\vartheta_i (n\tau_{max}^{dw}+n\tau_{max}^{dw}(1+{\kappa_0^i}^2)^2\bar{X}^2_{max})^{0.5} \big):=\bar{\mu}_{max}^i
% \end{align*}

\subsection{Stability Analysis}
\label{perturbationLemma_Stab}

The proof of Theorem \ref{Stability_thm17} follows similar arguments to those used in the proof of Theorem 3 in \cite{chekan2024any}, with the necessary modifications to account for the adjusted scaling of the input perturbation noise covariance matrix. We first establish the following lemma, which provides a lower bound on the minimum eigenvalue of the covariance matrix.

\begin{lemma}\label{lem:minlwerbndei}
    Fix $\delta\in (0,\,1)$ and let the designed feedback to be perturbed with an additive noise $\eta_t \sim \mathcal{N}\big(0, \Gamma_t\big)$ where 
    \begin{align*}
    \Gamma_t := 2\frac{\bar{p}\bar{\kappa}^2 \sigma_\omega^2}{\big(t + \bar{c}\big)^{\zeta}},
\end{align*}
 and $\zeta \in (0,\, 1/2)$. Then with probability at least $1-\delta$ we have
    \begin{align}
        \lambda_{min}\big(\sum_{k=1}^t z_kz_k^\top\big)\geq \frac{\sigma_{\omega}^2 \bar{p} }{80} (t+\bar{c})^{1-\zeta}- \underbrace{\frac{\sigma_{\omega}^2\bar{c}}{80}}_{=:\bar{C}} \label{eq:lowbndVmin}
    \end{align}
    for $t\geq 400(n+m+\log t/\delta)$.
\end{lemma}
\begin{proof}
The detailed proof is given in \cite{chekan2024any} for the special case \(\zeta=\frac{1}{2}\). We sketch the modifications required for a general \(\zeta\in(0,\frac{1}{2})\).

By following the same arguments as in Lemma 9 of \cite{chekan2024any}, and choosing $\bar{c}$ according to (\ref{eq:magncbar}), we obtain
\begin{align*}
    2 \bar{\kappa}^2 \bar{\vartheta}_{B_*}^2 \frac{\bar{p}}{(t+\bar{c})^{\zeta}} \leq 1,
\end{align*}
for all $t\ge 1$. Since $\bar{\kappa}, \bar{\vartheta}_{B_*}\geq 1$, it follows that
\begin{align}
    \frac{(t+\bar{c})^{\zeta}}{ \bar{p}}-1\geq 0. \label{eq:useful092}
\end{align}
Using the same derivation as in Lemma 9 of \cite{chekan2024any}, (\ref{eq:useful092}) implies
\begin{align*}
    \mathbb{E}[z_t z_t^\top|\mathcal{F}_{t-1}]\succeq \frac{\sigma_{\omega}^2\bar{p}}{2(t+\bar{c})^{\zeta}} I \quad \quad   \forall t.
\end{align*}

Next, applying the arguments of Lemmas 10 and 11 in \cite{chekan2024any}, together with (\ref{eq:useful092}), yields
  \begin{align}
        a^\top \big(\sum_{k=1}^{t} z_k z_k^\top\big) a \geq \frac{\sigma_{\omega}^2 \bar{p}}{40}(t+\bar{c})^{1-\zeta}- {\frac{\sigma_{\omega}^2\bar{c}}{40}} \label{eq:kmk092}
    \end{align}
    for every unit vector $a\in\mathbb R^{n+m}$, with probability at least $1-\delta$, provided that $t\geq 200 \log t/\delta$.

Finally, repeating the covering argument used in Lemma 12 of \cite{chekan2024any}, we obtain
\begin{align}
        \lambda_{min}\big(\sum_{k=1}^t z_kz_k^\top\big)\geq \frac{\sigma_{\omega}^2 \bar{p}}{80}(t+\bar{c})^{1-\zeta}- \underbrace{\frac{\sigma_{\omega}^2\bar{c}}{80}}_{\bar{C}}\label{eq:resLem10}
    \end{align}
   uniformly in time for $t\geq 400(n+m+\log t/\delta)$.
\end{proof}

% The proof directly follows by Lemma 18 in \cite{cohen2019learning}. However for the sake of completeness and avoiding possible misunderstandings due to the notation difference we provide it here. The proof also includes the definition of $\lambda_i$'s and upper-bound of $\mu_i$'s mentioned in Theorem \ref{Stability_thm17}.

\begin{proof}[{\bf Proof of Theorem \ref{Stability_thm17}}] 
Referring to equation (81) in \cite{chekan2024any}, we have
\begin{align}
\nonumber P(\mathcal{C}_{n_i(t)})\succeq & \ Q^i+K^\top (\mathcal{C}_{n_i(t)})R^i K(\mathcal{C}_{n_i(t)})+(A_*+B_*^iK(\mathcal{C}_{n_i(t)}))^\top P(\mathcal{C}_{n_i(t)})(A_*+B_*^iK(\mathcal{C}_{n_i(t)}))\\
&-2\mu_{n_i(t)}\|P(\mathcal{C}_{n_i(t)})\|\begin{pmatrix} I \\ K(\mathcal{C}_{n_i(t)}) \end{pmatrix}{V}^{-1}_{n_i(t)}\begin{pmatrix} I \\ K(\mathcal{C}_{n_i(t)}) \end{pmatrix}^\top
\label{eq:Lemma8_res22}
\end{align}
where
\begin{align*}
\mu_{n_i(t)}= r_{n_i(t)}+\sqrt{r_{n_i(t)}}\vartheta_i \|V_{n_i(t)}\|^{1/2}.
\end{align*}

The above inequality follows from the complementary slackness conditions of the relaxed primal--dual SDP pair \eqref{eq:RelaxedSDP}--\eqref{eq:obtPol}, together with the perturbation result of Lemma 7 in \cite{chekan2024any}. For further details, we refer the reader to \cite{cohen2019learning,chekan2024any}.

To establish that all generated policies are stabilizing, we tune the regularization parameter $\lambda_i$, the perturbation noise covariance matrix $\Gamma_t^i$, and the initial estimation accuracy $\epsilon(\bar{\kappa}_i)$ such that
\begin{align}
    \mu_{n_i(t)}\|P(\mathcal{C}_{n_i(t)})\| \,{V}^{-1}_{n_i(t)} \preceq  \frac{\alpha^i_0}{4}I 
    \label{eq:juststab0}
\end{align}
holds. A sufficient condition for (\ref{eq:juststab0}) is
\begin{align}
   \mu_{n_i(t)}\|P(\mathcal{C}_{n_i(t)})\| \,\|{V}^{-1}_{n_i(t)}\| \leq \frac{\alpha^i_0}{4}. \label{eq:juststab02}
\end{align}

Now using the fact that
 $\bar{\kappa}\geq\sqrt{\frac{2\|P(\mathcal{C}_{n_i(t)})\|}{\alpha^i_0}}>1$ the sufficient condition (\ref{eq:juststab02}) can be rewritten as
\begin{align}
    \mu_{n_i(t)} \,\,\|V_{n_i(t)}^{-1}\| \preceq 
    \frac{1}{2 \bar{\kappa}_i^{2}}I. 
    \label{eq:verygood01}
\end{align}

We next specify $\lambda_i$, $\epsilon(\bar{\kappa}_i)$, and $\bar{p}_t^i$ so that (\ref{eq:verygood01}) holds. Using the definition of $\mu_{n_i(t)}$ and applying Lemma \ref{lem:minlwerbndei}, a sufficient condition for \eqref{eq:verygood01} is that the following inequalities hold:
\begin{align}
   \frac{2\vartheta \sqrt{r_{n_i(t)}}\|V_{n_i(t)}\|^{\frac{1}{2}}}{\frac{\sigma_{\omega}^2 \bar{p}^i_t}{80}(n_i(t)+\bar{c}_i)^{1-\zeta}-\bar{C}_i+\lambda_i} \leq\frac{1}{4 \bar{\kappa}_i^{2}} \label{eq:pcsdg0}
\end{align}
and
\begin{align}
     \frac{r_{n_i(t)}}{\frac{\sigma_{
\omega}^2 \bar{p}^i_t}{80}(n_i(t)+\bar{c}_i)^{1-\zeta}-\bar{C}_i+\lambda}\leq \frac{1}{4 \bar{\kappa}_i^{2}.} \label{eq:pcsdg+}
\end{align}
 We choose $\lambda_i$ such that
\begin{align*}
    \lambda_i\geq \bar{C}_i.
\end{align*}
To ensure that both inequalities above are satisfied, it suffices to require
\begin{align*}
    \bar{p}^i_t\geq \max\bigg\{\frac{640 \bar{\kappa}_i^2\vartheta_i\sqrt{r_{n_i(t)}}\|V_{n_i(t)}\|^{1/2}}{\sigma_{\omega}^2(n_i(t)+\bar{c}_i)^{1-\zeta}},\;\frac{320\,r_{n_i(t)}\,\bar{\kappa}_i^2}{\sigma_{\omega}^2(n_i(t)+\bar{c}_i)^{1-\zeta}}\bigg\}.
\end{align*}
It is straightforward to verify that the first term in the maximum dominates the second.

Noting that by Lemma \ref{lem:defCbarz}
\begin{align}
    \sum_{k=1}^{n_i(t)} \|z^i_k z_k^{i^\top}\|^2\leq  c^2_z n_i(t)\log \frac{t}{\delta} \label{eq:Lam_1nskom}
\end{align}
for any $i\in \mathcal{M}$.

Now we establish the following inequalities. We have
\begin{align}
    \log \frac{d_x\det V_{n_i(t)}}{\delta\, \det(\lambda_i I)}
    &\le (d_x+d_u^i)\log\frac{d_x}{\delta}\Big(1 + n_i(t)\log \frac{t}{\delta}\Big) \label{eq:seqp}\\
    &\le 2(d_x+d_u^i)\log  \frac{2d_xt}{\delta}. \label{eq:seqpf1ns}
\end{align}
The inequality (\ref{eq:seqp}) holds under the choice
\begin{align}
    \lambda_i\geq c^2_z(\bar{\kappa}_*).\label{eq:Lam_1ns}
\end{align}

By choosing $\epsilon_i$ appropriately such that
\begin{align}
    \sqrt{\lambda_i} \epsilon_i \leq \sigma_{\omega}\sqrt{2d_x(d_x+d_u^i)} \label{eq:epslamb}
\end{align}
it then follows from (\ref{eq:seqpf1ns}) and (\ref{eq:epslamb}) that
\begin{align*}
    r_{n_i(t)}\leq 16\sigma_{\omega}^2d_x(d_x+d_u^i)\log \frac{2d_xt}{\delta}.
\end{align*}

Let us upper-bound the dominant term

\begin{align*}
   \bar{p}^i\leq  \frac{640 \bar{\kappa}_i^2\vartheta_i\sqrt{r_{n_i(t)}}\|V_{n_i(t)}\|^{1/2}}{\sigma_{\omega}^2(n_i(t)+\bar{c}_i)^{1-\zeta}}\leq \frac{5120 \bar{\kappa}_i^2\vartheta_i\sqrt{\sigma_{\omega}^2d_x (d_x+d_u^i)\log \frac{2n_i\,t}{\delta}}}{\sigma_{\omega}^2(n_i(t)+\bar{c}_i)^{1-\zeta}}\sqrt{\lambda_i}+\frac{5120 \bar{\kappa}_i^2\vartheta_i\sqrt{\sigma_{\omega}^2d_x (d_x+d_u^i)\,n_i(t)}}{\sigma_{\omega}^2(n_i(t)+\bar{c}_i)^{1-\zeta}}c_z\log \frac{t}{\delta}
\end{align*}
if 
\begin{align}
    (n_i(t)+\bar{c}_i)^{\frac{1}{2}-\zeta}\geq \log \frac{2d_xt}{\delta}\label{eq:stabCon}
\end{align}
then the stability condition will be fulfilled if we choose 

\begin{align}
    \bar{p}^i:=\frac{5120 \bar{\kappa}_i^2\vartheta_i\sqrt{\sigma_{\omega}^2d_x (d_x+d_u^i)}}{\sigma_{\omega}^2}(c_z+\sqrt{\lambda_i}) \label{eq:candbarptp}
\end{align}

Next, we specify $\bar{c}_i$ such that
\begin{align}    
B^i_*\frac{2\bar{\kappa}_i^2 \bar{p}^i\sigma_{\omega}^2}{(n_i(t)+\bar{c}_i)^{\zeta}} B_*^{i^\top}\preceq \sigma_{\omega}^2I.\label{eq:noiseNbig}
\end{align}
This condition ensures that the effect of the injected input perturbation noise on the closed-loop system does not exceed the magnitude of the process noise.

% A sufficient condition for (\ref{eq:noiseNbig}) is
%  \begin{align}
%      \frac{2  \, \bar{\kappa}_*^2 \, \bar{\vartheta}_{B^i_*}^2 \bar{p}^i_t}{({n_i(t)+\bar{c}_i})^{\zeta}}\leq 1 \label{eq:goodcbarfind}
%  \end{align}
%  where we define $\bar{\vartheta}_{B_*} := \max\{1,\vartheta_{B_*}\}$ and use $\|B_*\|\le \vartheta_{B_*}$.

% We now derive an upper bound on the left-hand side of (\ref{eq:goodcbarfind}). Using the above estimates, we obtain
% \begin{align}
%  \frac{2  \, \bar{\kappa}_*^2 \, \bar{\vartheta}_{B^i_*}^2 \bar{p}^i_t}{({n_i(t)+\bar{c}_i})^{\zeta}}\leq& \frac{5120 \bar{\vartheta}_{B_*}^{i^2}\bar{\kappa}_i^4\vartheta_i\sqrt{\sigma_{\omega}^2n_i (n_i+m_i)\log \frac{2n_i\,t}{\delta}}}{\sigma_{\omega}^2(n_i(t)+\bar{c}_i)}\sqrt{\lambda_i} \label{eq:tm10}\\
%   &+\frac{5120 \bar{\vartheta}_{B_*}^{i^2}\bar{\kappa}_i^4\vartheta_i\sqrt{\sigma_{\omega}^2n (n_i+m_i)\,n_i(t)}}{\sigma_{\omega}^2(n_i(t)+\bar{c}_i)}c_z(\bar{\kappa}_*)\log \frac{2n_i t}{\delta}.\label{eq:tm20}
% \end{align}
% We enforce the condition:
% \begin{align}
%     \frac{5120 \bar{\vartheta}_{B_*}^{i^2}\bar{\kappa}_i^4\vartheta_i\sqrt{\sigma_{\omega}^2n_i (n_i+m_i)}}{\sigma_{\omega}^2(n_i(t)+\bar{c}_i)^{1/2}}c_z(\bar{\kappa}_*)\log \frac{2n_it}{\delta}\leq \frac{1}{2}.\label{eq:helpingCond}
% \end{align}
This can be achieved by an appropriate choice of $\bar{c}_i$ to be
\begin{align}
    \bar{c}_i= (2\bar{\kappa}_i^2\bar{p}^i \bar{\vartheta}_{B_*^i} )^{\frac{1}{\zeta}}\label{eq:magncbar}
\end{align}

Recalling (\ref{eq:Lam_1ns}), which provides an additional condition for tuning $\lambda_i$, we select
\begin{align}
    \lambda_i = \max \left\{ \frac{\sigma_{\omega}^2 \bar{c}_i}{80},\, c_z^2 \right\} = \frac{\sigma_{\omega}^2 \bar{c}_i}{80},\label{eq:lamchosenns}
\end{align}
where the equality follows from the magnitude of $\bar{c}_i$ in (\ref{eq:magncbar}).
Recalling (\ref{eq:epslamb}), the corresponding choice of $\epsilon_i$ is
\begin{align}
    \epsilon_i\leq \underline{\epsilon}(\bar{\kappa}_*) := \frac{\sigma_\omega \sqrt{2 d_x (d_x+d_u^i)}}{\sqrt{\lambda_i}}.\label{eq:epscrudebnd0102}
\end{align}

If condition (\ref{eq:stabCon}) is satisfied and the parameters are chosen as described above, then the stability condition (\ref{eq:verygood01}) holds. Following the same reasoning as in Lemma 14 of \cite{chekan2024any}, it follows that any generated policy will be $(\bar{\kappa}_i,\bar{\gamma}_i)$-strongly stabilizing, where
\begin{align*}
\bar{\kappa}_i = \sqrt{\frac{2\nu_i}{\alpha^i_0}}, \quad \bar{\gamma}_i = \frac{1}{2 \bar{\kappa}_i^2}.
\end{align*} 

Finally, using a bootstrapping argument similar to that in the proof of Theorem~1 in \cite{chekan2024any}, we can show that the claim of the theorem holds with probability at least \(1-3\delta\). Specifically, this follows from a union bound over the following events:
\begin{align*}
     \mathcal{A}^i_{t}&:=\left\{\forall s = 1, \ldots, t,\quad \Theta_*^{i}\in \mathcal{C}_{n_{i}(s)}\right\}\\
      \mathcal{D}_t^i&:=\left\{\forall s = 1, \ldots, t,\quad \lambda_{min}(V^i_{n_i(s)})\geq \lambda_i+\frac{\sigma_{\omega}^2\bar{p}^i_s}{80}(n_i(s)+\bar{c}_i)^{1-\zeta}-\bar{C}_i \right\}\\
      \mathcal{G}_t^i&:=\{\forall s = 1, \ldots,n_i(t),\quad \max_{1\leq j\leq s} \| B^i_*\eta^i_j+\omega_{j+1}\|\leq \sqrt{20 \sigma_{\omega}^2d_x\log \frac{t}{\delta}}\}.
\end{align*}
Each of these events holds uniformly over time with probability at least \(1-\delta\). Therefore, by the union bound, their intersection holds uniformly over time with probability at least \(1-3\delta\), which establishes the claim.
\end{proof}

It is worth mentioning that the original proof in \cite{chekan2024any} considers the specific choice \(\zeta = 1/2\). However, as will be shown in the regret analysis section, establishing the regret bound in our setting requires choosing \(\zeta < 1/2\). Moreover, by combining the sufficient condition for (\ref{eq:juststab0}) with Lemma \ref{lem:minlwerbndei}, and using the same tuning parameters as in \cite{chekan2024any}, the sufficient condition for stability remains satisfied, though in a slightly conservative way.

\subsection{Minimum Dwell Time}

\begin{proof}[{\bf Proof of Theorem \ref{thm:minimum_average_dwell_revised}}]
The proof follows the same line of arguments as the proof of Theorem~\ref{thm:alfabet_K_S}.
For an epoch starting at time $t_k^a$ in mode $i_k$, consider the candidate Lyapunov function
\begin{align*}
	\mathcal{V}_{i_k}(t)=x_t^\top P(\mathcal{C}_{n_{i_k}(t_k^a)})x_t
\end{align*}
valid for $t_k^a\le t< t_{k+1}^a$, where $P(\mathcal{C}_{n_i(t_k^a)})$ satisfies
\begin{align}
\nonumber P(\mathcal{C}_{n_{i_k}(t_k^a)})\succeq & Q^{i_k}+K^\top (\mathcal{C}_{n_{i_k}(t_k^a)})R^{i_k} K(\mathcal{C}_{n_{i_k}(t_k^a)})+(A_*+B_*^{i_k}K(\mathcal{C}_{n_{i_k}(t_k^a)}))^\top P(\mathcal{C}_{n_{i_k}(t_k^a)})(A_*+B_*^{i_k}K(\mathcal{C}_{n_{i_k}(t_k^a)}))\\
&-2\mu_{n_{i_k}(t_k^a)}\|P(\mathcal{C}_{n_{i_k}(t_k^a)})\|\begin{pmatrix} I \\ K(\mathcal{C}_{n_{i_k}(t_k^a)}) \end{pmatrix}{V}^{-1}_{n_{i_k}(t_k^a)}\begin{pmatrix} I \\ K(\mathcal{C}_{n_{i_k}(t_k^a)}) \end{pmatrix}^\top.
\label{eq:Lemma8_re}
\end{align}
Since the policy remains fixed throughout each epoch, $P(\mathcal{C}_{n_i(t_k^a)})$ is constant over the interval
$[t_k^a,t_{k+1}^a]$. Consequently, the one-step Lyapunov drift analysis,
together with the Rayleigh--Ritz bounds and recursive conditioning
arguments used in the proof of Theorem~\ref{thm:alfabet_K_S},
applies verbatim.
The only modification is that the Lyapunov inequality is now given by
(\ref{eq:Lemma8_re}), which introduces the additional uncertainty term
associated with the confidence set
$\mathcal{C}_{n_i(t_k^a)}$.
Proceeding exactly as in the proof of
Theorem~\ref{thm:alfabet_K_S}, we obtain the corresponding recursive
bound.
Using the fact that the system must remain in each mode for at
least one time step yields the claim of
Theorem~\ref{thm:minimum_average_dwell_revised}. The entire analysis relies on the closed-loop stability of each mode. As established in Theorem~\ref{Stability_thm17}, this property holds for every mode with probability at least \(1-3\delta\). Therefore, by a union bound, the claim of this theorem holds with probability at least \(1-6\delta\).
\end{proof}
The following lemma establishes an upper bound on the length of any epoch. 

  \begin{lemma}\label{thm:minimum_average_dwell} 
The following statement holds:
\begin{align}
\tau^a_{k,k+1}\leq \frac{2\ln \kappa^2_*-\ln \bar{\alpha}}{\ln(1-\frac{1}{\bar{\kappa}^2_*})}=:\bar{\tau}_*
\label{eq:tau_a_slow} 
\end{align}
for any $k$ under implementation of Algorithm \ref{alg:OSL3}, where $\bar{\kappa}_*=\max_{i\in \mathcal{M}} \bar{\kappa}_i$.
 \end{lemma}
 \begin{proof}
 The proof is straightforward. Applying the bounds in (\ref{eq:adivetabnd}) and (\ref{eq:khoobeaali}) yields the desired upper bound.
 \end{proof}

The following lemma specifies how close successive policies are to one another and how far they are from $P(\Theta_*^i)$.

\begin{lemma}
\label{lem:tightnessSol}
    For any mode \(i_k\) and any \(t\in [t_k^a,\, t_{k+1}^a]\), the following holds with probability at least \(1-\delta\):
    \begin{align}
       P(\mathcal{C}_{n_{i_k}(t_k^a)})\preceq P(\Theta_*^{i_k})\preceq &P(\mathcal{C}_{n_{i_k}(t)})+ \chi_{t}^{i_k}I\label{eq:closness ofSolu}
    \end{align}
    where
\begin{align}
 \nonumber \chi_{t}^{i_k}= \frac{16\| P(\mathcal{C}_{n_{i_k}(t)})\|^2\mu_{n_{i_k}(t)}}{\alpha_0^2} \|P(\mathcal{C}_{n_{i_k}(t)})\| \|\begin{pmatrix}
I \\
K(\mathcal{C}_{n_{i_k}(t)})
\end{pmatrix}^\top {V}^{-1}_{n_{i_k}(t)} \begin{pmatrix}
I \\
K(\mathcal{C}_{n_{i_k}(t)}) 
\end{pmatrix} \|.
\end{align}
\end{lemma}
\begin{proof}
   The proof follows directly from the proof of Lemma 15 in \cite{chekan2024any}.
\end{proof}

The following lemma establishes a useful upper bound that will be instrumental in bounding $\chi_t^{i_k}$ and, consequently, in the regret analysis.

    \begin{lemma}\label{lem:useflboun}
With parameters $\bar{p}_t^i$, $\lambda_i$ and $\underline{\epsilon}(\bar{\kappa}_i)$ be chosen as in Theorem \ref{Stability_thm17}, so that the resulting closed-loop system operating in mode \(i\) is \((\bar{\kappa}_i,\bar{\gamma}_i)\)-strongly stabilizing. Then the following bound holds:
\begin{align*}
    \mu_{n_i(t)} \|V^{-1}_{n_i(t)}\|
    \leq
    \bar{D}_i(n_i(t)+\bar{c}_i)^{-\frac{1}{2}+\zeta} \log \frac{t}{\delta}
\end{align*}
for some problem-dependent constant \(\bar{D}\).
\end{lemma}

\begin{proof}
Given the tuned parameters \(\bar{p}_t^i\), \(\lambda_i\), and \(\underline{\epsilon}(\bar{\kappa}_i)\), the proof follows directly from the definition of \(\mu_{n_i(t)}\) together with the lower bound on the minimum eigenvalue of the covariance matrix established in Lemma~\ref{lem:minlwerbndei}.
\end{proof}

\subsection{Expected sate norm bound}

\begin{proof}[{\bf Proof of Lemma \ref{thm:res_sequential_stability_resp}}]
With the dwell-time design in Theorem~\ref{thm:alfabet_K_S} applied to (\ref{eq:ineqcrudbet}), we obtain the following inequality for an epoch:
\begin{align*}
    &\mathbb{E}[x^\top _{t^a_{k+1}}x_{t^a_{k+1}}| \mathcal{F}_{t^a_k}]\leq \bar{\alpha} \; \mathbb{E}[x^\top _{t^a_k}x_{t^a_k}| \mathcal{F}_{t^a_k}]+  \frac{\mathcal{X} (\mathcal{C}_{n_{i_{k+1}}(t_k^a)}) \rho (\mathcal{C}_{n_{i_k}(t_k^a)})}{\eta \big(\mathcal{C}_{n_{i_k}(t_k^a)}\big)} \sigma_{\omega}^2. 
\end{align*}

Recall \(\mathcal{H}(\mathcal{C}_{n_i(t)})\) defined in (\ref{eq:HDef3p}), then we obtain
\begin{align*}
   \mathcal{H}(\mathcal{C}_{n_{i_k}(t_k^a)})\geq \frac{\alpha_0^{i_k}}{2}.
\end{align*}
Consequently, we have

\begin{align}
  &\frac{1}{\eta \big(\mathcal{C}_{n_{i_k}(t_k^a)}\big)}=\frac{P\big(\mathcal{C}_{n_{i_k}(t_k^a)}\big)}{\mathcal{H} \big(\mathcal{C}_{n_{i_k}(t_k^a)}\big)}\leq \bar{\kappa}_{i_k}^2\leq \bar{\kappa}_*^2 \label{eq:adivetabnd}\\
   &  \mathcal{X} (\mathcal{C}_{n_{i_{k+1}}(t_k^a)}) \rho (\mathcal{C}_{n_{i_k}(t_k^a)})\leq \bar{\kappa}_*^4.
    \label{eq:khoobeaali}
\end{align}
Then, we can write

\begin{align*}
    \frac{\mathcal{X} (\mathcal{C}_{n_{i_{k+1}}(t_k^a)}) \rho (\mathcal{C}_{n_{i_k}(t_k^a)})}{\eta \big(\mathcal{C}_{n_{i_k}(t_k^a)}\big)}  \leq \bar{\kappa}_*^6=:\tilde{\beta}.
\end{align*}
Thus, for any epoch,
\begin{align*}
    &\mathbb{E}[x^\top _{t^a_{k+1}}x_{t^a_{k+1}}| \mathcal{F}_{t^a_k}]\leq \bar{\alpha} \; \mathbb{E}[x^\top _{t^a_k}x_{t^a_k}| \mathcal{F}_{t^a_k}]+ \tilde{\beta} \sigma_{\omega}^2
\end{align*}
with probability at least $1-6\delta$.

By assuming the initial state \(x_0\), applying the tower property of expectation, and iterating the recursion over \(k\), we obtain the following bound at the switching time \(t_k^a\):
\begin{align}
  \mathbb{E}[x^\top_{t_{k}^a}x_{t_{k}^a}]\leq & \bar{\alpha}^{k} \|x_0\|^2 +\tilde{\beta} \sigma_{\omega}^2\sum_{i=0}^{k} \bar{\alpha}^i\leq \bar{\alpha}^{k} \|x_0\|^2 +\frac{\tilde{\beta}}{1-\bar{\alpha}}\sigma_{\omega}^2\leq \|x_0\|^2+\frac{\tilde{\beta}}{1-\bar{\alpha}}\sigma_{\omega}^2
    \label{eq:onemorest}
\end{align}
which holds with probability at least \(1-3|\mathcal{M}|\delta\) by a union bound over all modes.
\end{proof}

\begin{lemma}\label{lem:defCbarz}
The following bound holds with probability at least $1-\delta$:
    \begin{align*}
       \sum_{k=1}^{n_i(t)} \|z^i_t\|^2\leq n_i(t) c_z^2\log \frac{t}{\delta}
    \end{align*}
    where
    \begin{align*}
        c_z^2=\frac{16 \tilde{g}^2(1+\kappa_*^2)}{(1-\bar{\alpha})^2}\bigg(\|x_1\|+\sqrt{20(d_x+d_u^*)\sigma_{\omega}^2}\bigg)^2,
    \end{align*}
$d_u^*:=\max_id_u^i$, $\bar{\kappa}_*=\max_i \bar{\kappa}_i$ and
\begin{align*}
    \tilde{g}
    :=
    \max\left\{
        2\kappa_*^4,
        \frac{\kappa_*^2}{\gamma_0^*}
    \right\}.
\end{align*}
\end{lemma}
\begin{proof}
We begin by recalling some preliminary results from \cite{chekan2024any}. We consider a single system $(A_*, B_*)$ with time-varying feedback policies $K_t$. Let $M=A_*+B_*K_t$ where the policies $K_t$ are $(\kappa,\gamma)$-strongly stabilizing. The closed-loop dynamics can then be written as
    \begin{align*}
        x_t=M_1x_0+\sum_{s=1}^{t-1}M_{s+1}(B_*\nu_{s}+\omega_{s+1}) 
    \end{align*}
   where 
   \begin{align*}
       M_t=I,\;\; M_{s}=M_{s+1}(A_*+B_*K_s)=\prod_{j=s}^{t-1}(A_*+B_*K_j), \; \forall 1\leq s\leq t-1. 
   \end{align*}
   Recalling Definition \ref{def:sequentially}, we may decompose $A_*+B_*K_j= H_j L_j H_j^{-1}$, then for all $1\leq s< t$ where $\|L_j\|<1$.
   Now consider an epoch of the system, starting at $s=t_k^a$ and ending at $t-1=t_{k+1}^a$, just before switching to mode $i_{k+1}$. 

For this epoch, we first bound $\|M_s\|$. Using submultiplicativity,
\begin{align*}
\|M_s\|
&\le
\left\|
(H_{t-1}L_{t-1}H_{t-1}^{-1})\cdots(H_s L_s H_s^{-1})
\right\| \\
&\le
\|H_{t-1}\|
\left\|
(L_{t-1}H_{t-1}^{-1}H_{t-2})\cdots(L_{s+1}H_{s+1}^{-1}H_s)
\right\|
\|L_s\|\,\|H_s^{-1}\| \\
&\le
\|H_{t-1}\|
\left(\prod_{j=s}^{t-1}\|L_j\|\right)
\left(\prod_{j=s+1}^{t-1}\|H_j^{-1}H_{j-1}\|\right)
\|H_s^{-1}\|.
\end{align*}

As there is no update in policy during actuation in mode $i_k$, using the telescoping structure of the $H_j$ terms, this further yields
\begin{align}
\|M_s\|
\le
\|H_{t-1}\|
\left(\prod_{j=s}^{t-1}\|L_j\|\right)
\|H_{t-1}^{-1}H_{t-2}\|
\|H_s^{-1}\|.\label{eq:khoobch}
\end{align}

In the statement above, we consider the fact that there is no update in the policy while actuating in the mode $i_k$ and the only switch is switch in the mode to $i_{k+1}$. We first establish a bound on the state norm within a single epoch, where one of the four feedback scenarios may occur: (L,L), (L,F), (F,L), or (F,F). Although we derive the bound only for the (L,L) case, the remaining cases follow by the same argument. Using these epoch-wise state norm bounds, we then derive a global upper bound on the state norm for an arbitrary evolution consisting of successive occurrences of the above feedback scenarios.

Let us first consider the (L,L) case. In the above statement, \begin{align*}
     &H_{t-1}=P^{1/2}(\mathcal{C}_{n_{i_{k+1}}(t_k^a)})\\
      &H_s=H_{t-2}=P^{1/2}(\mathcal{C}_{n_{i_{k}}(t_k^a)}).
 \end{align*}
Referring the reader to Appendix B.2 of \cite{chekan2024any} for the definition of \(L_j\), by definition
\begin{align*}
   L_j= P^{1/2}(\mathcal{C}_{n_{i_k}(t_k^a)})(A_*+B_*^{i_k}K(\mathcal{C}_{n_{i_k}(t_k^a)}))P^{-1/2}(\mathcal{C}_{n_{i_k}(t_k^a)})
\end{align*}
for $j=s,...,t-2$, during which the system actuates in mode $i_{k}$. To obtain an upper bound on $\|L_j\|$, we rearrange (\ref{eq:Lemma8_re}) as follows:
\begin{align}
\nonumber P(\mathcal{C}_{n_{i_k}(t_k^a)})-\mathcal{H}(\mathcal{C}_{n_{i_k}(t_k^a)})\succeq (A_*+B_*^{i_k}K(\mathcal{C}_{n_{i_k}(t_k^a)}))^\top P(\mathcal{C}_{n_{i_k}(t_k^a)})(A_*+B_*^{i_k}K(\mathcal{C}_{n_{i_k}(t_k^a)})).
\end{align}
Here, $\mathcal{H}(\mathcal{C}_{n_{i_k}(t_k^a)})$ is defined in (\ref{eq:HDef3p}). Equivalently, we can write
\begin{align}
\nonumber &P^{-1/2}(\mathcal{C}_{n_{i_k}(t_k^a)})(A_*+B_*^{i_k}K(\mathcal{C}_{n_{i_k}(t_k^a)}))^\top P(\mathcal{C}_{n_{i_k}(t_k^a)})(A_*+B_*^{i_k}K(\mathcal{C}_{n_{i_k}(t_k^a)}))P^{-1/2}(\mathcal{C}_{n_{i_k}(t_k^a)})\preceq\\
\nonumber & I-P^{-1/2}(\mathcal{C}_{n_{i_k}(t_k^a)})\mathcal{H}(\mathcal{C}_{n_{i_k}(t_k^a)})P^{-1/2}(\mathcal{C}_{n_{i_k}(t_k^a)}).
\end{align}
Therefore,
\begin{align*}
    \|L_j^\top L_j\|\leq \|I-P^{-1/2}(\mathcal{C}_{n_{i_k}(t_k^a)})\mathcal{H}(\mathcal{C}_{n_{i_k}(t_k^a)})P^{-1/2}(\mathcal{C}_{n_{i_k}(t_k^a)})\| \leq\bigg(1-\eta(\mathcal{C}_{{i_k}(t_k^a)})\bigg)
\end{align*}
for $j=s,...,t-2$, which justifies (\ref{eq:Ljstate}), where $\eta(\mathcal{C}_{{i_k}(t_k^a)})$ is defined in (\ref{wq:etaDef}). Hence, we conclude that
  \begin{align}
      &\|L_j\|\leq\begin{cases}
\bigg(1-\eta(\mathcal{C}_{{i_k}(t_k^a)})\bigg)^{1/2}, & j=s,...,t-2 \\
\bigg(1-\eta(\mathcal{C}_{{i_{k+1}}(t_k^a)})\bigg)^{1/2}, & j=t-1
\end{cases}.  \label{eq:Ljstate}
  \end{align}
  We can now further upper-bound the right-hand side of (\ref{eq:khoobch}), which yields
  \begin{align*}
     \nonumber  \|M_s\|\leq &\overline{\lambda}\big(P^{1/2}(\mathcal{C}_{n_{i_{k+1}}(t_k^a)})\big)\bigg(1-\eta(\mathcal{C}_{{i_k}(t_k^a)})\bigg)^{\frac{1}{2}(t-s-1)}\,\bigg(1-\eta(\mathcal{C}_{{i_{k+1}}(t_k^a)})\bigg)^{\frac{1}{2}} \overline{\lambda}\big(P^{-1/2}(\mathcal{C}_{n_{i_{k+1}}(t_k^a)})\big)\\
    & \times \overline{\lambda}\big(P^{1/2}(\mathcal{C}_{n_{i_{k}}(t_k^a)})\big)\, \overline{\lambda}\big(P^{-1/2}(\mathcal{C}_{n_{i_{k}}(t_k^a)})\big).
  \end{align*} 
Noting that $P(\mathcal{C}_{n_{i_{k}}(t_k^a)})$ is positive definite, we have
\begin{align*}
      \overline{\lambda}\big(P^{1/2}(\mathcal{C}_{n_{i_{k}}(t_k^a)})\big)\overline{\lambda}\big(P^{-1/2}(\mathcal{C}_{n_{i_{k}}(t_k^a)})=&\bigg(\frac{\overline{\lambda}\big(P^{1/2}(\mathcal{C}_{n_{i_{k}}(t_k^a)})\big)}{\underline{\lambda}\big(P^{1/2}(\mathcal{C}_{n_{i_{k}}(t_k^a)})\big)}\bigg)^{1/2}
\end{align*}

We can derive a similar bound for $P(\mathcal{C}_{n_{i_{k+1}}(t_k^a)})$, which yields
  \begin{align}
     \nonumber  \|M_s\|&\leq \bigg(1-\eta(\mathcal{C}_{{i_k}(t_k^a)})\bigg)^{\frac{1}{2}(t-s-1)}\,\bigg(1-\eta(\mathcal{C}_{{i_{k+1}}(t_k^a)})\bigg)^{\frac{1}{2}} \rho^{1/2} (\mathcal{C}_{n_{i_k}(t_k^a)}, \mathcal{C}_{n_{i_{k+1}}(t_k^a)})\mathcal{X}^{1/2} (\mathcal{C}_{n_{i_k}(t_k^a)}, \mathcal{C}_{n_{i_{k+1}}(t_k^a)})\\
      &\leq  \bigg(1-\eta(\mathcal{C}_{{i_k}(t_k^a)})\bigg)^{\frac{1}{2}(t-s-1)}\, \rho^{1/2} (\mathcal{C}_{n_{i_k}(t_k^a)})\mathcal{X}^{1/2} (\mathcal{C}_{n_{i_{k+1}}(t_k^a)}).\label{eq:impsty}
  \end{align}
  where we used definitions in (\ref{eq:rhoDef}).

  The second inequality follows since we are considering the system immediately after a switch, and thus the first full contraction step under mode $i_{k+1}$ has not yet been realized. As a result, we drop the factor $\big(1-\eta(\mathcal{C}_{i_{k+1}}(t_k^a))\big)^{1/2}$ corresponding to that initial step. This is consistent with the dwell-time Lyapunov analysis used in the proof of Theorem \ref{thm:minimum_average_dwell_revised}.

  When $s=t_k^a$ and $t-1=t_{k+1}^a$ correspond to two consecutive switching instants, the minimum mode-dependent dwell-time condition implies
  \begin{align*}
     \|M_{t_k^a}\|\leq \sqrt{\bar{\alpha}}. 
  \end{align*}
For $s>t_k^a$, we obtain
  \begin{align*}
      \|M_s\|\leq \bar{\kappa}_*^2 \bigg(1-\eta(\mathcal{C}_{{i_k}(t_k^a)})\bigg)^{\frac{1}{2}(t-s-1)}.
  \end{align*}
which follows by using
\begin{align*}
    \frac{\alpha_0^{i_k}}{2}I\preceq P(\mathcal{C}_{n_{i_{k}}(t_k^a)}) \preceq\frac{\nu_i}{\sigma_{\omega}^2}I
\end{align*}
together with the definition of $\bar{\kappa}_{i_k}$, and an analogous argument for mode $i_{k+1}$. Denote $\bar{\kappa}_*:=\max_i \bar{\kappa}_i$.

Using the above bound, the state norm at the switching time can be upper-bounded as
    \begin{align}
        \nonumber \|x_{t_{k+1}^a}\|\leq& \|M_{t_k^a}\|\|x_{t_k^a}\|+\sum_{s=t_k^a}^{t_{k+1}^a-1}\|M_{s+1}\|\|B^{i_k}_*\eta_j+\omega_{s}\|\\
       \nonumber \leq &\sqrt{\bar{\alpha}} \|x_{t_k^a}\|+\bar{\kappa}^2_* \sum_{s=t_k^a}^{t_{k+1}^a-1}(1-\eta(\mathcal{C}_{{i_k}(t_k^a)}))^{(t_{k+1}^a-s)/2}  \|B_*\eta_s+\omega_s\| \\
  \nonumber  \leq & \sqrt{\bar{\alpha}}\, \|x_{t_k^a}\|+\frac{\bar{\kappa}_*^2}{\gamma*}  \max_{t^a_k\leq s< t_{k+1}^a-1}\|B^{i_k}_*\eta_{s}+\omega_{s}\| :=\sqrt{\bar{\alpha}}\, \|x_{t_k^a}\|+2\bar{\kappa}_*^4 \max_{t^a_k\leq s< t_{k+1}^a-1}\|B^{i_k}_*\eta_{s}+\omega_{s}\|
    \end{align}
where in the third inequality we applied (\ref{eq:adivetabnd}). 

We next consider the remaining combinations of controller types across two consecutive epochs. For the case (L,F), the same argument as in
(\ref{eq:impsty}) applies, with the term $\mathcal{X}(\mathcal{C}_{n_{i_{k+1}}(t_k^a)})$ replaced by the corresponding
quantity associated with the fallback controller,
\begin{align*}
 \bar{\mathcal{X}}(K_0^{i_{k+1}}):= \frac{\overline{\lambda}\big(P_{K_0^{i_{k+1}}}\big)}{\underline{\lambda}\big(P_{K_0^{i_{k+1}}})}\leq \kappa_0^*
\end{align*}
where $\kappa_0^*:=\max_i \kappa_0^i$. Consequently, we obtain
    \begin{align}
        \nonumber \|x_{t_{k+1}^a}\|\leq&\sqrt{\bar{\alpha}} \|x_{t_k^a}\|+\bar{\kappa}_* \kappa_0^* \sum_{s=t_k^a}^{t_{k+1}^a-1}(1-\eta(\mathcal{C}_{{i_k}(t_k^a)}))^{(t_{k+1}^a-s)/2}  \|B_*\eta_s+\omega_s\| \\
  \nonumber  \leq & \sqrt{\bar{\alpha}}\, \|x_{t_k^a}\|+2\bar{\kappa}_*^3 \kappa_0^*  \max_{t^a_k\leq s< t_{k+1}^a-1}\|B^{i_k}_*\eta_{s}+\omega_{s}\|. 
    \end{align}
    For the case (F,L), following an analogous induction argument, we obtain
   \begin{align}
        \nonumber \|x_{t_{k+1}^a}\|\leq&\sqrt{\bar{\alpha}}\, \|x_{t_k^a}\|+\frac{\bar{\kappa}_*\kappa_0^*}{\gamma_0^*}  \max_{t^a_k\leq s< t_{k+1}^a-1}\|B^{i_k}_*\eta_{s}+\omega_{s}\| 
    \end{align}
where $\gamma_0^*=\min_i \gamma_0^i$. 

Finally, for the case (F,F), both consecutive epochs employ fallback
controllers, and the same reasoning gives
  \begin{align}
        \nonumber \|x_{t_{k+1}^a}\|\leq&\sqrt{\bar{\alpha}}\, \|x_{t_k^a}\|+\frac{(\kappa_0^*)^2}{\gamma_0^*}  \max_{t^a_k\leq s< t_{k+1}^a-1}\|B^{i_k}_*\eta_{s}+\omega_{s}\|. 
    \end{align}

Therefore, regardless of the controller pattern across consecutive epochs, the state satisfies
  \begin{align}
        \nonumber \|x_{t_{k+1}^a}\|\leq&\sqrt{\bar{\alpha}}\, \|x_{t_k^a}\|+\tilde{g}  \max_{t^a_k\leq s< t_{k+1}^a-1}\|B^{i_k}_*\eta_{s}+\omega_{s}\|. 
    \end{align}
where
\begin{align*}
    \tilde{g}=\max \{2\bar{\kappa}_*^4, 2\bar{\kappa}_*^3 \kappa_0^*, \frac{\bar{\kappa}_*\kappa_0^*}{\gamma_0^*}, \frac{(\kappa_0^*)^2}{\gamma_0^*}\}.
\end{align*}
Next, recall that the algorithm selects the input perturbation covariance such
that $\|\Gamma_t^i\|\leq\sigma_\omega^2$. Therefore,
\begin{align*}  
B^{i}_*\frac{2\bar{\kappa}_i^2 \bar{p}^i_t\sigma_{\omega}^2}{(n_i(t)+\bar{c})^{\zeta}} B_*^{i^\top}\preceq \sigma_{\omega}^2I.
\end{align*}
Applying the Hanson--Wright inequality yields
\begin{align*}
   \max_{1\leq s\leq n_i(t)} \| B^i_*\eta^i_s+\omega_{s}\|\leq \sqrt{20 \sigma_{\omega}^2d_x\,\log \frac{n_i(t)}{\delta}} 
\end{align*}
with probability at least $1-\delta$.

By recursive (backward) propagation, regardless of the type pf patters, over the switching times, we obtain
\begin{align*}
    \|x_{t_{k+1}^a}\|\leq(\sqrt{\bar{\alpha}})^{k+1} \|x_1\|+2 \frac{\tilde{g}}{1-\bar{\alpha}}\sqrt{20 \sigma_{\omega}^2d_x \log \frac{t_{k+1}^a}{\delta}}.
\end{align*}
And the following bound holds for any time
\begin{align*}
    \|x_{t}\|\leq \|x_1\|+2 \frac{\tilde{g}}{1-\bar{\alpha}}\sqrt{20 \sigma_{\omega}^2d_x \log \frac{t}{\delta}}
\end{align*}
using which
\begin{align*}
    \|x_{t}\|^2\leq \underbrace{4\frac{\tilde{g}^2}{(1-\bar{\alpha})^2}\bigg(\|x_1\|+\sqrt{20d_x \sigma_{\omega}^2}\bigg)^2}_{=:c_x^2}\log \frac{t}{\delta}=:X_*^2.
\end{align*}

Recalling the structure of $z_t$, we obtain
  \begin{align*}
        \|z_t\|^2\leq & 2(1+\kappa_*^2)\|x_t\|^2+2 \|\eta^i_{t}\|^2.
    \end{align*}
    where $\kappa_*:=\max\{\bar{\kappa}_*,\kappa_0^*\}$
Furthermore,
\begin{align*}
    \max_{1\leq s\leq t} \|\eta^i_s\|\leq \sqrt{10d_u^i \sigma_{\omega}^2 \log \frac{t}{\delta}}.
\end{align*}
Combining the above estimates yields
\begin{align*}
        \|z_{t}\|^2\leq  \underbrace{\frac{16 \tilde{g}^2(1+\kappa_*^2)}{(1-\bar{\alpha})^2}\bigg(\|x_1\|+\sqrt{20(d_x+d_u^*)\sigma_{\omega}^2}\bigg)^2}_{=:c^2_z}\log \frac{t}{\delta}=:Z_*^2.
\end{align*}
where $\max_i d_u^i=:d_u^*$. With $\kappa_*:=\max\{\bar{\kappa}_*,\kappa_0^*\}$, we redefine
\begin{align*}
    \tilde{g}
    :=
    \max\left\{
        2\kappa_*^4,
        \frac{\kappa_*^2}{\gamma_0^*}
    \right\}.
\end{align*}
Although this choice leads to a potentially looser state-norm upper bound, it preserves validity while simplifying the notation.
  \end{proof}

\subsection{Minimum mode-dependent dwell-time estimation error}

In this subsection, we derive an upper bound on the dwell-time estimation error $\tau^{a}_{k,k+1}-\tau^*_{k,k+1}$, which is used in the regret analysis.

\begin{proof}[{\bf Proof of Theorem \ref{Thm:dwellTimeError}}]
From strong duality between relaxed primal and dual SDP formulations, we known that

\begin{align*}
  \begin{pmatrix}
	Q^i & 0 \\
	0 & R^i
	\end{pmatrix}\bullet \Sigma(\mathcal{C}_{n_i(t)})=P(\mathcal{C}_{n_i(t)})\bullet W=\sigma_{\omega}^2\operatorname{tr}(P(\mathcal{C}_{n_i(t)})).
\end{align*}
Similarly we have similar equality between the original primal and dual SDP formulations, (\ref{eq:SDPKhali}) and (\ref{SDPDualkhali}).

In addition by (\ref{eq:closness ofSolu}) we have
\begin{align}
    P(\Theta_*^{i})-P(\mathcal{C}_{n_i(t)})\preceq \chi_{t}^{i} I,\quad P(\mathcal{C}_{n_i(t)})\preceq P(\Theta_*^{i})\label{eq:usefulforbndn}
\end{align}

combing which yields

\begin{align*}
   \begin{pmatrix}
	Q^i & 0 \\
	0 & R^i
	\end{pmatrix}\bullet \big(\Sigma(\Theta_*^{i}))-\Sigma(\mathcal{C}_{n_i(t)})\big) =\sigma_{\omega}^2\bigg(\operatorname{tr}(P(\Theta_*^{i}))-\operatorname{tr}(P(\mathcal{C}_{n_i(t)}))\bigg)\preceq n \sigma_{\omega}^2\chi_t^i
\end{align*}
which results in 
\begin{align}
 \|\Sigma(\Theta_*^{i}))-\Sigma(\mathcal{C}_{n_i(t)})\|_F   \leq \operatorname{tr}\big(\Sigma(\Theta_*^{i}))-\Sigma(\mathcal{C}_{n_i(t)})\big)\leq \frac{n\sigma_{\omega}^2}{\alpha_0^i}\chi_t^i. \label{eq:febsigma}
\end{align}

By definition we have
\begin{align*}
    K(\mathcal{C}_{n_i(t)})=\Sigma_{ux}(\mathcal{C}_{n_i(t)}){\Sigma_{xx}^{-1}(\mathcal{C}_{n_i(t)})}
\end{align*}
and 
\begin{align*}
    K(\Theta_*^i)=\Sigma_{ux}(\Theta_*^i)\Sigma_{xx}^{-1}(\Theta_*^i)
\end{align*}
where $\Sigma(\Theta_*^i)$ is the solution of primal original SDP.

We then can write
\begin{align*}
     \|K(\mathcal{C}_{n_i(t)})- K(\Theta_*^i)\|&=\|\big(\Sigma_{ux}(\Theta_*^i)-\Sigma_{ux}(\mathcal{C}_{n_i(t)})\big)\Sigma^{-1}_{xx}(\Theta_*^i)+\Sigma_{ux}(\mathcal{C}_{n_i(t)})\big(\Sigma_{xx}^{-1}(\Theta_*^i)-\Sigma_{xx}^{-1}(\mathcal{C}_{n_i(t)})\big)\|\\
     &\leq \|\big(\Sigma_{ux}(\Theta_*^i)-\Sigma_{ux}(\mathcal{C}_{n_i(t)})\|_F \|\Sigma^{-1}_{xx}(\Theta_*^i)\|+\|\Sigma_{ux}(\mathcal{C}_{n_i(t)})\|_F\|\Sigma_{xx}^{-1}(\Theta_*^i)-\Sigma_{xx}^{-1}(\mathcal{C}_{n_i(t)})\|.
\end{align*}
For first term applying (\ref{eq:febsigma}) and $\|\Sigma^{-1}_{xx}(\Theta_*^i)\|\leq \sigma_{\omega}^2$, implied by (\ref{eq:SDPKhali}), yields 

\begin{align*}
    \|\big(\Sigma_{ux}(\Theta_*^i)-\Sigma_{ux}(\mathcal{C}_{n_i(t)})\|_F\|\Sigma^{-1}_{xx}(\Theta_*^i)\|\leq \frac{n}{\alpha_0}\chi_t^i.
\end{align*}
Furthermore,
\begin{align}
   \nonumber  \|\Sigma_{xx}^{-1}(\Theta_*^i)-\Sigma_{xx}^{-1}(\mathcal{C}_{n_i(t)})\|&=\|\Sigma_{xx}^{-1}(\Theta_*^i)\bigg(\Sigma_{xx}(\mathcal{C}_{n_i(t)})-\Sigma_{xx}(\Theta_*^i)\bigg)\Sigma_{xx}^{-1}(\mathcal{C}_{n_i(t)})\|\\
    &\leq \|\Sigma_{xx}^{-1}(\Theta_*^i)\| \Sigma_{xx}(\mathcal{C}_{n_i(t)})-\Sigma_{xx}(\Theta_*^i)\|_F\|\Sigma_{xx}^{-1}(\mathcal{C}_{n_i(t)})\|. \label{eq:intermitn}
\end{align}
By the existence of $\Sigma_{xx}^{-1}(\mathcal{C}_{n_i(t)})$ shown in \cite{cohen2019learning}, there exist $\tilde{\sigma}_{xx}$ such that $\Sigma_{xx}(\mathcal{C}_{n_i(t)})\succeq \tilde{\sigma}_{xx} I$ for any $i$ and all $t$ then we can further upper-bound the term (\ref{eq:intermitn}) as follows:

\begin{align*}
     \|\Sigma_{xx}^{-1}(\Theta_*^i)-\Sigma_{xx}^{-1}(\mathcal{C}_{n_i(t)})\|\leq \frac{n}{\alpha_0 \tilde{\sigma}_{xx}}\chi_t^i
\end{align*}
applying which results in 
\begin{align}
    \|K(\mathcal{C}_{n_i(t)})- K(\Theta_*^i)\|\leq \frac{n}{\alpha_0}(1+\frac{1}{\tilde{\sigma}_{xx}})\chi_t^{i}.
\end{align}

To bound the dwell-time estimation error, we consider different scenarios. Recall the definitions
\begin{align*}
     \tau^*_{k,\,k+1}
     &= \max\bigg\{1,
     \underbrace{-\frac{\ln {\rho}_*^{{i_k}}\ln \mathcal{X}_*^{i_{k+1}}-\ln \bar{\alpha}}{\ln \big(1-{\eta }_*^{i_k}\big)}}_{=:\tilde{\tau}_{k,k+1}^*}
     \bigg\}
\end{align*}
and
\begin{align*}
   \tau_{k,k+1}^{a}
   := \max \bigg\{1,
   \underbrace{-\frac{\ln \rho (\mathcal{C}_{n_{i_k}(t_k^a)})+\ln \mathcal{X} (\mathcal{C}_{n_{i_{k+1}}(t_k^a)})-\ln \bar{\alpha}}{\ln \big(1-\eta \big(\mathcal{C}_{n_{i_k}(t_k^a)}\big)\big)}}_{=:\tilde{\tau}^a_{k,k+1}}
   \bigg\}.
\end{align*}

First, consider the case \(\tau_{k,k+1}^* = 1 > \tilde{\tau}_{k,k+1}^*\). Since \(\tau_{k,k+1}^a \ge 1\), we either have \(\tau_{k,k+1}^a = 1\) or \(\tau_{k,k+1}^a = \tilde{\tau}_{k,k+1}^a>1\). In the former case, \(\tau_{k,k+1}^a - \tau_{k,k+1}^* = 0\); and in teh later one
\begin{align}
    \tau_{k,k+1}^a - \tau_{k,k+1}^*
    =
    \tilde{\tau}_{k,k+1}^a - \tau_{k,k+1}^*
    \le
    \tilde{\tau}_{k,k+1}^a - \tilde{\tau}_{k,k+1}^*,
    \label{eq:use1234}
\end{align}
where the inequality follows from \(\tau_{k,k+1}^* > \tilde{\tau}_{k,k+1}^*\). Thus, when \(\tau_{k,k+1}^* = 1\), the estimation error is either zero or bounded by \eqref{eq:use1234}.

Next, consider the case \(\tau_{k,k+1}^* = \tilde{\tau}_{k,k+1}^* > 1\). Since \(\tau_{k,k+1}^a \ge \tau_{k,k+1}^*\), it follows that \(\tau_{k,k+1}^a = \tilde{\tau}_{k,k+1}^a\), and therefore
\begin{align*}
     \tau_{k,k+1}^a - \tau_{k,k+1}^*
     =
     \tilde{\tau}_{k,k+1}^a - \tilde{\tau}_{k,k+1}^*.
\end{align*}

Therefore, in all cases, the dwell-time estimation error is either zero or can be upper-bounded by bounding \(\tilde{\tau}_{k,k+1}^a - \tilde{\tau}_{k,k+1}^*\).

An upper bound on the dwell-time estimation error can be written as follows:

\begin{align*}
    \tilde{\tau}_{k,k+1}^a-\tilde{\tau}_{k,k+1}^*&=\frac{\ln \rho (\mathcal{C}_{n_{i_k}(t_k^a)})+\ln \mathcal{X} (\mathcal{C}_{n_{i_{k+1}}(t_k^a)})-\ln \bar{\alpha}}{-\ln \big(1-\eta \big(\mathcal{C}_{n_{i_k}(t_k^a)}\big)\big)}- \frac{\ln {\rho}_*^{{i_k}}+\ln \mathcal{X}_*^{i_{k+1}}-\ln \bar{\alpha}}{-\ln \big(1-{\eta }_*^{i_k}\big)}\\
    &=\frac{D}{\ln \big(1-\eta \big(\mathcal{C}_{n_{i_k}(t_k^a)}\big)\big)\ln \big(1-{\eta }_*^{i_k}\big)}
\end{align*}
where
\begin{align*}
    D=&-\ln \big(1-{\eta }_*^{i_k}\big) \bigg(\ln \rho (\mathcal{C}_{n_{i_k}(t_k^a)})+\ln \mathcal{X} (\mathcal{C}_{n_{i_{k+1}}(t_k^a)})-\ln \bar{\alpha}\bigg)\\
    &+\ln \big(1-\eta \big(\mathcal{C}_{n_{i_k}(t_k^a)}\big)\big)\bigg(\ln {\rho}_*^{{i_k}}+\ln \mathcal{X}_*^{i_{k+1}}-\ln \bar{\alpha}\bigg)\\
    &+\ln \big(1-{\eta }_*^{i_k}\big)\bigg(\ln {\rho}_*^{{i_k}}+\ln \mathcal{X}_*^{i_{k+1}}-\ln \bar{\alpha}\bigg)\\
    &-\ln \big(1-{\eta }_*^{i_k}\big)\bigg(\ln {\rho}_*^{{i_k}}+\ln \mathcal{X}_*^{i_{k+1}}-\ln \bar{\alpha}\bigg)\\
    =&\underbrace{-\ln \big(1-{\eta }_*^{i_k}\big) \bigg(\ln \rho (\mathcal{C}_{n_{i_k}(t_k^a)})+\ln \mathcal{X} (\mathcal{C}_{n_{i_{k+1}}(t_k^a)})-\ln {\rho}_*^{{i_k}}-\ln \mathcal{X}_*^{i_{k+1}}\bigg)}_{D_1}\\
    &+\underbrace{\bigg(\ln {\rho}_*^{{i_k}}+\ln \mathcal{X}_*^{i_{k+1}}-\ln \bar{\alpha}\bigg)\bigg(\ln \big(1-\eta \big(\mathcal{C}_{n_{i_k}(t_k^a)}\big)\big)-\ln \big(1-{\eta }_*^{i_k}\big)\bigg)}_{D_2}.
\end{align*} 
We first upper-bound \(D_1\). We begin with the bound
\begin{align}
 \nonumber \ln \rho (\mathcal{C}_{n_{i_k}(t_k^a)})-\ln {\rho}_*^{{i_k}}=&\ln \bar{\lambda}\big(P(\mathcal{C}_{n_{i_{k}}(t_k^a)})\big)-\ln\bar{\lambda}\big(P(\Theta_*^{i_{k}})\big)+\\
 \nonumber &\ln \underline{\lambda}\big(P(\Theta_*^{i_{k}})\big)-\ln \underline{\lambda}\big(P(\mathcal{C}_{n_{i_{k}}(t_k^a)})\big)\\
  \leq &\underline{\lambda}\big(P(\Theta_*^{i_{k}})\big)-\ln \underline{\lambda}\big(P(\mathcal{C}_{n_{i_{k}}(t_k^a)})\big)\label{eq:rhobndkh}
\end{align}
where the last inequality follows from $P(\Theta_*^{i_{k}})\succeq P(\mathcal{C}_{n_{i_{k}}(t_k^a)})$.

Similarly,
\begin{align}
  \nonumber\ln \mathcal{X}(\mathcal{C}_{n_{i_{k+1}}(t_k^a)})-\ln {\mathcal{X}}_*^{i_{k+1}}=&\ln \bar{\lambda}\big(P(\mathcal{C}_{n_{i_{k+1}}(t_k^a)})\big)-\ln\bar{\lambda}\big(P(\Theta_*^{i_{k+1}})\big)+\\
 \nonumber  &\ln \underline{\lambda}\big(P(\Theta_*^{i_{k+1}})\big)-\ln \underline{\lambda}\big(P(\mathcal{C}_{n_{i_{k+1}}(t_k^a)})\big)\\
  \leq &\ln \underline{\lambda}\big(P(\Theta_*^{i_{k+1}})\big)-\ln \underline{\lambda}\big(P(\mathcal{C}_{n_{i_{k+1}}(t_k^a)})\big).\label{eq:mathcx}
\end{align}
Applying the following inequality
\begin{align}
    \ln a-\ln b\leq \begin{cases}
a-b, &\quad  a\geq b\geq 1 \\
\frac{a-b}{b},& \quad  \text{otherwise}
\end{cases}.\label{eq:usefulineqlns}
\end{align}

to (\ref{eq:rhobndkh}) and (\ref{eq:mathcx}), and subsequently applying (\ref{eq:usefulforbndn}) to their sum, we obtain the following upper bound on \(D_1\):
\begin{align*}
    D_1\leq -C_1\ln (1-\eta_*^i)\big(\chi_{t_k^a}^{i_k}+\chi_{t_k^a}^{i_{k+1}}\big)
\end{align*}
where \(C_1\) is a finite problem-dependent constant.

The remaining task is to upper-bound \(D_2\). For this purpose, we apply (\ref{eq:usefulineqlns}), which yields
\begin{align*}
     \ln \big(1-\eta \big(\mathcal{C}_{n_{i_k}(t_k^a)}\big)\big)-\ln \big(1-{\eta }_*^{i_k}\big)\leq \frac{|\eta \big(\mathcal{C}_{n_{i_k}(t_k^a)}\big)-{\eta}_*^{i_k}|}{\max\{\eta \big(\mathcal{C}_{n_{i_k}(t_k^a)}\big),{\eta}_*^{i_k}\}}.
\end{align*}

Next,

\begin{align*}
 \eta \big(\mathcal{C}_{n_{i_k}(t_k^a)}\big)-{\eta}_*^{i_k}=&\frac{\underline{\lambda}(\mathcal{H}\big(\mathcal{C}_{n_{i_k}(t_k^a)}\big))}{\bar{\lambda}(P\big(\mathcal{C}_{n_{i_k}(t_k^a)}\big))}-\frac{\underline{\lambda}(H(\Theta_*^{i_k}))}{\bar{\lambda}(P(\Theta_*^{i_k}))}=\frac{F}{\bar{\lambda}(P\big(\mathcal{C}_{n_{i_k}(t_k^a))}\big)) \bar{\lambda}(P(\Theta_*^{i_k}))}
\end{align*}
where
\begin{align*}
    F=&\bar{\lambda}(P(\Theta_*^{i_k}))\underline{\lambda}(\mathcal{H}\big(\mathcal{C}_{n_{i_k}(t_k^a)}\big))-\bar{\lambda}(P\big(\mathcal{C}_{n_{i_k}(t_k^a)}\big)) \underline{\lambda}(H(\Theta_*^{i_k}))\\
    &+\bar{\lambda}(P(\Theta_*^{i_k}))\underline{\lambda}(H(\Theta_*^{i_k}))-\bar{\lambda}(P(\Theta_*^{i_k}))\underline{\lambda}(H(\Theta_*^{i_k}))\\
    &=\underbrace{\underline{\lambda}(H(\Theta_*^{i_k}))\bigg(\bar{\lambda}(P(\Theta_*^{i_k}))-\bar{\lambda}(P\big(\mathcal{C}_{n_{i_k}(t_k^a)}\big)\bigg)}_{F_1}+\underbrace{\bar{\lambda}(P(\Theta_*^{i_k}))\bigg(\underline{\lambda}(\mathcal{H}\big(\mathcal{C}_{n_{i_k}(t_k^a)}\big))-\underline{\lambda}(H(\Theta_*^{i_k}))\bigg)}_{F_2}.
\end{align*}

The term \(F_1\), similarly to (\ref{eq:rhobndkh}), can be bounded as

\begin{align*}
    F_1:=\underline{\lambda}(H(\Theta_*^{i_k}))\bigg(\bar{\lambda}(P(\Theta_*^{i_k}))-\bar{\lambda}(P\big(\mathcal{C}_{n_{i_k}(t_k^a)}\big)\bigg)\leq C_2\underline{\lambda}(H(\Theta_*^{i_k})) \alpha_{t_k^a}^{i_k}.
\end{align*}

For \(F_2\), applying Weyl’s inequality yields
\begin{align*}
    \underline{\lambda}(\mathcal{H}\big(\mathcal{C}_{n_{i_k}(t_k^a)}\big))-\underline{\lambda}(H(\Theta_*^{i_k}))\leq \|\mathcal{H}\big(\mathcal{C}_{n_{i_k}(t_k^a)}\big)-H(\Theta_*^{i_k})\|.
\end{align*}
Recalling the definitions of $H(\Theta_*^{i_k})$ and $\mathcal{H}\big(\mathcal{C}_{n_{i_k}(t_k^a)}\big)$, we obtain

\begin{align*}
    \underline{\lambda}(\mathcal{H}\big(\mathcal{C}_{n_{i_k}(t_k^a)}\big))-\underline{\lambda}(H(\Theta_*^{i_k}))&\leq | \underline{\lambda}(\mathcal{H}\big(\mathcal{C}_{n_{i_k}(t_k^a)}\big))-\underline{\lambda}(H(\Theta_*^{i_k}))|\\
    &\leq \|K^\top(\Theta_*^{i_k})R^{i_k}K(\Theta_*^{i_k})-K^\top\big(\mathcal{C}_{n_{i_k}(t_k^a)}\big)R^{i_k}K\big(\mathcal{C}_{n_{i_k}(t_k^a)}\big)+(1+\kappa_{i_k}^2)\alpha^{i_k}_{t_k^a}\|\\
    &\leq \|K^\top(\Theta_*^{i_k})R^{i_k}K(\Theta_*^{i_k})-K^\top\big(\mathcal{C}_{n_{i_k}(t_k^a)}\big)R^{i_k}K\big(\mathcal{C}_{n_{i_k}(t_k^a)}\big)\|+(1+\kappa_{i_k}^2)\alpha^{i_k}_{t_k^a}.
\end{align*}

Define
\begin{align*}
    \Delta_{i_k}:=K(\Theta_*^{i_k})-K\big(\mathcal{C}_{n_{i_k}(t_k^a)}\big)
\end{align*}
then one can write
\begin{align*}
    K^\top(\Theta_*^{i_k})R^{i_k}K(\Theta_*^{i_k})-K^\top\big(\mathcal{C}_{n_{i_k}(t_k^a)}\big)R^{i_k}K\big(\mathcal{C}_{n_{i_k}(t_k^a)}\big)=\Delta_{i_k}^\top R^{i_k} K\big(\mathcal{C}_{n_{i_k}(t_k^a)}\big)+K^\top\big(\mathcal{C}_{n_{i_k}(t_k^a)}\big)R^{i_k}\Delta_{i_k}+\Delta_{i_k}^\top R^{i_k}\Delta_{i_k}.
\end{align*}

Hence,
\begin{align*}
    \|K^\top(\Theta_*^{i_k})R^{i_k}K(\Theta_*^{i_k})-K^\top\big(\mathcal{C}_{n_{i_k}(t_k^a)}\big)R^{i_k}K\big(\mathcal{C}_{n_{i_k}(t_k^a)}\big)\|&\leq 2\|K\big(\mathcal{C}_{n_{i_k}(t_k^a)}\big)\|\|R^{i_k}\|\|\Delta_{i_k}\|+\|R^{i_k}\|\|\Delta_{i_k}\|^2\\
    &\leq \alpha_1^{i_k}\bigg(2\kappa_{i_k}^2 \big(\frac{n}{\alpha_0}(1+\frac{1}{\tilde{\sigma}_{xx}})\alpha_t^{i_k}\big)+\big(\frac{n}{\alpha_0}(1+\frac{1}{\tilde{\sigma}_{xx}})\alpha_t^{i_k}\big)^2\bigg).
\end{align*}
Therefore,
\begin{align*}
    F_2\leq \alpha_1^{i_k}\bigg(2\kappa_{i_k}^2 \big(\frac{n}{\alpha_0}(1+\frac{1}{\tilde{\sigma}_{xx}})\chi^{i_k}_{t_k^a}\big)+\big(\frac{n}{\alpha_0}(1+\frac{1}{\tilde{\sigma}_{xx}})\chi^{i_k}_{t_k^a}\big)^2\bigg). 
\end{align*}
Combining the bounds, we conclude that
\begin{align*}
    \tau_{k,k+1}^a-\tau_{k,k+1}^*\leq \bar{C}_0\chi_{t_k^a}^{i_{k+1}}+\bar{C}_1\chi^{i_k}_{t_k^a}+\bar{C}_2\big(\chi^{i_k}_{t_k^a}\big)^2
\end{align*}
where $\bar{C}_0$, $\bar{C}_1$, and $\bar{C}_3$ are problem-dependent constants.
\end{proof}

\begin{corollary}\label{cor:LSBound}
    Consider the case where, in mode $i_k$, the system operates under the learning-based controller, and after switching to mode $i_{k+1}$, it applies the fallback controller $K_0^{i_{k+1}}$ and computes the corresponding minimum dwell time. Then, the dwell-time estimation error satisfies
    \begin{align*}
        \tau_{k,k+1}^a-\tau_{k,k+1}^*
        \leq {}& \tilde{C}_0\Big(\ln \bar{\lambda}\big(P_{K_0^{i_{k+1}}}\big)
        -\ln\bar{\lambda}\big(P(\Theta_*^{i_{k+1}})\big)\Big) \\
        &+\tilde{C}_1\chi^{i_k}_{t_k^a}
        +\tilde{C}_2\big(\chi^{i_k}_{t_k^a}\big)^2 .
    \end{align*}
\end{corollary}

\begin{proof}
    The proof follows the same steps as those of Theorem~\ref{Thm:dwellTimeError}, with the only difference arising in the upper bound of $D_1$ for mode $i_{k+1}$. Specifically, since the fallback controller $K_0^{i_{k+1}}$ is applied in mode $i_{k+1}$, the corresponding matrix $P_{K_0^{i_{k+1}}}$ replaces the matrix in (\ref{eq:mathcx}). Therefore, we have
    \begin{align}
        \ln \mathcal{X}(\mathcal{C}_{n_{i_{k+1}}(t_k^a)})
        -\ln {\mathcal{X}}_*^{i_{k+1}}
        ={}&
        \ln \bar{\lambda}\big(P_{K_0^{i_{k+1}}}\big)
        -\ln\bar{\lambda}\big(P(\Theta_*^{i_{k+1}})\big)
        \nonumber\\
        &+
        \ln \underline{\lambda}\big(P(\Theta_*^{i_{k+1}})\big)
        -\ln \underline{\lambda}\big(P_{K_0^{i_{k+1}}}\big)
        \nonumber\\
       \nonumber  \leq&
        \ln \bar{\lambda}\big(P_{K_0^{i_{k+1}}}\big)
        -\ln\bar{\lambda}\big(P(\Theta_*^{i_{k+1}})\big),
    \end{align}
    where the last inequality follows from the fact that
    $P(\Theta_*^{i_{k+1}})
    \preceq P_{K_0^{i_{k+1}}}$.
\end{proof}

\begin{lemma}[Bound on the total number of fallback steps]
\label{lem:fallback_steps}
Let $n_i^{\mathrm{fb}}(t)$ denote the number of time steps up to time $t$ during
which the fallback controller is applied while the system operates in mode
$i\in\mathcal{M}$. Suppose that for every mode $i\in\mathcal{M}$,
\[
n_i^{\mathrm{fb}}(t)
\leq
\left(
\log \frac{2d_x\bar{\tau}_*n_s}{\delta}
\right)^{\frac{2}{1-2\zeta}}.
\]
Then, the total number of time steps during which the fallback controller is
applied satisfies
\[
\sum_{i\in\mathcal{M}} n_i^{\mathrm{fb}}(t)
\leq
|\mathcal{M}|
\left(
\log \frac{2d_x\bar{\tau}_*n_s}{\delta}
\right)^{\frac{2}{1-2\zeta}}.
\]
\end{lemma}

\begin{proof}
By the criterion for the fallback policy in (\ref{eq:switchtresh}), for each mode we have
\[
n_i^{\mathrm{fb}}(t)
\leq
\left(
\log \frac{2d_x\bar{\tau}_*n_s}{\delta}
\right)^{\frac{2}{1-2\zeta}},
\qquad \forall i\in\mathcal{M}.
\]
Therefore, summing this inequality over all modes gives
\begin{align*}
\sum_{i\in\mathcal{M}} n_i^{\mathrm{fb}}(t)
&\leq
\sum_{i\in\mathcal{M}}
\left(
\log \frac{2d_x\bar{\tau}_*n_s}{\delta}
\right)^{\frac{2}{1-2\zeta}}\\
&=
|\mathcal{M}|
\left(
\log \frac{2d_x\bar{\tau}_*n_s}{\delta}
\right)^{\frac{2}{1-2\zeta}}.
\end{align*}
\end{proof}

\subsection{Regret bound analysis} \label{sec:app-reg}

We begin by decomposing the regret into several components and then bounding each term separately. Recall that
\begin{align*}
R(\mathcal{I}) = \mathcal{C}_{\mathcal{I}}^{\bar{M}\bar{S}}- \mathcal{C}_{\mathcal{I}}^{M\bar{S}}.
\end{align*}
By partitioning the horizon into switching epochs (cf. Remark~\ref{rem:sumswi}), the regret can be expressed as the sum of per-epoch regret terms. Moreover, at each switching epoch, either the learned controller is certified to satisfy the safety condition or the fallback controller is applied. Accordingly, define the sets of certified and fallback switching epochs as
\begin{align*}
    \mathcal{I}_{\mathrm{cert}}
    &=
    \left\{
        k:(n_i(t_k^a)+\bar{c}_{i_k^a})^{\frac{1}{2}-\zeta}\geq \log \frac{2nt_k^a}{\delta}
    \right\},\\
    \mathcal{I}_{\mathrm{fallback}}
    &=
    \mathcal{I}\setminus\mathcal{I}_{\mathrm{cert}},
\end{align*}
where \(\mathcal{I}=\{0,\ldots,n_s-1\}\) denotes the set of all switching epochs. The per-epoch regret is defined as
\begin{align*}
    R_{k,k+1}(\zeta)=\begin{cases}
R_{k,k+1}^{\mathrm{cert}}, &\quad \quad \quad \quad  k\in \mathcal{I}_{\mathrm{cert}} \\
R_{k,k+1}^{\mathrm{fallback}}, & \quad \quad \quad \quad k\in \mathcal{I}_{\mathrm{fallback}}
\end{cases}.
\end{align*}

Consequently, the expected regret admits the decomposition
\begin{align*}
    \mathbb{E}\left[R_{\zeta}(\mathcal{I})\right]
    &=
    \mathbb{E}\left[
    \sum_{k\in\mathcal{I}_{\mathrm{cert}}}
    R^{\mathrm{cert}}_{k,k+1}(\zeta)
    +
    \sum_{k\in\mathcal{I}_{\mathrm{fallback}}}
    R^{\mathrm{fallback}}_{k,k+1}(\zeta)
    \right]
\end{align*}
where, by a slight abuse of notation, the subscript $\zeta$ emphasizes the dependence of the regret on the perturbation-noise exponent $\zeta$ introduced in~(\ref{eq:Gamarslop}).

The per-epoch regret over certified switching epochs can be written as
\begin{align*}
     R^{\mathrm{cert}}_{k,k+1}(\zeta)=\sum_{t=t_k^a}^{t_{k+1}^a-1}\bigg(x_t^\top Q^{i_k}x_t+u_t\top R^{i_k}u_t-J_*^{f(t)}\bigg)
\end{align*}
where
\begin{align*}
    J_*^{f(t)}=\begin{cases}
J_*^{i_k}, &\quad \quad \quad \quad  t_k^a\leq t\leq t_k^a+\tau^*_{k,k+1} \\
0, & t_k^a+\tau^*_{k,k+1} <t\leq t_{k+1}^a
\end{cases}.
\end{align*}
The per-epoch regret over fallback switching epochs is given by~(\ref{eq:fallbackpa}).

Accordingly, the per-epoch regret over certified switching epochs can be rewritten as
\begin{align*}
  R^{\mathrm{cert}}_{k,k+1}(\zeta)&=  \sum_{t=t_k^a}^{t_{k+1}^a-1} \bigg(x_t^\top Q^{i_k}x_t+u_t^\top R^{i_k}u_t-J_*^{i_k}\bigg)+\sum_{t=t_k^a+\tau_{k,k+1}^*}^{t_{k+1}^a-1}J_*^{i_k}.
\end{align*}
Next, decompose the control input into its nominal feedback component and perturbation as
\begin{align}
    \bar{u}_t=K(\mathcal{C}_{n_{i_k}(t_k^a)})x_t \quad t^a_k\leq t<t^a_{k+1}. \label{eq:baru}
\end{align}
Substituting \eqref{eq:baru} into the stage cost yields
\begin{align*}
  \sum_{t=t_k^a}^{t_{k+1}^a-1} \bigg(x_t^\top Q^{i_k}x_t+u_t^\top R^{i_k}u_t-J_*^{i_k}\bigg)=\sum_{t=t_k^a}^{t_{k+1}^a-1} \bigg(x_t^\top Q^{i_k}x_t+\bar{u}_t^\top R^{i_k}\bar{u}_t-J_*^{i_k}\bigg)+\sum_{t=t_k^a}^{t_{k+1}^a-1} 2\bar{u}_t^\top R^{i_k}\eta_t +\sum_{t=t_k^a}^{t_{k+1}^a-1} \eta_t^\top R^{i_k}\eta_t.
\end{align*}
Multiplying both sides by the indicator $1_{\mathcal{E}_k^{\zeta}}$ and combining the above expressions, we obtain
\begin{align}
\nonumber R^{\mathrm{cert}}_{k,k+1}(\zeta)1_{\mathcal{E}_k^{\zeta}}= & \overbrace{\sum_{t=t_k^a}^{t_{k+1}^a-1} \bigg(x_t^\top Q^{i_k}x_t+\bar{u}_t^\top R^{i_k}\bar{u}_t-J_*^{i_k}\bigg)1_{\mathcal{E}_k^{\zeta}}}^{R^{(1)}_{k,k+1}}+\overbrace{\sum_{t=t_k^a+\tau_{k,k+1}^*}^{t_{k+1}^a-1}J_*^{i_k}1_{\mathcal{E}_k^{\zeta}}}^{R^{(2)}_{k,k+1}}\\
&+\overbrace{\sum_{t=t_k^a}^{t_{k+1}^a-1} 2\bar{u}_t^\top R^{i_k}\eta_t 1_{\mathcal{E}_k^{\zeta}}}^{R^{(31)}_{k,k+1}}+\underbrace{\sum_{t=t_k^a}^{t_{k+1}^a-1} \eta_t^{i_k^\top} R^{i_k}\eta^{i_k}_t 1_{\mathcal{E}_k^{\zeta}}}_{R^{(4)}_{k,k+1}}. \label{eq:regdeco}
\end{align}
Following the decomposition developed in~\cite{chekan2024any}, the term
$R^{(1)}_{k,k+1}$ can be further decomposed as
\begin{align}
 &R^{(1)}_{k,k+1}\leq  R^{(11)}_{k,k+1}+R^{(12)}_{k,k+1}+R^{(13)}_{k,k+1}+R^{(14)}_{k,k+1}\label{eq:R_1stat}\\
 \nonumber &R^{(11)}_{k,k+1}=\sum_{t=t_k^a}^{t_{k+1}^a-1} \bigg(x_t^\top P(\mathcal{C}_{n_{i_k}(t_k^a)})x_t-x_{t+1}^\top P(\mathcal{C}_{n_{i_k}(t_k^a)})x_{t+1}\bigg)1_{\mathcal{E}_k^{\zeta}} \\
 \nonumber &R^{(12)}_{k,k+1}=\sum_{t=t_k^a}^{t_{k+1}^a-1} \big(\omega_{t+1}^\top P(\mathcal{C}_{n_{i_k}(t_k^a)})x_t \big)1_{\mathcal{E}_k^{\zeta}}\\
 \nonumber & R^{(13)}_{k,k+1}=\sum_{t=t_k^a}^{t_{k+1}^a-1} \bigg(\omega_{t+1}^\top P(\mathcal{C}_{n_{i_k}(t_k^a)})\omega_{t+1}-\sigma_{\omega}^2 \|P(\mathcal{C}_{n_{i_k}(t_k^a)}))\|_*\bigg)1_{\mathcal{E}_k^{\zeta}}\\
 & R^{(14)}_{k,k+1}=\sum_{t=t_k^a}^{t_{k+1}^a-1}2\|P(\mathcal{C}_{n_{i_k}(t_k^a)})\|\mu_{n_{i_k}(t_k^a)}\big(\bar{z}_t^\top {V}^{-1}_{n_{i_k}(t_k^a)}\bar{z}_t\big)1_{\mathcal{E}_k^{\zeta}}.
\label{eq:R14stat}
 \end{align}

We next consider the regret incurred during fallback switching epochs. Multiplying
the per-epoch regret by the indicator of the event
$\mathcal{E}_k^{\zeta}$ yields
 \begin{align}
     R^{\mathrm{fallback}}_{k,k+1}(\zeta)1_{\mathcal{E}_k^{\zeta}}=&\sum_{t=t_k^a}^{t_{k+1}^a} \big(x_t^\top Q^{i_k}x_t+u_t^\top R^{i_k}u_t-J_*^{f(t)}\big)1_{\mathcal{E}_k^{\zeta}} \label{eq:fallbackpa}\\
   \nonumber   =& \sum_{t=t_k^a}^{t_{k+1}^a} \bigg(x_t^\top \big(Q^{i_k}+K_0^{{i_k}^\top}R^{i_k} K_0^{i_k}\big)x_t-J_*^{f(t)}\bigg)1_{\mathcal{E}_k^{\zeta}}+ \underbrace{\sum_{t=t_k^a}^{t_{k+1}^a-1} 2x_t^\top K_0^{{i_k}^\top} R^{i_k}\eta_t 1_{\mathcal{E}_k^{\zeta}}}_{R^{(32)}_{k,k+1}}+\sum_{t=t_k^a}^{t_{k+1}^a-1} \eta_t^{i_k^\top} R^{i_k}\eta^{i_k}_t 1_{\mathcal{E}_k^{\zeta}}.
 \end{align}
Using the Lyapunov equation satisfied by the fallback controller
$K_0^{i_k}$, the nominal stage-cost term can be bounded as
\begin{align*}
    \sum_{t=t_k^a}^{t_{k+1}^a} \bigg(x_t^\top \big(Q^{i_k}+K_0^{{i_k}^\top}R^{i_k} K_0^{i_k}\big)x_t-J_*^{f(t)}\bigg)1_{\mathcal{E}_k^{\zeta}}\leq &\sum_{t=t_k^a}^{t_{k+1}^a} \bigg(x_t^\top \big(Q^{i_k}+K_0^{{i_k}^\top}R^{i_k} K_0^{i_k}\big)x_t-J_*^{i_k}\bigg)1_{\mathcal{E}_k^{\zeta}}\\
    \leq&\sum_{t=t_k^a}^{t_{k+1}^a} x_t^\top \bigg(P_{K_0^{i_k}}-(A_*^{i_k}+B_*^{i_k}K_0^{i_k})^\top P_{K_0^{i_k}}(A_*^{i_k}+B_*^{i_k}K_0^{i_k}) \bigg)x_t1_{\mathcal{E}_k^{\zeta}}\\
     =&  \underbrace{\sum_{t=t_k^a}^{t_{k+1}^a} \bigg(x_t^\top P_{K_0^{i_k}}x_t- x_{t+1}^\top P_{K_0^{i_k}}x_{t+1} \bigg)1_{\mathcal{E}_k^{\zeta}}}_{R_{k,k}^{(21)}}+\underbrace{\sum_{t=t_k^a}^{t_{k+1}^a-1} \big(\omega_{t+1}^\top P_{K_0^{i_k}}x_t \big)1_{\mathcal{E}_k^{\zeta}}}_{R_{k,k+1}^{(22)}}\\
     &+ \underbrace{\sum_{t=t_k^a}^{t_{k+1}^a-1} \omega_{t+1}^\top P_{K_0^{i_k}}\omega_{t+1}1_{\mathcal{E}_k^{\zeta}}}_{R_{k,k+1}^{(23)}}.
\end{align*}
Since the sets
\(\mathcal{I}_{\mathrm{cert}}\) and
\(\mathcal{I}_{\mathrm{fallback}}\)
are disjoint and satisfy
\(\mathcal{I}_{\mathrm{cert}}\cup
\mathcal{I}_{\mathrm{fallback}}
=\mathcal{I}\),
it follows that
\begin{align*}
\mathbb{E}\!\left[
\sum_{k\in\mathcal{I}_{\mathrm{cert}}}
R^{(4)}_{k,k+1}
+
\sum_{k\in\mathcal{I}_{\mathrm{fallback}}}
R^{(4)}_{k,k+1}
\right]
=
\mathbb{E}\!\left[
\sum_{k=0}^{n_s-1}
R^{(4)}_{k,k+1}
\right].
\end{align*}
Combining the preceding decompositions, we obtain the following upper bound:
\begin{align*}
    \mathbb{E}\left[\tilde{R}_{\zeta}(\mathcal{I})\right]
    \leq&
    \mathbb{E}\left[
    \sum_{k\in\mathcal{I}_{\mathrm{cert}}} \big(R^{(11)}_{k,k+1}+R^{(12)}_{k,k+1}+R^{(13)}_{k,k+1}+R^{(14)}_{k,k+1}+R^{(2)}_{k,k+1}+R^{(31)}_{k,k+1}\big) \right]+\\
    &\mathbb{E}\left[\sum_{k\in\mathcal{I}_{\mathrm{fallback}}} \big(R^{(21)}_{k,k+1}+R^{(22)}_{k,k+1}+R^{(23)}_{k,k+1}+R^{(32)}_{k,k+1}\big)\right]+\mathbb{E}\left[\sum_{k=0}^{n_s-1}R^{(4)}_{k,k+1}\right].  
\end{align*}

We now prove Lemma~\ref{lem:epseventprob}, which establishes a lower bound on the probability of the good event $\mathcal{E}_{n_s-1}^{\zeta}$. Subsequently, the upper bound on $\mathbb{E} [\tilde{R}_{\zeta}(\mathcal{I})]$ is converted into an upper bound on $\mathbb{E}[R_{\zeta}(\mathcal{I})]$ by accounting for the probability of the good event.

\begin{proof}[{\bf Proof of Lemma \ref{lem:epseventprob}}]
Define the following epoch-based good events
  \begin{align*}
      \tilde{\mathcal{A}}^i_k&=:\{\forall s = 0, \ldots, k, \quad \Theta_*^{i}\in \mathcal{C}_{n_{i}(t_s^a)}\}\\
      \mathcal{B}_k^i(\zeta)&=: \{\forall s = 0, \ldots, k, \quad \mu_{n_i(t_k^a)} \|V^{-1}_{n_i(t_k^a)}\|
    \leq
   \bar{D}(n_i(t_k^a)+\bar{c})^{-\frac{1}{2}+\zeta}\}\\
       \mathcal{H}_k&=: \{\forall s = 0, \ldots, k, \quad \|z_s\|^2\leq c_z^2\log \frac{t_k^a}{\delta}\}.
  \end{align*}
Then,
\[
\mathcal E_k^\zeta
=
\left(
\bigcap_{i\in\mathcal M}
\left(
\tilde{\mathcal A}_k^i
\cap
\mathcal B_k^i(\zeta)
\right)
\right)
\cap
\mathcal H_k.
\]

 With an initial estimate \(\Theta_0^i\) chosen according to~(\ref{eq:epsstst0}), the event \(\tilde{\mathcal{A}}_k^i\) holds uniformly over time with probability at least \(1-\delta\). This follows from the uniform-in-time self-normalized martingale concentration bound; see Theorem~2 of \cite{abbasi2011improved}.  Consequently, Since the epochs form a subset of the time indices, the same guarantee immediately holds over the sequence of epochs.

Next, on the event
\[
\mathcal A_t^i
\cap
\mathcal D_t^i
\cap
\mathcal G_t^i,
\]
Theorem~\ref{Stability_thm17} establishes that the generated controllers are
uniformly $(\bar\kappa_i,\bar\gamma_i)$-strongly stabilizing.
Furthermore, the same theorem shows that
\[
\mu_{n_i(t)}
\|V^{-1}_{n_i(t)}\|
\le
\bar D
(n_i(t)+\bar c)^{-\frac12+\zeta},
\qquad
\forall t\ge0,
\]
which immediately implies
\[
\mathcal{A}_t^i\cap \mathcal{D}_t^i\cap \mathcal{G}_t^i
\subseteq
\mathcal{B}_k^i(\zeta).
\]

Therefore, \(\mathcal{B}_k^i(\zeta)\) holds uniformly over time with probability at least \(1-3\delta\), and therefore also on the epochs with the same probability. Applying a union bound over all modes yields a probability of at least \(1-3|\mathcal{M}|\delta\).

Finally, the event \(\mathcal{H}_k\) requires that the closed-loop system remains stable throughout each epoch. As established previously, this property, together with the validity of the designed dwell time, holds on the event
\[
\bigcap_{i\in\mathcal{M}}
\left(
\mathcal{A}_t^i
\cap
\mathcal{D}_t^i
\cap
\mathcal{G}_t^i
\right).
\]
Applying a union bound over all modes, we conclude that the event \(\mathcal{E}_{n_s-1}^{\zeta}\) holds with probability at least \(1-3|\mathcal{M}|\delta\).
\end{proof}

We now proceed to bound each term individually. We first show that some of the terms vanish identically.

\begin{lemma}
On the event $\mathcal{E}^{\zeta}_k$, the following equalities hold:

\begin{align*} \mathbb{E}\big[\sum_{k\in\mathcal{I}_{\mathrm{cert}}}R_{k,k+1}^{(12)}\big]= \mathbb{E}\big[\sum_{k\in\mathcal{I}_{\mathrm{cert}}}R_{k,k+1}^{(13)}\big]= \mathbb{E}\big[\sum_{k\in\mathcal{I}_{\mathrm{cert}}}R_{k,k+1}^{(31)}\big]=\mathbb{E}\big[\sum_{k\in\mathcal{I}_{\mathrm{fallback}}}R_{k,k+1}^{(22)}\big]=\mathbb{E}\big[\sum_{k\in\mathcal{I}_{\mathrm{fallback}}}R_{k,k+1}^{(32)}\big]=0.
\end{align*}
\end{lemma}

\begin{proof}
We show the result term by term.

\paragraph{(1) $\mathbb{E}\big[\sum_{k\in\mathcal{I}_{\mathrm{fallback}}}R_{k,k+1}^{(22)}\big]=\mathbb{E}\big[\sum_{k\in\mathcal{I}_{\mathrm{cert}}}R_{k,k+1}^{(12)}\big]=0$:}

We write
\begin{align*}
   \mathbb{E}\big[\sum_{k\in\mathcal{I}_{\mathrm{cert}}}R_{k,k+1}^{(12)}\big]=\sum_{k\in\mathcal{I}_{\mathrm{cert}}}\sum_{t=t_k^a}^{t_k^a+\tau_{k,k+1}^*-1}  \mathbb{E}\big[\big(\omega_{t+1}^\top P(\mathcal{C}_{n_{i_k}(t_k^a)})x_t \big)1_{\mathcal{E}^{\zeta}_k}\big].
\end{align*}
For a fixed $k$ and $t$, we apply iterated expectation with respect to $\mathcal{F}_{t_k^a}$:
\begin{align*}
\mathbb{E}\!\left[\big(\omega_{t+1}^\top P(\mathcal{C}_{n_{i_k}(t_k^a)}) x_t\big)1_{\mathcal{E}_k}\right]
&= \mathbb{E}\!\left[
\mathbb{E}\!\left[\big(\omega_{t+1}^\top P(\mathcal{C}_{n_{i_k}(t_k^a)}) x_t\big)1_{\mathcal{E}^{\zeta}_k}
\mid \mathcal{F}_{t_k^a}\right]\right].
\end{align*}

Since $P(\mathcal{C}_{n_{i_k}(t_k^a)})$ and $1_{\mathcal{E}^{\zeta}_k}$ are $\mathcal{F}_{t_k^a}$-measurable, we obtain

\begin{align*}
\mathbb{E}\!\left[
\big(\omega_{t+1}^{\top}P(\mathcal{C}_{n_{i_k}(t_k^a)})x_t\big)
1_{\mathcal{E}_k^{\zeta}}
\,\middle|\,
\mathcal{F}_{t_k^a}
\right]
&=
1_{\mathcal{E}_k^{\zeta}}
P(\mathcal{C}_{n_{i_k}(t_k^a)})\bullet
\mathbb{E}\!\left[
\omega_{t+1}x_t^{\top}
\,\middle|\,
\mathcal{F}_{t_k^a}
\right] \\
&=
1_{\mathcal{E}_k^{\zeta}}
P(\mathcal{C}_{n_{i_k}(t_k^a)})\bullet
\mathbb{E}\!\left[
\omega_{t+1}
\,\middle|\,
\mathcal{F}_{t_k^a}
\right]
\mathbb{E}\!\left[
x_t
\,\middle|\,
\mathcal{F}_{t_k^a}
\right]^{\!\top},
\end{align*}
where $t^a_k\leq t<t^a_{k+1}$.
Here the second equality follows
from the independence of $\omega_{t+1}$ and
$(x_t,\mathcal{F}_{t_k^a})$.

Together with the martingale difference property,
\[
\mathbb{E}\!\left[\omega_{t+1}\mid\mathcal{F}_{t_k^a}\right]=0,
\]
we conclude that
\[
\mathbb{E}\!\left[\big(\omega_{t+1}^\top P(\mathcal{C}_{n_{i_k}(t_k^a)}) x_t\big)1_{\mathcal{E}^{\zeta}_k}\right] = 0.
\]
Summing over $k$ and $t$ yields
\[
\mathbb{E}\big[\sum_{k\in\mathcal{I}_{\mathrm{cert}}}R_{k,k+1}^{(12)}\big]=0.
\]
Similarly, we can show
\begin{align*}    \mathbb{E}\big[\sum_{k\in\mathcal{I}_{\mathrm{fallback}}}R_{k,k+1}^{(22)}\big]=0.
\end{align*}

\paragraph{(2) $\mathbb{E}\big[\sum_{k\in\mathcal{I}_{\mathrm{cert}}}R_{k,k+1}^{(13)}\big]=0$ :}

For the second term we have
\begin{align*}
   \mathbb{E}\big[\sum_{k\in\mathcal{I}_{\mathrm{cert}}}R_{k,k+1}^{(13)}\big]=\sum_{k\in\mathcal{I}_{\mathrm{cert}}}\sum_{t=t_k^a}^{t_k^a+\tau_{k,k+1}^*-1}  \mathbb{E}\big[\big(\omega_{t+1}^\top P(\mathcal{C}_{n_{i_k}(t_k^a)})\omega_{t+1}-\sigma_{\omega}^2 \|P(\mathcal{C}_{n_{i_k}(t_k^a)})\|_*\big)1_{\mathcal{E}^{\zeta}_k}\big]
\end{align*}
We can further write
\begin{align*}
    \mathbb{E}\big[\big(\omega_{t+1}^\top P(\mathcal{C}_{n_{i_k}(t_k^a)})\omega_{t+1}-\sigma_{\omega}^2 \|P(\mathcal{C}_{n_{i_k}(t_k^a)})\|_*\big)1_{\mathcal{E}^{\zeta}_k}\big]=\mathbb{E}\big[ \mathbb{E}\big[\big(\omega_{t+1}^\top P(\mathcal{C}_{n_{i_k}(t_k^a)})\omega_{t+1}-\sigma_{\omega}^2 \|P(\mathcal{C}_{n_{i_k}(t_k^a)})\|_*\big)1_{\mathcal{E}^{\zeta}_k}\mid \mathcal{F}_{t_k^a}\big]\big].
\end{align*}
With similar justification as above,
 for all $t \in [t_k^a,\; t_k^a+\tau_{k,k+1}^*)$, we have
\begin{align*}
    \mathbb{E}\big[\big(\omega_{t+1}^\top P(\mathcal{C}_{n_{i_k}(t_k^a)})\omega_{t+1}-\sigma_{\omega}^2 \|P(\mathcal{C}_{n_{i_k}(t_k^a)})\|_*\big)1_{\mathcal{E}^{\zeta}_k}\mid \mathcal{F}_{t_k^a}\big]&=1_{\mathcal{E}^{\zeta}_k} P(\mathcal{C}_{n_{i_k}(t_k^a)})\bullet \mathbb{E}\big[\omega_{t+1} \omega_{t+1}^\top\mid \mathcal{F}_{t_k^a}\big]-\sigma_{\omega}^2 \|P(\mathcal{C}_{n_{i_k}(t_k^a)})\|_*\big)\\
    &=\sigma_{\omega}^2\operatorname{tr}(P(\mathcal{C}_{n_{i_k}(t_k^a)}))-\sigma_{\omega}^2 \|P(\mathcal{C}_{n_{i_k}(t_k^a)})\|_*\big)=0
\end{align*}
where the last equality holds as $\mathbb{E}\big[\omega_{t+1} \omega_{t+1}^\top\mid \mathcal{F}_{t_k^a}\big]=\sigma_{\omega}^2 I$.
Summing over $k$ and $t$ yields
\[
 \mathbb{E}\big[\sum_{k\in\mathcal{I}_{\mathrm{cert}}}R_{k,k+1}^{(13)}\big] = 0.
\]

\paragraph{(3) 
$\mathbb{E}\Big[\sum_{k\in\mathcal{I}_{\mathrm{cert}}} R_{k,k+1}^{(31)}\Big]
=
\mathbb{E}\Big[\sum_{k\in\mathcal{I}_{\mathrm{fallback}}} R_{k,k+1}^{(32)}\Big]
=0$:}

Consider the term
\begin{align*}
    \mathbb{E}\Big[
    \sum_{k\in\mathcal{I}_{\mathrm{cert}}} 
    R_{k,k+1}^{(31)}
    \Big]
    =
    \mathbb{E}\Big[
    \sum_{k\in\mathcal{I}_{\mathrm{cert}}}
    \sum_{t=t_k^a}^{t_{k+1}^a-1}
    2\,\bar{u}_t^\top R^{i_k}\eta_t
    1_{\mathcal E_k^\zeta}
    \Big],
\end{align*}
where $\bar{u}_t$ is defined in (\ref{eq:baru}). For each $k$ and
$t\in[t_k^a,t_{k+1}^a)$, the policy
$K(\mathcal{C}_{n_{i_k}(t_k^a)})$ is
$\mathcal F_{t_k^a}$-measurable. Since
$\mathcal F_{t_k^a}\subseteq\mathcal F_{t-1}$ and
$x_t$ is $\mathcal F_{t-1}$-measurable, it follows that
$\bar u_t$ is $\mathcal F_{t-1}$-measurable. Therefore,

\begin{align*}
\mathbb{E}\!\left[
\bar u_t^\top R^{i_k}\eta_t
1_{\mathcal E_k^\zeta}
\right]&=
\mathbb{E}\!\left[
\mathbb{E}\!\left[
\bar u_t^\top R^{i_k}\eta_t
1_{\mathcal E_k^\zeta}
\,\middle|\,
\mathcal F_{t-1}
\right]
\right]\\
&=
\mathbb{E}\!\left[
\bar u_t^\top R^{i_k}
1_{\mathcal E_k^\zeta}
\mathbb{E}\!\left[
\eta_t
\,\middle|\,
\mathcal F_{t-1}
\right]
\right]
=0,
\end{align*}
where we used that $\eta_t$ is zero-mean and independent of
$\mathcal F_{t-1}$. Summing over $t$ and $k$ yields
\[
\mathbb{E}\Big[
\sum_{k\in\mathcal{I}_{\mathrm{cert}}}
R_{k,k+1}^{(31)}
\Big]=0.
\]
Similarly, we can establish
\[
\mathbb{E}\Big[
\sum_{k\in\mathcal{I}_{\mathrm{fallback}}}
R_{k,k+1}^{(32)}
\Big]=0.
\]
\end{proof}

\begin{lemma} \label{lem:R^{11}}
  The following statement holds:
    \begin{align*}
       \mathbb{E}\left[  \sum_{k\in\mathcal{I}_{\mathrm{cert}}}R_{k,k+1}^{(11)}+\sum_{k\in\mathcal{I}_{\mathrm{fallback}}}R_{k,k+1}^{(21)}\right]\lesssim \mathcal{O}\big(n_m\log \frac{n_s}{\delta}+ |\mathcal{M}|\big(\log \frac{n_s}{\delta}\big)^{\frac{3-2\zeta}{1-2\zeta}}+|\mathcal{M}|^{1-\gamma}n_s^{\gamma} \, \log^2 \frac{n_s}{\delta}\big)
    \end{align*}
\end{lemma}
where $\gamma=\frac{1}{2}+\zeta$ with $\zeta\in (0,\; \frac{1}{2})$.
\begin{proof}
By definition, within each epoch the policy is fixed. Therefore,
\begin{align*}
    R_{k,k+1}^{(11)}=\sum_{t=t_k^a}^{t_{k+1}^a-1} \bigg(x_t^\top \tilde{P}_{i_k}x_t-x_{t+1}^\top \tilde{P}_{i_k}x_{t+1}\bigg)1_{\mathcal{E}_k^{\zeta}}=\bigg(x_{t_k^a}^\top \tilde{P}_{i_k}-x_{t_{k+1}^a}^\top \tilde{P}_{i_k} x_{t_{k+1}^a}\bigg)1_{\mathcal{E}_k^{\zeta}}.
\end{align*}
where
\begin{align*}
    \tilde{P}_{i_k}=\begin{cases}
P(\mathcal{C}_{n_{i_k}(t_k^a)}) &\quad \quad \quad \quad  k\in  \mathcal{I}_{\mathrm{cert}} \\
P_{K_0^{i_k}} &\quad \quad \quad \quad k\in  \mathcal{I}_{\mathrm{fallback}}
\end{cases}.
\end{align*}

Since $\tilde P_{i_k}$ is uniquely defined for each epoch, we obtain
\begin{align*}
\sum_{k\in\mathcal{I}_{\mathrm{cert}}}R_{k,k+1}^{(11)}+\sum_{k\in\mathcal{I}_{\mathrm{fallback}}}R_{k,k+1}^{(21)}\leq& x_{t_0^a}^\top  \tilde{P}_{i_0} x_{t_0^a}1_{\mathcal{E}_0^{\zeta}}+\sum_{k=1}^{n_s-1}x_{t^a_k}^\top\bigg( \tilde{P}_{i_k}- \tilde{P}_{i_{k-1}}\bigg) x_{t^a_k} 1_{\mathcal{E}_k^{\zeta}}\\
    = & x_{t_0^a}^\top  \tilde{P}_{i_0} x_{t_0^a}1_{\mathcal{E}_0^{\zeta}}+\sum_{k=1}^{n_s-1}x_{t^a_k}^\top\bigg( \tilde{P}_{i_k}-P(\Theta_*^{i_k})+P(\Theta_*^{i_{k-1}})- \tilde{P}_{i_{k-1}}\bigg) x_{t^a_k} 1_{\mathcal{E}_k^{\zeta}}\\
    &+\sum_{k=1}^{n_s-1}x_{t^a_k}^\top\bigg(P(\Theta_*^{i_k})-P(\Theta_*^{i_{k-1}})\bigg) x_{t^a_k} 1_{\mathcal{E}_k^{\zeta}}\\
    \leq & x_{t_0^a}^\top  \tilde{P}_{i_0}x_{t_0^a}1_{\mathcal{E}_0^{\zeta}}+\underbrace{\sum_{k=1}^{n_s-1}x_{t^a_k}^\top\bigg( \tilde{P}_{i_k}-P(\Theta_*^{i_k})+P(\Theta_*^{i_{k-1}})- \tilde{P}_{i_{k-1}}\bigg) x_{t^a_k} 1_{\mathcal{E}_k^{\zeta}}}_{\textit{Term1}}\\
    &+n_m\max_{1\leq k\leq n_s-1}x_{t^a_k}^\top P(\Theta_*^{i_k})x_{t^a_k} 1_{\mathcal{E}_k^{\zeta}}.
\end{align*}
The last term contributes only for malignant switches. Hence, assuming that there are at most
$n_m$ malignant switches in $\mathcal I$, the above bound follows.

We next upper bound \textit{Term 1} for the four possible combinations of controller types across two consecutive epochs: (1) \((L,L)\), (2) \((L,F)\), (3) \((F,L)\), and (4) \((F,F)\), where \(L\) and \(S\) denote the learning-based and fallback controllers, respectively.

\paragraph{(1) Case (L,L):}

In this case, since $P(\mathcal{C}_{n_{i_k}(t_k^a)})-P(\Theta_*^{i_k})\preceq 0$ holds by Lemma~\ref{lem:tightnessSol}, \textit{Term 1} can be upper bounded as

\begin{align*}
   \textit{Term 1}\leq \sum_{k=1}^{n_s-1}x_{t^a_k}^\top\bigg(P(\Theta_*^{i_{k-1}})-P(\mathcal{C}_{n_{i_{k-1}}(t_{k-1}^a)})\bigg) x_{t^a_k} 1_{\mathcal{E}_k^{\zeta}}&=\sum_{k=1}^{n_s-1}x_{t^a_k}^\top\bigg(P(\Theta_*^{i_{k-1}})-P(\mathcal{C}_{n_{i_{k-1}}(t_{k-1}^a)})\bigg) x_{t^a_k} 1_{\mathcal{E}_k^{\zeta}}\\
    &\leq X_*^2 \sum_{k=1}^{n_s-1} \chi_{t_k^a}^{i_k}\, 1_{\mathcal{E}_k^{\zeta}}.
\end{align*}

\paragraph{(2) Case (L,F):}

For this case we can write
\begin{align}
   \nonumber \textit{Term 1} =& \sum_{k=1}^{n_s-1}x_{t^a_k}^\top\bigg(P(\Theta_*^{i_{k-1}})-P(\mathcal{C}_{n_{i_{k-1}}(t_{k-1}^a)})\bigg) x_{t^a_k} 1_{\mathcal{E}_k^{\zeta}}+\sum_{k=1}^{n_s-1}x_{t^a_k}^\top\bigg(P_{K_0^{i_k}}-P(\Theta_*^{i_k})\bigg) x_{t^a_k} 1_{\mathcal{E}_k^{\zeta}}\\
    \leq& X_*^2 \sum_{k=1}^{n_s-1} \chi_{t_k^a}^{i_k}\, 1_{\mathcal{E}_k^{\zeta}}+ \frac{\nu_*}{\sigma^2_{\omega}}  X_*^2|\mathcal{M}| \log \big(\log \frac{2d_x \bar{\tau}_*n_s}{\delta}\big)^{\frac{2}{1-2\zeta}} \label{eq:argmnt_use}
\end{align}
where the last inequality follows from Lemma~\ref{lem:tightnessSol} and the bound on the total number of fallback steps established in Lemma~\ref{lem:fallback_steps}.

\paragraph{(3) Case (F,L):}

For this case we can write

\begin{align*}
    \textit{Term 1} =& \sum_{k=1}^{n_s-1}x_{t^a_k}^\top\bigg(P(\mathcal{C}_{n_{i_{k}}(t_{k-1}^a)})-P(\Theta_*^{i_{k}})\bigg) x_{t^a_k} 1_{\mathcal{E}_k^{\zeta}}+\sum_{k=1}^{n_s-1}x_{t^a_k}^\top\bigg(P(\Theta_*^{i_{k-1}})-P_{K_0^{i_{k-1}}}\bigg) x_{t^a_k} 1_{\mathcal{E}_k^{\zeta}} \leq 0
\end{align*}
as $P(\mathcal{C}_{n_{i_{k}}(t_{k-1}^a)})-P(\Theta_*^{i_{k}})\preceq 0$ and $P(\Theta_*^{i_{k-1}})\preceq P_{K_0^{i_{k-1}}}$ by definition.

\paragraph{(4) Case (F,F):}

For this case we can write

\begin{align*}
    \textit{Term 1} =& \sum_{k=1}^{n_s-1}x_{t^a_k}^\top\bigg(P_{K_0^{i_k}}-P(\Theta_*^{i_{k}})\bigg) x_{t^a_k} 1_{\mathcal{E}_k^{\zeta}}+\sum_{k=1}^{n_s-1}x_{t^a_k}^\top\bigg(P(\Theta_*^{i_{k-1}})-P_{K_0^{i_{k-1}}}\bigg) x_{t^a_k} 1_{\mathcal{E}_k^{\zeta}} \\
    \leq &\sum_{k=1}^{n_s-1}x_{t^a_k}^\top\bigg(P_{K_0^{i_k}}-P(\Theta_*^{i_{k}})\bigg) x_{t^a_k} 1_{\mathcal{E}_k^{\zeta}}\leq  \frac{\nu_*}{\sigma^2_{\omega}}|\mathcal{I}| X_*^2 \log \big(\log \frac{2d_x \bar{\tau}_*n_s}{\delta}\big)^{\frac{2}{1-2\zeta}}
\end{align*}

Combining the bounds for all four scenarios yields

\begin{align*}
    \textit{Term 1} \leq 2\frac{\nu_*}{\sigma^2_{\omega}}|\mathcal{M}| X_*^2 \log \big(\log \frac{2d_x \bar{\tau}_*n_s}{\delta}\big)^{\frac{2}{1-2\zeta}}+ 2 X_*^2 \sum_{k=1}^{n_s-1} \chi_{t_k^a}^{i_k}\, 1_{\mathcal{E}_k^{\zeta}}.
\end{align*}

Next, to upper-bound the second term, we collect the contributions mode-wise. Let $\bar T_i$ denote the set of time indices at which the system starts operating in mode $i$.
\begin{align*}
    \sum_{k=1}^{n_s-1} \chi_{t_k^a}^{i_k}\, 1_{\mathcal{E}_k^{\zeta}}\leq  \sum_{i\in \mathcal{M}} \sum_{\bar{t}_i\in\bar{T}_i}\chi^i_{\bar{t}_i}.
\end{align*}

On the event $\mathcal{E}_k^\zeta$, Lemma~\ref{lem:useflboun} implies that
\begin{align*}
   \chi^{i_k}_{t_k^a}&\leq  \underbrace{\frac{16\bar{D}_{i_k}\nu_{i_k}^3(1+\kappa_{i_k}^2)}{\alpha_0^{{i_k}^2}\sigma_{\omega}^6}}_{=:\beta_{i_k}}(n_{i_k}(t^a_{i})+\bar{c}_{i_k})^{-\frac{1}{2}+\zeta}\log \frac{t_k^a}{\delta}.
\end{align*}

Therefore,
\begin{align}
 \sum_{i\in \mathcal{M}} \sum_{\bar{t}_i\in\bar{T}_i}\chi^i_{\bar{t}_i}\leq   \sum_{i\in \mathcal{M}}\beta_i\sum_{\bar{t}_i\in\bar{T}_i} (n_i(\bar{t}_i)+\bar{c})^{-\frac{1}{2}+\zeta}\leq \beta_* \log \big(\frac{t_{n_s-1}^a}{\delta}\big)\,\sum_{i\in \mathcal{M}} (n_i(t_{n_s-1}^a)+\bar{c}_i)^{\frac{1}{2}+\zeta} \label{eq:komakstg}
\end{align}
where $\beta_*:=\max_i \beta_i$ and the last inequality follows since, for every mode $i$, the last time at which the system switches to mode $i$ is no later than $t_{n_s-1}^a$.

Applying Hölder's inequality yields
\begin{align}
    \sum_{i\in \mathcal{M}}(n_i(t_{n_s-1}^a)+\bar{c}_i)^{\frac{1}{2}+\zeta}\leq  \mathcal{M}^{\frac{1}{2}-\zeta}\big(\sum_{i\in \mathcal{M}} (n_i(t_{n_s-1}^a)+\bar{c}_i)\,\big)^{\frac{1} {2}+\zeta}\label{eq:helposf}
\end{align}

and since
\begin{align}
    \big(\sum_{i\in \mathcal{M}} (n_i(t_{n_s-1}^a)+\bar{c}_i)\,\big)^{\frac{1}{2}+\zeta}\leq (t_{n_s-1}^a+|\mathcal{M}|\bar{c}_*\big)^{\frac{1}{2}+\zeta} \label{eq:komak0191}
\end{align}
where $c_*=\max_{i} c_i$ we obtain
\begin{align}
   \sum_{i\in \mathcal{M}} \sum_{\bar{t}_i\in\bar{T}_i}\chi^i_{\bar{t}_i}\leq \beta_* \sum_{i\in \mathcal{M}} (n_i(t_{n_s-1}^a)+\bar{c}_i)^{\frac{1}{2}+\zeta}\leq \beta_* \log \big(\frac{t_{n_s-1}^a}{\delta}\big) \mathcal{M}^{\frac{1}{2}-\zeta}\big(t_{n_s-1}^a+|\mathcal{M}|\bar{c}_*)^{\frac{1}{2}+\zeta}.\label{eq:chibar}
\end{align}

Hence,
\begin{align*}   \sum_{k\in\mathcal{I}_{\mathrm{cert}}}R_{k,k+1}^{(11)}+\sum_{k\in\mathcal{I}_{\mathrm{fallback}}}R_{k,k+1}^{(21)}\leq& (1+n_m) \frac{\nu_*}{\sigma_{\omega}^2} X_*^22\frac{\nu_*}{\sigma^2_{\omega}}|\mathcal{M}| X_*^2 \log \big(\log \frac{2d_x \bar{\tau}_*n_s}{\delta}\big)^{\frac{2}{1-2\zeta}}\\
&+2X_*^2 \beta_* \log \big(\frac{t_{n_s-1}^a}{\delta}\big) \mathcal{M}^{\frac{1}{2}-\zeta}\big(t_{n_s-1}^a+|\mathcal{M}|\bar{c}_*)^{\frac{1}{2}+\zeta}.
\end{align*}
Finally, using $t_{n_s-1}^a\leq n_s \bar{\tau}_*$ where $\bar{\tau}_*$ is the upper bound on the epoch length from Lemma~\ref{thm:minimum_average_dwell}, we conclude that
\begin{align*}  \mathbb{E}\big[\sum_{k\in\mathcal{I}_{\mathrm{cert}}}R_{k,k+1}^{(11)}+\sum_{k\in\mathcal{I}_{\mathrm{fallback}}}R_{k,k+1}^{(21)}\big]\lesssim \mathcal{O}\big(n_m\log \frac{n_s}{\delta}+ |\mathcal{M}|\big(\log \frac{n_s}{\delta}\big)^{\frac{3-2\zeta}{1-2\zeta}}+|\mathcal{M}|^{1-\gamma}n_s^{\gamma} \, \log^2 \frac{n_s}{\delta}\big)
\end{align*}
where $\gamma=\frac{1}{2}+\zeta$.
\end{proof}

\begin{lemma}
    The following statement holds:
    \begin{align*}
           \mathbb{E}\Big[\sum_{k\in\mathcal{I}_{\mathrm{cert}}}R^{(14)}_{k,k+1}\big]\lesssim\mathcal{O}\big(|\mathcal{M}|^{1-\gamma}n_s^{\gamma}(\log \frac{n_s}{\delta})^2\big). 
    \end{align*}
\end{lemma}

\begin{proof}
We begin by writing
\begin{align*}
\sum_{k\in\mathcal{I}_{\mathrm{cert}}}R^{(14)}_{k,k+1}&=\sum_{k\in\mathcal{I}_{\mathrm{cert}}}\sum_{t=t_k^a}^{t_{k+1}^a-1}2\|P(\mathcal{C}_{n_{i_k}(t_k^a)})\|\mu_{n_{i_k}(t_k^a)}\big(\bar{z}_t^\top {V}^{-1}_{n_{i_k}(t_k^a)}\bar{z}_t\big)1_{\mathcal{E}^{\zeta}_k}\\
  &\leq 4\kappa_*^2 \frac{\nu_*}{\sigma_{\omega}^2}X_*^2\sum_{k\in\mathcal{I}_{\mathrm{cert}}} \tau_{k,k+1}^a \mu_{n_{i_k}(t_k^a)}\| {V}^{-1}_{n_{i_k}(t_k^a)}\|1_{\mathcal{E}^{\zeta}_k}\\
  &\leq 4\kappa_*^2 \frac{\nu_*}{\sigma_{\omega}^2}X_*^2\, \bar{\tau}_*\sum_{k\in\mathcal{I}_{\mathrm{cert}}} \mu_{n_{i_k}(t_k^a)}\| {V}^{-1}_{n_{i_k}(t_k^a)}\|1_{\mathcal{E}^{\zeta}_k}.
\end{align*}
Let $\bar T_i$ denote the set of switching times at which the system starts operating in mode $i$. Then,
\begin{align*}
\sum_{k\in\mathcal{I}_{\mathrm{cert}}} \mu_{n_{i_k}(t_k^a)}\| {V}^{-1}_{n_{i_k}(t_k^a)}\|1_{\mathcal{E}^{\zeta}_k}
  \leq  \sum_{i\in \mathcal{M}}\sum_{\bar{t}_i\in \bar{T}_i} \mu_{n_{i}(\bar{t}_i)}\|{V}^{-1}_{n_{i}(\bar{t}_i)}\|\leq \bar{D}_*\sum_{i\in \mathcal{M}}\sum_{\bar{t}_i\in \bar{T}_i} (n_i(\bar{t}_i)+\bar{c})^{-\frac{1}{2}+\zeta} \log \frac{t_{n_s-1}^a}{\delta}
\end{align*}
where the second inequality follows from Lemma~\ref{lem:useflboun}.

Proceeding as in the proof of Lemma~\ref{lem:R^{11}}, we convert the summation into a mode-wise summation. Applying~(\ref{eq:komakstg}), we obtain
\begin{align*}
    \sum_{i\in \mathcal{M}}\sum_{\bar{t}_i\in \bar{T}_i} (n_i(\bar{t}_i)+\bar{c})^{-\frac{1}{2}+\zeta} \leq \bar{D}_* (\log \frac{t_{n_s-1}^a}{\delta})\,\sum_{i\in \mathcal{M}} (n_i(t_{n_s-1}^a)+\bar{c}_i)\,\big)^{\frac{1}{2}+\zeta}.
\end{align*}

Next, by~(\ref{eq:komakstg}) and (\ref{eq:helposf}),
\begin{align*}
    \sum_{i\in \mathcal{M}} (n_i(t_{n_s-1}^a)+\bar{c}_i)\,\big)^{\frac{1}{2}+\zeta}\leq |\mathcal{M}|^{\frac{1}{2}-\zeta}(t_{n_s-1}^a+|\mathcal{M}|\bar{c}_*\big)^{\frac{1}{2}+\zeta}.
\end{align*}
Combining the above inequalities yields
\begin{align*}
   \sum_{k\in\mathcal{I}_{\mathrm{cert}}}R^{(14)}_{k,k+1}\leq   4\kappa_*^2 \bar{D}_*\frac{\nu_*}{\sigma_{\omega}^2}X_*^2\bar{\tau}_* \log \frac{t_{n_s-1}^a}{\delta} |\mathcal{M}|^{\frac{1}{2}-\zeta} (t_{n_s-1}^a+|\mathcal{M}|\bar{c}_*)^{\frac{1}{2}+\zeta}. 
\end{align*}
Finally, since $t_{n_s-1}^a\leq \bar{\tau}_* n_s$, taking expectations gives
\begin{align*}
    \mathbb{E}\Big[\sum_{k\in\mathcal{I}_{\mathrm{cert}}}R^{(14)}_{k,k+1}\big]\lesssim\mathcal{O}\big(|\mathcal{M}|^{1-\gamma}n_s^{\gamma}(\log \frac{n_s}{\delta})^2\big). 
\end{align*}
\end{proof}

\begin{lemma}
   The following statement holds
    \begin{align*}
           \mathbb{E}\Big[\sum_{k\in\mathcal{I}_{\mathrm{cert}}} R^{(2)}_{k,k+1}\big]\lesssim \mathcal{O}\big(|\mathcal{M}|^{1-\gamma}n_s^\gamma \log \frac{n_s}{\delta}\big)
    +
    \mathcal{O}\big(|\mathcal{M}|(\log \frac{n_s}{\delta})^{\frac{2}{1-2\zeta}}\big).
    \end{align*}
\end{lemma}
\begin{proof}

To prove the claim, we consider two possible cases depending on whether the controller applied after switching to mode $i_{k+1}$ is learning-based or the fallback controller, namely, the $(L,L)$ and $(L,F)$ scenarios. Accordingly, we partition the set of switching epochs as
\begin{align*}
    \mathcal{K}_{\mathrm{L}}
    &:=\{k \in \mathcal{I}_{\mathrm{cert}} :\text{the controllers applied in modes }i_{k+1}
    \text{ is learning-based}\},\\
    \mathcal{K}_{\mathrm{F}}
    &:=\{k \in \mathcal{I}_{\mathrm{cert}}:\text{the controller applied in mode }i_{k+1}
    \text{ is the fallback controller}\}.
\end{align*}
Then,
\begin{align*}
   \sum_{k\in\mathcal{I}_{\mathrm{cert}}} R^{(2)}_{k,k+1}
    =
    \sum_{k\in\mathcal{K}_{\mathrm{L}}}R^{(2)}_{k,k+1}
    +
    \sum_{k\in\mathcal{K}_{\mathrm{F}}}R^{(2)}_{k,k+1}.
\end{align*}

If $k\in  \mathcal{K}_{\mathrm{L}}$ then
    \begin{align*}
      R^{(2)}_{k,k+1}=&\sum_{t=t_k^a+\tau_{k,k+1}^*}^{t_{k+1}^a-1}J_*^{i_k}\leq J_*^{i_k} (t_{k+1}^a-1-t_{k}^a-\tau_{k,k+1}^*+1)=J_*^{i_k} (\tau^a_{k,k+1}-\tau_{k,k+1}^*)\\
      \leq &J_*^{i_k}\big(\bar{C}_0\chi_{t_k^a}^{i_{k+1}}+\bar{C}_1\chi^{i_k}_{t_k^a}+\bar{C}_2\big(\chi^{i_k}_{t_k^a}\big)^2\big)
    \end{align*}
    where the last inequality follows from Theorem~\ref{Thm:dwellTimeError}.

Summing over all switching epochs yields
\begin{align}
     \nonumber \sum_{k\in\mathcal{K}_L} R^{(2)}_{k,k+1}&\leq \sum_{k=0}^{n_s-1} J_*^{i_k}\big(\bar{C}_0\chi_{t_k^a}^{i_{k+1}}+\bar{C}_1\chi^{i_k}_{t_k^a}+\bar{C}_2\big(\chi^{i_k}_{t_k^a}\big)^2\big)\\
    \nonumber   &\leq  \bar{J}_*\big(\bar{C}_1\chi^{i_0}_{t_0^a}+\bar{C}_2(\chi^{i_0}_{t_0^a})^2+\bar{C}_0\chi_{t_{n_s-1}^a}^{i_{n_s}}\big) +\bar{J}_*\sum_{k=1}^{n_s-1} \big(\bar{C}_0\chi_{t_k^a}^{i_{k}}+\bar{C}_1\chi^{i_k}_{t_k^a}+\bar{C}_2\big(\chi^{i_k}_{t_k^a}\big)^2\big) \\    &\leq\bar{J}_*\big(\bar{C}_1\chi^{i_0}_{t_0^a}+\bar{C}_2(\chi^{i_0}_{t_0^a})^2+\bar{C}_0\chi_{t_{n_s-1}^a}^{i_{n_s}}\big)+ \bar{J}_* \sum_{i\in \mathcal{M}}\sum_{\bar{t}_i\in \bar{T}_i}\big((\bar{C}_0+\bar{C}_1)\chi_{\bar{t}_i}^{i}++\bar{C}_2\big(\chi^{i}_{\bar{t}_i}\big)^2\big)\label{eq:myref1}
\end{align}
 where $\bar{J}_*=\max_i J_*^i$. 

 By the definition of $\chi_t^i$ and (\ref{eq:juststab02}) in the proof of Theorem~\ref{Stability_thm17}, we have  $\chi^{i_k}_{t_k^a}\leq \alpha_0^*/4$ where $\alpha_0^*=\max_i \alpha_0^i/4$.
 
Therefore,
 \begin{align*}   \bar{J}_*\big(\bar{C}_1\chi^{i_0}_{t_0^a}+\bar{C}_2(\chi^{i_0}_{t_0^a})^2+\bar{C}_0\chi_{t_{n_s-1}^a}^{i_{n_s}}\big)\leq \bar{J}_*\bigg(\bar{C}_0+\bar{C}_1+\frac{\alpha^*_0}{4}\bar{C}_2\bigg) \frac{{\alpha_0^*}}{4}
 \end{align*}
which is a constant independent of $n_s$.

It remains to bound the second term in~(\ref{eq:myref1}). By Lemma~\ref{lem:useflboun}, $\sum_{\bar{t}_i\in \bar{T}_i}\big(\chi^{i}_{\bar{t}_i}\big)^2$ is of order $(n_i(t_{n_s-1}^a)+\bar{c})^{2\zeta}$, whereas $\sum_{\bar{t}_i\in \bar{T}_i} \chi^{i}_{\bar{t}_i}$, is of order $(n_i(t_{n_s-1}^a)+\bar{c})^{1/2+\zeta}$ and therefore dominates the former. Hence, we can write
     \begin{align*}
       \sum_{i\in \mathcal{M}}\sum_{\bar{t}_i\in \bar{T}_i}\big((\bar{C}_0+\bar{C}_1)\chi_{\bar{t}_i}^{i}++\bar{C}_2\big(\chi^{i}_{\bar{t}_i}\big)^2\big)&\leq \sum_{i\in \mathcal{M}}\sum_{\bar{t}_i\in \bar{T}_i}(\bar{C}_0+\bar{C}_1+\bar{C}_2)\chi_{\bar{t}_i}^{i}\\
       &\leq (\bar{C}_0+\bar{C}_1+\bar{C}_2) \beta_*\mathcal{M}^{\frac{1}{2}-\zeta}\big(t_{n_s-1}^a+|\mathcal{M}|\bar{c}_*)^{\frac{1}{2}+\zeta}\log \frac{t_{n_s-1}^a}{\delta}
\end{align*}
  where the last inequality follows from(\ref{eq:chibar}).   

For $k\in\mathcal{K}_{\mathrm{F}}$, Corollary~\ref{cor:LSBound} gives
     \begin{align*}
         R^{(2)}_{k,k+1}\leq & J_*^{i_k}\big(\tilde{C}_0\Big(\ln \bar{\lambda}\big(P_{K_0^{i_{k+1}}}\big)
        -\ln\bar{\lambda}\big(P(\Theta_*^{i_{k+1}})\big)\Big) +\tilde{C}_1\chi^{i_k}_{t_k^a}
        +\tilde{C}_2\big(\chi^{i_k}_{t_k^a}\big)^2 \big).
    \end{align*}
    In this scenario, the number of such switching epochs is logarithmically bounded. This follows directly from Lemma~\ref{lem:fallback_steps}, which shows that, for each mode, the total number of time steps during which the system operates with the fallback controller is logarithmically bounded. Consequently, the number of switches associated with such scenarios is also logarithmically bounded. Therefore,
\begin{align*}
     \sum_{k\in\mathcal{K}_{\mathrm{F}}}R^{(2)}_{k,k+1} \leq \bar{J}_*\bigg(\tilde{C}_0 \ln \big(\frac{(1+\kappa_0^{*^2})\kappa_0^{*^2}\alpha_1^*}{\gamma_0^*}\big)+\tilde{C}_1 \frac{\alpha^*_0}{4}+(\frac{\alpha^*_0}{4})^2\tilde{C}_2\bigg)|\mathcal{M}|\log \big(\log \frac{2d_x \bar{\tau}_*n_s}{\delta}\big)^{\frac{2}{1-2\zeta}}
\end{align*}
where we applied Lemma \ref{lem:komaki} to upper-bound $\bar{\lambda}(P_{K_0^{i_{k+1}}})$.

Combining the bounds obtained for the two cases gives
\begin{align*}
    \mathbb{E}\left[
    \sum_{k=0}^{n_s-1}R^{(2)}_{k,k+1}
    \right]
    \lesssim
    \mathcal{O}\big(|\mathcal{M}|^{1-\gamma}n_s^\gamma \log \frac{n_s}{\delta}\big)
    +
    \mathcal{O}\big(|\mathcal{M}|(\log \frac{n_s}{\delta})^{\frac{2}{1-2\zeta}}\big)
\end{align*}
where we used the bound $t_{n_s-1}^a\leq \bar{\tau}_* n_s$ and took expectations to complete the proof.
\end{proof}

\begin{lemma}
    The following statement holds:

\begin{align*}
  \mathbb{E}\Big[ \sum_{k\in\mathcal{I}_{\mathrm{fallback}}} R_{k,k+1}^{(23)}\big]\lesssim  |\mathcal{M}|\log\big(\frac{n_s}{\delta}\big)^{\frac{2}{1-2\zeta}}.
\end{align*}
\end{lemma}
\begin{proof}
The proof follows from the fact that the total number of time steps during which the fallback controller is applied is logarithmically bounded, as established in Lemma~\ref{lem:fallback_steps}.
\end{proof}

\begin{lemma}
The following statement holds:
    \begin{align*}
        \mathbb{E}\!\left[
\sum_{k=0}^{n_s-1}
R^{(4)}_{k,k+1}
\right]\lesssim \mathcal{O}\big(|\mathcal{M}|^{\zeta}\,t^{1-\zeta}\big)
    \end{align*}
\end{lemma}

\begin{proof}
Since $\eta_t^i\sim\mathcal{N}(0,\Gamma_t^i)$, we have
\begin{align*}
\mathbb{E}\Big[\sum_{k=0}^{n_s-1} R^{(4)}_{k,k+1}\Big]
&=
\mathbb{E}\Big[\sum_{k=0}^{n_s-1}
\Big(\sum_{j=t_k^a}^{t_{k+1}^a}
{\eta_j^{i_k}}^\top R^{i_k}\eta_j^{i_k}\Big)
1_{\mathcal{E}^{\zeta}_k}\Big]\\
&\leq \alpha_1^*\mathbb{E}\Big[\sum_{i\in \mathcal{M}}
\sum_{k=1}^{n_i(t_{n_s-1}^a)}
{\eta_j^{i}}^\top\eta_j^{i}
\Big]\\
&= \alpha_1^*\sum_{i\in \mathcal{M}}
\mathbb{E}\!\left[
\mathbb{E}\!\left[
\sum_{k=1}^{n_i(t_{n_s-1}^a)}
{\eta_k^{i}}^\top \eta_k^{i}
\,\middle|\,
n_i(t_{n_s-1}^a)
\right]
\right]
\end{align*}
where the inequality follows from \(R^i\preceq\alpha_1^*I\).

Next, consider the definition of $\Gamma_k^i$, we have
\begin{align*}
\operatorname{tr}(\Gamma_k^i)
\lesssim \frac{1}{k^\zeta}.
\end{align*}
Therefore,
\begin{align*}
\mathbb{E}\!\left[
\sum_{k=1}^{n_i(t_{n_s-1}^a)}
{\eta_k^{i}}^\top \eta_k^{i}
\,\middle|\,
n_i(t_{n_s-1}^a)
\right]=\sum_{k=1}^{n_i(t_{n_s-1}^a)} \operatorname{tr}(\Gamma_k^{i})\lesssim (n_i(t_{n_s-1}^a))^{1-\zeta}.
\end{align*}

Consequently,
\begin{align*}
    \mathbb{E}\Big[\sum_{k=0}^{n_s-1} R^{(4)}_{k,k+1} 1_{\mathcal{E}^{\zeta}_k}\Big] \lesssim \sum_{i\in \mathcal{M}} \mathbb{E}\big[(n_i(t_{n_s-1}^a))^{1-\zeta}\big].
\end{align*}
For $\zeta\in(0,1/2)$, Hölder's inequality gives
\begin{align*}
    \sum_{i\in\mathcal{M}} (n_i(t_{n_s-1}^a))^{1-\zeta}
\leq
|\mathcal{M}|^{\zeta}
\left(\sum_{i\in\mathcal{M}} n_i(t_{n_s-1}^a)\right)^{1-\zeta}\lesssim |\mathcal{M}|^{\zeta} n_s^{1-\zeta} .
\end{align*}
where we used
\begin{align*}
    \sum_{i\in\mathcal{M}} n_i(t_{n_s-1}^a)\leq t_{n_s-1}^a\leq n_s\bar{\tau}_*.
\end{align*}
Combining the above inequalities yields
\begin{align*}
    \mathbb{E}\Big[\sum_{k=0}^{n_s-1} R^{(4)}_{k,k+1} 1_{\mathcal{E}^{\zeta}_k}\Big] \lesssim |\mathcal{M}|^{\zeta} n_s^{1-\zeta}.
\end{align*}
\end{proof}

\begin{proof}[{\bf Proof of Theorem \ref{thm:RegretBound}}]
Combining the bounds established in the preceding results, we obtain
\begin{align*}
    \mathbb{E}\big[\tilde{\mathcal{R}}_{\zeta}(\mathcal{I})\big]
    \lesssim
    \max \bigg\{
    &
    \mathcal{O}\big(
        |\mathcal{M}|^{\zeta} n_s^{1-\zeta}
    \big),
    \nonumber\\
    &
    \mathcal{O}\bigg(
        n_m\log\frac{n_s}{\delta}
        +
        |\mathcal{M}|
        \left(
        \log\frac{n_s}{\delta}
        \right)^{\frac{3-2\zeta}{1-2\zeta}}
        +
        |\mathcal{M}|^{1-\gamma}
        n_s^{\gamma}
        \log^2\frac{n_s}{\delta}
    \bigg)
    \bigg\},
\end{align*}
where $\gamma=\frac12+\zeta$.

To balance the dominant $n_s$-dependent terms in the above regret bound, we
choose $\zeta=\frac14$. Denoting the resulting regret by
$\mathcal{R}_*(\mathcal{I})$, the above bound becomes
\begin{align*}
    \mathbb{E}\big[\tilde{\mathcal{R}}_{*}(\mathcal{I})\big]
    \lesssim
    \mathcal{O}\left(
        n_m\log\frac{n_s}{\delta}
        +
        |\mathcal{M}|
        \log^5\frac{n_s}{\delta}
        +
        |\mathcal{M}|^{1/4}
        n_s^{3/4}
        \log^2\frac{n_s}{\delta}
    \right).
\end{align*}

By Lemma~\ref{lem:epseventprob}, the good event
$\mathcal{E}_{n_s-1}^{\zeta}$ holds with probability at least
$1-6|\mathcal{M}|\delta$. Therefore, the above bound on the truncated regret
implies the same upper bound on the original regret
$\mathbb{E}[\mathcal{R}_{*}(\mathcal{I})]$ with probability at least
$1-6|\mathcal{M}|\delta$, completing the proof.
\end{proof}

\end{document}